\documentclass{CSML}

\def\dOi{11(3:22)2015}
\lmcsheading%
{\dOi}
{1--39}
{}
{}
{Aug.~28, 2014}
{Sep.~25, 2015}
{}

\ACMCCS{
  [{\bf  Mathematics of computing}]:
     Discrete mathematics---Combinatorics---Enumeration; 
    Discrete mathematics---Graph theory---Graphs and surfaces;
  [{\bf Theory of computation}]:  Models of computation---Computability---Lambda calculus;
    Logic---Linear logic}

\usepackage{cite}
\usepackage{nameref,hyperref,cleveref}
\usepackage[utf8]{inputenc} 

\usepackage{amsmath}
\usepackage{amssymb}
\usepackage{amsthm}
\usepackage{stmaryrd}
\usepackage{pxfonts}
\usepackage{proof}
\usepackage{graphicx}
\usepackage[all]{xy}

\usepackage{mathtools}
\usepackage{thmtools}

\usepackage{multicol}
\usepackage{diagbox}

\usepackage{caption}
\usepackage{subcaption}

\usepackage[normalem]{ulem}

\makeatletter
\def\@endtheorem{\endtrivlist}
\makeatother

\newcommand*{\imgcenter}[1]{\begingroup\setbox0=\hbox{#1}\parbox{\wd0}{\box0}\endgroup}

\newcommand\symiso{\gamma}
\newcommand\turniso{\sigma}

\newif\ifcolorversion
\colorversiontrue

\ifcolorversion
\def\DIAGRAMS{diagrams}
\def\DIAGRAMSTWO{diagrams2}
\else
\def\DIAGRAMS{diagrams-nocolor}
\def\DIAGRAMSTWO{diagrams2-nocolor}
\fi

\declaretheorem[name=Theorem,numberwithin=section]{theorem}

\declaretheorem[sibling=theorem]{proposition}

\declaretheorem[sibling=theorem]{definition}

\def\cal#1{ {\mathcal #1} }

\newcommand\defeq{\stackrel{\text{\tiny def}}{=}}
\newcommand\id{\mathrm{id}}

\newcommand\plugL{leval}
\newcommand\plugR{reval}

\newcommand\lc[2][]{\lambda^{#1}[{#2}]}
\newcommand\rc[2][]{\rho^{#1}[{#2}]}

\newcommand\mul{\mathbin{\bullet}}

\newcommand\Lolli{\multimap}

\newcommand\G{\Gamma}
\newcommand\D{\Delta}

\newcommand\N{\mathbb{N}}

\newcommand\op{\mathrm{op}}

\newcommand\isoto{\overset{\sim}{\to}}

\newcommand\definand[1]{\emph{#1}}

\newcommand\sem[1]{\left\llbracket{#1}\right\rrbracket}

\newcommand\resL[2][\bot]{{#2} \mathbin{\setminus} {#1}}
\newcommand\resR[2][\bot]{{#1} \mathbin{/} {#2}}

\newcommand\plug[2]{\langle #1\mid #2\rangle}
\newcommand\focus[1]{\uline{#1}}

\newcommand\SLam{{\mathcal S}}
\newcommand\SNeu{{\mathcal S}_B}
\newcommand\SNF{{\mathcal S}_R}
\newcommand\wild{\textvisiblespace}
\newcommand\set[1]{\{\,#1\,\}}

\newcommand\rdual[1]{{#1}^*}
\newcommand\ldual[1]{{}^*#1}

\newcommand\enc[2][]{\sem{#2}_{#1}}

\newcommand\composIsthm{i_{c}}
\newcommand\decomposIsthm{i_{d}}
\newcommand\composNonIsthm[1]{n_{c}^{(#1)}}
\newcommand\decomposNonIsthm{n_{d}}

\newcommand\composFunOpen{f_{c}}

\newcommand\composValOpen[1]{v_{c}^{(#1)}}

\newcommand\fundecomp{\overset{\leadsto}{\imgcenter{\includegraphics[scale=0.4]{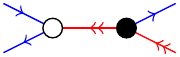}}}}
\newcommand\valdecomp{\overset{\leadsto}{\imgcenter{\includegraphics[scale=0.4]{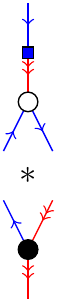}}}}

\newcommand\tmstk[2]{\left<#1\mid#2\right>}
\newcommand\evalsto{\Downarrow}

\newcommand\ONH[1]{O(#1)}

\begin{document}

\title[A correspondence between rooted planar maps and normal planar lambda terms]
        {A correspondence between rooted planar maps and \\ normal planar lambda terms}

\author[N.~Zeilberger]{Noam Zeilberger\rsuper a}
\address{{\lsuper a}MSR-Inria, Joint Centre, 91120 Palaiseau, France}
\email{noam.zeilberger@gmail.com}

\author[A.~Giorgetti]{Alain Giorgetti\rsuper b}
\address{{\lsuper b}FEMTO-ST institute (UMR CNRS 6174 - UBFC/UFC/ENSMM/UTBM), University
of Franche-Comt\'e, 25030 Besan\c{c}on, France / 
CASSIS project, Inria, 54600 Villers-les-Nancy, France}
\email{alain.giorgetti@femto-st.fr}


\begin{abstract}
A \emph{rooted planar map} is a connected graph embedded in the 2-sphere, with one edge marked and assigned an orientation.
A term of the pure lambda calculus is said to be \emph{linear} if every variable is used exactly once, \emph{normal} if it contains no $\beta$-redexes,  and \emph{planar} if it is linear and the use of variables moreover follows a deterministic stack discipline.
We begin by showing that the sequence counting normal planar lambda terms by a natural notion of size coincides with the sequence (originally computed by Tutte) counting rooted planar maps by number of edges.
Next, we explain how to apply the machinery of string diagrams to derive a graphical language for normal planar lambda terms, extracted from the semantics of linear lambda calculus in symmetric monoidal closed categories equipped with a \emph{linear reflexive object} or a \emph{linear reflexive pair}.
Finally, our main result is a size-preserving bijection between rooted planar maps and normal planar lambda terms, which we establish by explaining how \emph{Tutte~decomposition} of rooted planar maps (into vertex maps, maps with an isthmic root, and maps with a non-isthmic root) may be naturally replayed in linear lambda calculus, as certain surgeries on the string diagrams of normal planar lambda terms.
\end{abstract}

\maketitle

\section{Introduction: a curious correspondence}
\label{sec:intro}

The pure lambda calculus is a universal programming language based on only two primitive operations: for any pair of terms $t$ and $u$, there is a term
$t(u)$
representing the \emph{application} of $t$ to $u$, while for any pair of a term $t$ and a variable $x$, there is a term
$\lambda x.t$
representing the \emph{abstraction} of $t$ in $x$.
Terms are always considered up to renaming of abstracted variables, so that for example
$\lambda x.x$ and $\lambda y.y$
both represent the same term (intuitively standing for the identity function).
The main source of computation is the rule of \emph{$\beta$-reduction}:
$$(\lambda x.t)(u) \to t[u/x]$$
Here $t[u/x]$ denotes the \emph{substitution} of $u$ for $x$ in $t$ (which technically must be ``capture-avoiding'' in a sense we need not get into here \cite{barendregt1984}).
Equating terms modulo $\beta$-reduction yields a theory known as $\Lambda_\beta$, wherein, for example, we can derive
$$
(\lambda x.xx)(\lambda y.y) = (\lambda y.y)(\lambda y.y) = \lambda y.y
\qquad\text{and}\qquad
((\lambda x.\lambda y.x)w)(\lambda z.z) = (\lambda y.w)(\lambda z.z) = w.
$$
Although it is remarkable that such a simple and conceptual language is Turing-complete, the focus of this paper will be on a much more restrictive but still important subset of lambda calculus known as the \emph{linear} fragment, defined by the requirement that every abstracted variable must be used exactly once.
For example, all of the terms
$$\lambda x.\lambda y.yx
\qquad
\lambda x.x(\lambda y.y)
\qquad
\lambda x.\lambda y.xy
$$
are linear, but all of the terms
$$\lambda x.xx
\qquad
\lambda x.\lambda y.x
\qquad
\lambda x.\lambda y.y
$$
are non-linear.
As one example of the special properties of the linear fragment, note that the problem of computing the $\beta$-normal form of a linear lambda term is PTIME-complete \cite{mairson2004}.

Among the linear terms, it is possible to identify an even more restrictive subset of terms which are \emph{planar}, in the sense that (reading left-to-right) variables are used in the reverse order which they are abstracted.
Thus the two linear terms
$$\lambda x.\lambda y.yx
\qquad
\lambda x.x(\lambda y.y)
$$
are planar, but the linear term
$$
\lambda x.\lambda y.xy
$$
is non-planar.

Motivated by questions related to the study of \emph{type refinement} \cite{mz15popl}, the first author of this paper counted $\beta$-normal planar lambda terms along a natural notion of size, and obtained the following sequence of first numbers through a simple recurrence equation:
$$ 1, 2, 9, 54, 378, 2916, 24057 $$ 
Surprisingly, this sequence already existed in the Online Encyclopedia of Integer Sequences \cite{oeis}, corresponding to the first entries of a series which is indexed in the OEIS as {\bf A000168}.

It turns out that this series is well-known in combinatorics (see, e.g., \cite[VII.8.2]{flajolet-sedgewick}), and counts \emph{rooted planar maps} by number of edges.
A rooted planar map is essentially a connected graph drawn on the sphere with no crossing-edges, and with one edge marked and assigned an orientation.
Here are two small examples, where we have chosen to project the maps onto the page so that the ``infinite'' (outer) face is to the left of the oriented root edge:
\begin{center}
\imgcenter{\includegraphics[width=5cm]{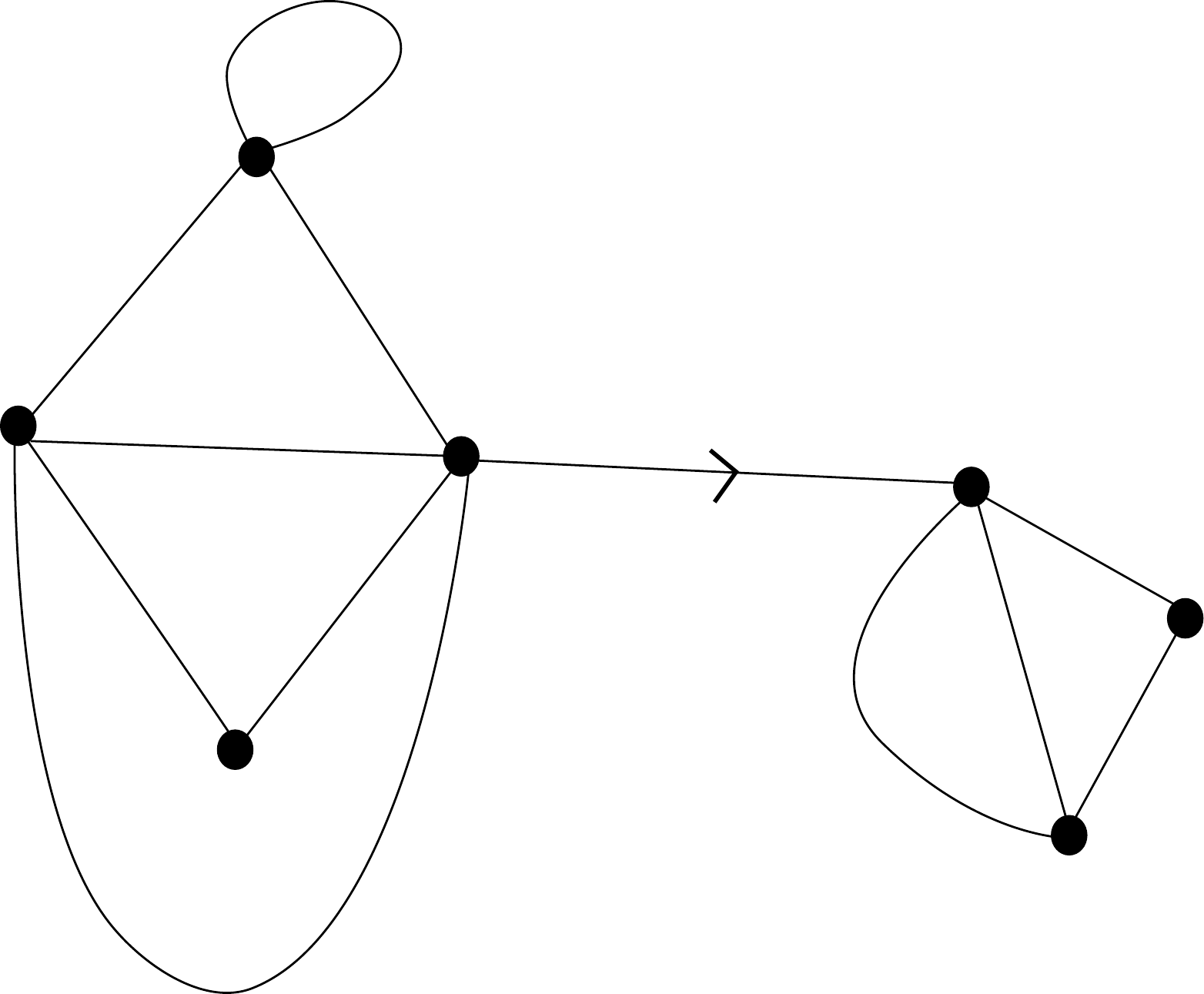}}
\qquad\qquad
\imgcenter{\includegraphics[width=3cm]{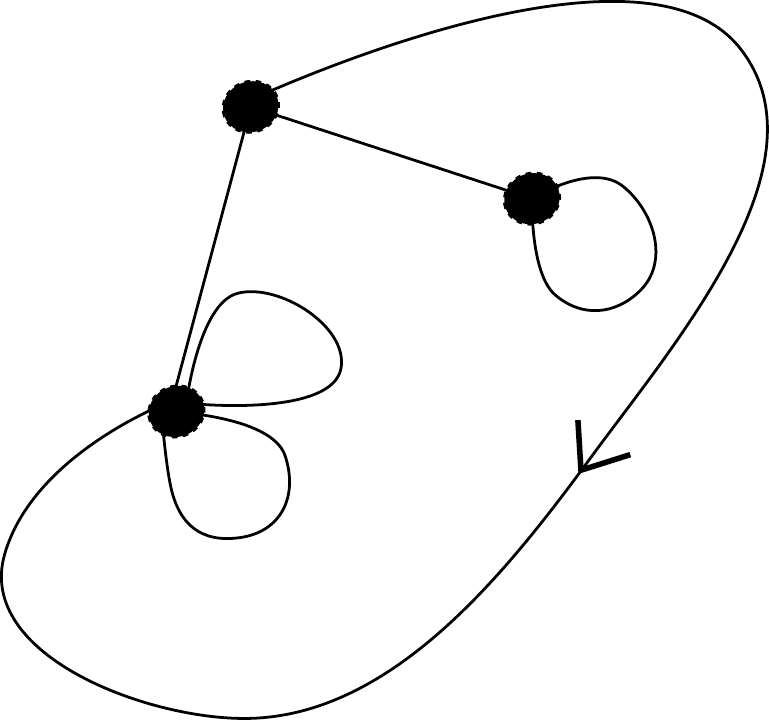}}
\end{center}

Rooted planar maps were originally enumerated in the 1960s by Tutte \cite{tutte1963,tutte1968} as part of an attack on the four-color theorem (which, of course, is about unrooted planar maps).
What makes rooted maps easier to count than unrooted maps is that the latter can have non-trivial symmetries but the former cannot.
Tutte was even able to derive a closed form for the total number of rooted planar maps with $n$ edges: $\frac{2\cdot (2n)!\cdot 3^n}{n!\cdot (n+2)!}$.

The main result of this paper is a size-preserving bijection between rooted planar maps and normal planar lambda terms.
We work towards this result as follows:
\begin{itemize}
\item
In \Cref{sec:genfun}, we introduce linear lambda calculus from a combinatorial perspective, defining linear lambda terms as certain decorations of ``lambda skeletons''.
We then give inductive definitions of various properties of linear lambda terms, and use these to derive functional equations for the generating functions counting normal (and ``neutral'') planar terms by size and number of free variables.
By solving these equations, we demonstrate in particular that the sequence counting closed normal planar lambda terms indeed coincides with A000168.
\item 
In \Cref{sec:diagrams}, we explain how to apply the machinery of string diagrams to derive a graphical language for normal planar lambda terms.
Our first step is a rational reconstruction of the well-known ``lambda-graphs'', as string diagrams extracted from the semantics of linear lambda calculus in symmetric monoidal closed categories equipped with a \emph{linear reflexive object}.
We then introduce the concept of a \emph{linear reflexive pair} as a refinement of linear reflexive object, and use this to extract a coloring protocol for the string diagrams representing normal linear lambda terms.
\item
Finally, in \Cref{sec:tutte} we give the size-preserving bijection between rooted planar maps and normal planar lambda terms.
The idea of the bijection is based on Tutte's analysis of rooted planar maps now known as \emph{Tutte decomposition} \cite{tutte1968}, which starts by establishing a trichotomy on rooted planar maps as either being the degenerate vertex map (with no edges), or else having an isthmic root, or else having a non-isthmic root.
After giving a review of Tutte decomposition, we explain how his analysis may be naturally replayed in linear lambda calculus as certain surgeries on the string diagrams of normal planar lambda terms, and show how to use this to obtain a size-preserving bijection.
\end{itemize}

\section{Lambda skeletons, planarity, neutral and normal terms}
\label{sec:genfun}

The main objects we study in this paper are lambda terms satisfying a combination of three properties: linearity, planarity, and absence of $\beta$-redexes.
Since this is just a small fragment of lambda calculus, we do not need to introduce the full machinery of classical lambda calculus (for which the reader can see \cite{barendregt1984}), and instead take an approach inspired by the operadic perspective advocated by Hyland \cite{hyland-lambda-calculus} since it makes the underlying combinatorial structure of linear/planar lambda terms more apparent.

We begin by defining ``skeletons'' of lambda terms, standing for lambda terms with placeholders for variable names.
\begin{definition}\label{defn:lskeleton}\emph{
A \definand{lambda skeleton} is an element of the set $\SLam(i)$, defined for $i\in \mathbb{N}$ as the least graded set satisfying the following rules:
$$
\infer[V]{\wild \in \SLam(1)}{}
\qquad
\infer[A]{p(q) \in \SLam(j+k)}{p \in \SLam(j) & q \in \SLam(k)}
\qquad
\infer[L]{\lambda \wild.p\in\SLam(i)}{p \in \SLam(i+1)}
$$
The \definand{degree} of a lambda skeleton $p \in \SLam(i)$ is the index $i$.}
\end{definition}
It is worth remarking that lambda skeletons are simply \emph{unary-binary trees} (also known as \emph{Motzkin trees}), where unary ``$L$-nodes'' stand for lambdas, binary ``$A$-nodes'' stand for applications, and ``($V$-)leaves'' stand for variables.
But whereas unary-binary trees are usually parameterized by total number of nodes (which gives rise to the \emph{Motzkin numbers}, as the number of unary-binary trees with a given number of nodes \cite[I.39]{flajolet-sedgewick}),  lambda skeletons are parameterized by the \emph{difference} between the number of leaves and the number of $L$-nodes.
In particular, all of the sets $\SLam(i)$ are infinite, so lambda skeletons cannot be counted directly along their degree.

Linear lambda terms will be defined as certain \emph{decorations} of lambda skeletons.
\newcommand\cterm[2]{[#1]#2}
\newcommand\LLin{\Lambda_1}
\newcommand\LPla{\Lambda_1^0}
\newcommand\islinear[3]{\cterm{#1}{#2} \in \LLin(#3)}
\newcommand\isplanar[3]{\cterm{#1}{#2} \in \LPla(#3)}
\begin{definition}\label{defn:linear}\emph{
A \definand{pseudo lambda term} is a lambda skeleton in which all occurrences of ``$\wild$'' have been replaced by variable names.
Given a lambda skeleton $p \in \SLam(i)$, a list of variable names $\Gamma = x_1,\dots,x_i$, and a pseudo lambda term $t$, we write $\islinear{\Gamma}{t}{p}$ to indicate that $t$ is a \definand{linear lambda term} (with free variables $\Gamma$) \definand{decorating} $p$, as defined by the following rules:
$$
\infer[V]{\islinear{x}{x}{\wild}}{}
\qquad
\infer[A]{\islinear{\Gamma,\Delta}{t(u)}{p(q)}}{\islinear{\Gamma}{t}{p} & \islinear{\Delta}{u}{q}}
\qquad
\infer[L]{\islinear{\Gamma}{\lambda x.t}{\lambda \wild.p}}{\islinear{x,\Gamma}{t}{p}}
$$
$$
\infer[T]{\islinear{\Gamma,x,y,\Delta}{t}{p}}{\islinear{\Gamma,y,x,\Delta}{t}{p}}
$$
We write $[\Gamma]t \in \LLin$ to indicate that $\islinear{\Gamma}{t}{p}$ for some $p$.
In general, linear lambda terms should always be considered with free variables indicated, though at times we will leave this implicit.
We say that two linear lambda terms $[\Gamma]t, [\Delta]u \in \LLin$ are \definand{$\alpha$-equivalent} if one can be obtained from the other by renaming of variables (in $\Gamma$ and $\Delta$, as well as the variables introduced by lambda abstraction within $t$ and $u$).
Linear lambda terms are always considered modulo $\alpha$-equivalence.}
\end{definition}
Intuitively, a linear lambda term is \emph{planar} if it is possible to show that it is linear without using the $T$(ransposition)-rule.
Explicitly, this is equivalent to the following definition.
\begin{definition}\label{defn:planar}\emph{
Let $[\Gamma]t \in \LLin$ be a linear lambda term.
We write $\isplanar{\Gamma}{t}{p}$ to indicate that $t$ is a \definand{planar lambda term} (with free variables $\Gamma$) \definand{decorating} the lambda skeleton $p$, as defined by the following rules:
$$
\infer[V]{\isplanar{x}{x}{\wild}}{}
\qquad
\infer[A]{\isplanar{\Gamma,\Delta}{t(u)}{p(q)}}{\isplanar{\Gamma}{t}{p} & \isplanar{\Delta}{u}{q}}
\qquad
\infer[L]{\isplanar{\Gamma}{\lambda x.t}{\lambda \wild.p}}{\isplanar{x,\Gamma}{t}{p}}
$$
Similarly we write $[\Gamma]t \in \LPla$ to indicate that $\isplanar{\Gamma}{t}{p}$ for some $p$.}
\end{definition}
One important observation about planar lambda terms is that they are entirely determined by their lambda skeleton, and conversely, that any lambda skeleton may be decorated by a (necessarily unique) planar lambda term.
\begin{proposition}
\label{prop:skelplam:unique}
Let $\isplanar{\Gamma}{t}{p}$ and $\isplanar{\Delta}{u}{p}$ be two planar lambda terms decorating the same skeleton $p$.
Then $[\Gamma]t$ and $[\Delta]u$ are $\alpha$-equivalent.
\end{proposition}
\begin{proof}
Immediate by induction on $p$.
\end{proof}
\begin{proposition}
\label{prop:skelplam:exists}
For any lambda skeleton $p \in \SLam(i)$, there is a planar lambda term $\isplanar{x_i,\dots,x_1}{p^\dagger}{p}$.
\end{proposition}
\begin{proof}
There is a simple algorithm for computing $p^\dagger$ recursively, by traversing the lambda skeleton $p$ while maintaining a current list of free variables as a stack (initialized to $\Gamma = x_i,\dots,x_1$, for some $i$ distinct variable names, with $x_i$ at the top):
\begin{itemize}
\item (Case $p = \wild$): pop a variable name from the top of the stack.
\item (Case $p = \lambda\wild.q$): generate a fresh variable name not in
$\Gamma$, push it onto the stack and continue traversing $q$.
\item (Case $p = q(r)$): traverse the application left-to-right, i.e., decorate $q$ (while consuming some variables from the stack) and then decorate $r$.
\end{itemize}
In \Cref{skeleton2planar:fig} we show this algorithm described in the style of operational semantics, as a relation between a pair of a lambda skeleton and an input stack and a pair of a planar lambda term and an output stack.
\begin{figure}[tb!] $$
\infer{\tmstk{\wild}{x,\Gamma} \evalsto \tmstk{x}{\Gamma}}{}
\quad
\infer{\tmstk{\lambda\wild.q}{\Gamma} \evalsto \tmstk{\lambda x.t}{\Gamma'}}
      {x\text{ fresh} & \tmstk{q}{x,\Gamma} \evalsto \tmstk{t}{\Gamma'}}
\quad
\infer{\tmstk{p_1 (p_2)}{\Gamma} \evalsto \tmstk{t_1(t_2)}{\Gamma''}}
      {\tmstk{p_1}{\Gamma} \evalsto \tmstk{t_1}{\Gamma'} &
       \tmstk{p_2}{\Gamma'} \evalsto \tmstk{t_2}{\Gamma''}}
$$
\caption{From lambda skeletons to planar lambda terms
\label{skeleton2planar:fig}}
\end{figure}
For example, the skeleton $\lambda\wild.\lambda\wild.\wild(\lambda \wild.\wild\wild) \in \SLam(0)$ can be decorated by the (closed) planar lambda term $\lambda x.\lambda y.y(\lambda z.zx)$, which is the unique planar decoration of this skeleton up to $\alpha$-equivalence.
\end{proof}

By consequence of \Cref{prop:skelplam:unique,prop:skelplam:exists}, the problem of enumerating planar lambda terms is equivalent to that of enumerating lambda skeletons.
From the point of view of lambda calculus, though, it is natural to further restrain the problem by asking that we only count \emph{$\beta$-equivalence classes}, which is equivalent (by the normalization theorem for linear lambda calculus) to only counting \emph{$\beta$-normal} lambda terms.

Let us recall the following well-known characterization of ($\beta$-)normal lambda terms, in mutual induction with so-called ``neutral'' terms:
\begin{itemize}
\item Any variable $x$ is neutral.
\item If $t$ is neutral and $u$ is normal then the application $t(u)$ is neutral.
\item If $t$ is neutral then $t$ is normal.
\item If $t$ is normal then the abstraction $\lambda x.t$ is normal.
\end{itemize}
These definitions ensure recursively that any term which is neutral or normal cannot contain a $\beta$-redex $(\lambda x.t)u$ as a subterm.
On the other hand, the standard formulation of neutral and normal terms can also plainly be recast as a property of the underlying lambda skeletons.
For reasons which will become apparent in \Cref{sec:diagrams:colored}, we refer to the proof that a lambda skeleton is neutral or normal as a ``coloring'' of that skeleton (and likewise for linear lambda terms, by reference to their underlying skeletons).
\begin{definition}\label{defn:coloring}\emph{
Letting $B$ and $R$ stand for ``blue'' and ``red'', we define two graded sets of lambda skeletons $\SNeu(i)$ and $\SNF(i)$, for $i\in\N$, via the following rules:
$$
\infer[v]{\wild \in \SNeu(1)}{}
\qquad
\infer[a]{p(q) \in \SNeu(j+k)}{p \in \SNeu(j) & q \in \SNF(k)}
\qquad
\infer[s]{p\in\SNF(i)}{p \in \SNeu(i)}
\qquad
\infer[\ell]{\lambda \wild.p\in\SNF(i)}{p \in \SNF(i+1)}
$$
Let $p \in \SLam(i)$ be a lambda skeleton.  A \definand{$c$-coloring} of $p$, for some $c \in \set{B,R}$, consists of a derivation $\pi$ of $p \in \SLam_c(i)$ using the rules $v$, $a$, $s$, and $\ell$.
We write $\pi : (p \in \SLam_c(i))$ to indicate that $\pi$ is a $c$-coloring of $p$.
In turn, a $c$-coloring of a linear lambda term consists of a $c$-coloring of its lambda skeleton.
We say that a linear lambda term is \definand{neutral} if it has a $B$-coloring, and \definand{normal} if it has a $R$-coloring.}
\end{definition}
\begin{proposition}
A linear lambda term is normal if and only if it has no subterms of the form $(\lambda x.t)(u)$.
\end{proposition}
We emphasize that throughout this paper we will be interested in colorings themselves, rather than in the mere fact that a linear lambda term is neutral or normal.
In particular, the \emph{size} of a normal or neutral linear lambda term will be defined as a property of its associated coloring.
\begin{definition}\label{defn:size}\emph{
Let $p \in \SLam(i)$ be a lambda skeleton, and let $\pi : (p \in \SLam_c(i))$ be a $c$-coloring of $p$.
The \definand{size} $|\pi|$ of $\pi$ is defined as the number of uses of the $s$-rule in $\pi$.
In turn, the size of a normal or neutral linear lambda term is defined as the size of its associated coloring.}
\end{definition}
By inspection, any linear lambda term has at most \emph{one} $B$-coloring or $R$-coloring, which is why it makes sense to define the size of a neutral or normal linear lambda term as the size of its associated coloring.
However, a term could certainly have \emph{both} a $B$-coloring and a $R$-coloring -- since every neutral lambda term is also normal -- and so size is really a property of a given linear lambda term \emph{when viewed as either a neutral term or as a normal term}.
Note that this is an instance of the concept of \emph{coherence} of an interpretation of typing judgments (in the sense of Reynolds \cite[Ch.~16]{reynolds98topl}, viewing the colors $B$ and $R$ as ``types''), and we will make further use of this style of definition in \Cref{sec:tutte}.

Under this definition, for example, the normal linear terms
$$
[x]x(\lambda y.y)\qquad
[x]\lambda y.yx\qquad
[x]\lambda y.xy
$$
all have size two, as exhibited by the following $R$-colorings of their skeletons:
$$
\small
\infer[\fbox{$s$}]{\wild(\lambda \wild.\wild) \in \SNF(1)}{
\infer[a]{\wild(\lambda \wild.\wild) \in \SNeu(1)}{
 \infer[v]{\wild \in \SNeu(1)}{} &
 \infer[\ell]{\lambda \wild.\wild \in \SNF(0)}{
 \infer[\fbox{$s$}]{\wild \in \SNF(1)}{\infer[v]{\wild \in \SNeu(1)}{}}
}}}
\qquad\qquad
\infer[\ell]{\lambda \wild.\wild \wild \in \SNF(1)}{
\infer[\fbox{$s$}]{\wild \wild \in \SNF(2)}{
\infer[a]{\wild\wild \in \SNeu(2)}{
 \infer[v]{\wild \in \SNeu(1)}{} & 
 \infer[\fbox{$s$}]{\wild \in \SNF(1)}{\infer[v]{\wild \in \SNeu(1)}{}}
}}}
$$
Observe, however, that $[x]x(\lambda y.y)$ has size one when viewed as a neutral term rather than as a normal term.

We now turn to the problem of counting neutral and normal planar lambda terms by size, which by \Cref{prop:skelplam:unique,prop:skelplam:exists} is equivalent to counting $B$-colorings and $R$-colorings.
It is straightforward to go from \Cref{defn:coloring,defn:size} to the following families of \emph{generating functions} $B_i(z)$ and $R_i(z)$ counting $B$- and $R$-colorings by size, where the index $i$ stands for the degree of the underlying lambda skeletons, and the coefficient of $z^n$ in each $B_i(z)$ and $R_i(z)$ counts the total number of $B/R$-colorings of size $n$:
$$
B_i(z) = [i=1] + \sum_{j+k=i} B_j(z) R_k(z) \qquad
R_i(z) = zB_i(z) + R_{i+1}(z)
$$
(Here ``$[i=1]$'' denotes the Iverson bracket, 1 if $i=1$ and 0 otherwise.)
Next, we can formally aggregate these families
$$
B(z,x) \defeq \sum_{i\ge 0} B_i(z) x^i  \qquad
R(z,x) \defeq \sum_{i\ge 0} R_i(z) x^i
$$
to define a single pair of generating functions counting colorings along both size and degree. 
By unfolding definitions, we can then check that $B(z,x)$ and $R(z,x)$ satisfy the following functional equations:
\begin{align}
B(z,x) &= x + B(z,x)R(z,x) \label{genfnA}\\
R(z,x) &= zB(z,x) + \frac{1}{x}(R(z,x)-R_0(z)) \label{genfnB}
\end{align}
Equations of this form can be solved by a technique known as the \emph{quadratic method} \cite[p.515]{flajolet-sedgewick}.
In particular we can solve for $R_0(z)$, the generating function counting $R$-colorings of lambda skeletons of degree 0 (hence, closed normal planar lambda terms) by size:
\begin{proposition}\label{prop:genfnB0}
The generating function $R_0(z)$ satisfies
$$R_0(z) = -\frac{1}{54z}\left(1-18z-(1-12z)^{3/2}\right).$$
\end{proposition}
\begin{proof}
Formula (\ref{genfnA}) becomes
$$B(z,x) = \frac{x}{1-R(z,x)}$$
and after substituting into (\ref{genfnB}) we derive:
\begin{align*}
R(z,x) &= \frac{zx}{1-R(z,x)} + \frac{1}{x}(R(z,x)-R_0(z))
\\ x(1-R(z,x))R(z,x) &= zx^2 + (1-R(z,x))(R(z,x)-R_0(z))
\\ (x-1)(1-R(z,x))R(z,x) &= zx^2 - (1-R(z,x))R_0(z)
\\ ((x-1)(1-R(z,x))-R_0(z))R(z,x) &= zx^2 - R_0(z)
\end{align*}
Then the idea is to define auxiliary functions $F(z,x)$ and $G(z,x)$ by
\begin{align*}
F(z,x) &\defeq x-1-R_0(z) - 2(x-1)R(z,x) \\
G(z,x) &\defeq F(z,x)^2 
\end{align*}
and look for a function $X(z)$ such that $F(z, X(z)) = 0$, implying that $G$ has a double root at $X$.
We have chosen $F(z,x)$ so that by the quadratic formula $G(z,x)$ simplifies to 
$$G(z,x) = (x-1-R_0(z))^2 - 4(x-1)(zx^2-R_0(z)),$$
and combined with the constraints $G(z,X(z)) = 0$ and $\frac{\partial}{\partial x}G(z,x)|_{x=X(z)} = 0$ we have a system of two equations in two unknowns $X(z)$ and $R_0(z)$.
This system of equations can be solved mechanically (for example using Maple),
$$
X(z) = \frac{12z + 1 - \sqrt{1-12z}}{18z} \qquad
R_0(z) = \frac{(12z-1)X(z) - 8z + 1}{3}
$$
and we obtain the stated formula for $R_0(z)$ by algebraic simplification.
\end{proof}
Now, the formula for $R_0(z)$ given in \Cref{prop:genfnB0} is just one factor of $z$ times the known generating function for counting rooted planar maps by number of edges \cite[Proposition VII.11]{flajolet-sedgewick}:
$$-\frac{1}{54z^2}\left(1-18z-(1-12z)^{3/2}\right)$$
Since we also trivially have $R_1(z) = R_0(z)$ (corresponding to the fact that any closed normal lambda term must be a lambda abstraction), we obtain the
\begin{cor}
The number of rooted planar maps with $n$ edges is equal to the number of closed normal planar lambda terms (= $R$-colorings of degree 0) of size $n+1$, and to the number of normal planar lambda terms with one free variable (= $R$-colorings of degree 1) of size $n+1$.
\end{cor}
\begin{figure}
\begin{center}
\begin{tabular}{c|cccccccccc}
\diagbox{$i$}{$n$} & 0 & 1 & 2 & 3 & 4 & 5 & 6 & 7 & 8 & 9 \\
\hline
1 & 1 & 1 & 3 & 14 & 83 & 570 & 4318 & 35068 & 299907 & 2668994 \\
2 & 0 & 1 & 4 & 20 & 120 & 820 & 6152 & 49448 & 418800 & 3694740 \\
3 & 0 & 0 & 2 & 15 & 105 & 770 & 5985 & 49014 & 419370 & 3720420 \\
4 & 0 & 0 & 0 & 5 & 56 & 504 & 4368 & 38136 & 339696 & 3094896 \\
5 & 0 & 0 & 0 & 0 & 14 & 210 & 2310 & 23100 & 224070 & 2161236 \\
6 & 0 & 0 & 0 & 0 & 0 & 42 & 792 & 10296 & 116688 & 1245816
\end{tabular}
\end{center}
\caption{The number of neutral planar lambda terms of size $n$ with $i$ free variables.}
\label{fig:counting-neutral}
\end{figure}
\begin{figure}
\begin{center}
\begin{tabular}{c|cccccccccc}
\diagbox{$i$}{$n$} & 1 & 2 & 3 & 4 & 5 & 6 & 7 & 8 & 9 & 10 \\
\hline
0 & 1 & 2 & 9 & 54 & 378 & 2916 & 24057 & 208494 & 1876446 & 17399772 \\
1 & 1 & 2 & 9 & 54 & 378 & 2916 & 24057 & 208494 & 1876446 & 17399772 \\
2 & 0 & 1 & 6 & 40 & 295 & 2346 & 19739 & 173426 & 1576539 & 14730778 \\
3 & 0 & 0 & 2 & 20 & 175 & 1526 & 13587 & 123978 & 1157739 & 11036038 \\
4 & 0 & 0 & 0 & 5 & 70 & 756 & 7602 & 74964 & 738369 & 7315618 \\
5 & 0 & 0 & 0 & 0 & 14 & 252 & 3234 & 36828 & 398673 & 4220722 \\
6 & 0 & 0 & 0 & 0 & 0 & 42 & 924 & 13728 & 174603 & 2059486
\end{tabular}
\end{center}
\caption{The number of normal planar lambda terms of size $n$ with $i$ free variables.}
\label{fig:counting}
\end{figure}
From the solution for $R_0(z)$ we can also derive algebraic generating functions for $B(z,x)$ and $R(z,x)$, and use these to compute tables of coefficients, such as the small ones listed in \Cref{fig:counting,fig:counting-neutral}.
As a couple of simple observations we note that:
\begin{itemize}
\item The series counting neutral planar lambda terms with one free variable (i.e., the coefficients of $B_1(z)$, corresponding to row $i=1$ of \Cref{fig:counting-neutral}) also appears in the OEIS as series {\bf A220910}.
\item Adding up each column of \Cref{fig:counting-neutral} gives the first row of \Cref{fig:counting}, what can be expressed in generating functions by the equation $R(z,0) = zB(z,1)$.
Bijectively, this corresponds to the fact that any closed normal lambda term begins with a series of $i$ lambda abstractions, applied to a neutral term with $i$ free variables.
\end{itemize}
Finally, although the notion of size given in \Cref{defn:size} turns out to be a natural one for neutral and normal linear terms, let us point out that it has various equivalent formulations.
For example, it is almost identical to counting the total number of variable uses in a term:
\begin{proposition}\label{prop:size}
The size of a normal linear lambda term $t$ with $R$-coloring $\pi$ is equal to the total number of ($V$-)leaves in its underlying lambda skeleton, or explicitly $|\pi| = |t|$, where
$$
|\lambda x.t| = |t| \quad\text{and}\quad
|x| = 1\quad\text{and}\quad
|t(u)| = |t| + |u|.
$$
On the other hand the size of a neutral linear lambda term $t$ with $B$-coloring $\pi$ is equal to the total number of leaves in its lambda skeleton minus one, i.e., $|\pi| = |t|-1$.
\end{proposition}
To prove this, we first recall the standard lambda calculus notion of \emph{head normal form} \cite[p.173]{barendregt1984}, which will also be useful in \Cref{sec:tutte}.
\begin{definition}\label{defn:headvar}\emph{
Let $\pi : (p \in \SNeu(i))$ be a $B$-coloring of a skeleton $p$.
The \definand{head} of $\pi$ is defined as the unique occurrence of ``$\wild$'' in $p$ reached by walking up the derivation $\pi$ from the last rule applied, and always following the left branch of an $a$-rule until reaching a $v$-rule.
In turn, the \definand{head variable} of a linear lambda term $[\Gamma]t \in \LLin(p)$ with $B$-coloring $\pi : (p \in \SNeu(i))$ is defined as the variable annotating the head of $\pi$.
(Note that the head variable of a neutral linear term $[\Gamma]t$ necessarily occurs in $\Gamma$.)
Alternatively, let $\pi : (p \in \SNF(i))$ be a $R$-coloring of a skeleton $p$.
The \definand{body} of $\pi$ is defined as the unique subskeleton of $p$ reached by walking up the derivation $\pi$ from the last rule applied, always moving to the premise of an $\ell$-rule until reaching the premise of an $s$-rule.
In turn, the \definand{neutral body} of a linear lambda term $[\Gamma]t \in \LLin(p)$ with $R$-coloring $\pi : (p \in \SNF(i))$ is defined as the (neutral) linear lambda term annotating the head of $\pi$.}
\end{definition}
For example, $y$ is the head variable of the neutral term $[y,x](y(\lambda z.z))(\lambda w.wx)$, while $[w,x]wx$ is the neutral body of the normal subterm $[x]\lambda w.wx$.
To check that this agrees with \Cref{defn:headvar}, here is the corresponding coloring, where for clarity we have kept all of the variable names decorating the skeleton:
$$
\infer[a]{(y(\lambda z.z))(\lambda w.wx) \in \SNeu(2)}{
 \infer[a]{y(\lambda z.z) \in \SNeu(1)}{
   \infer[v]{y \in \SNeu(1)}{} &
   \infer[\ell]{\lambda z.z \in \SNF(0)}{\infer[s]{z \in \SNF(1)}{\infer[v]{z \in \SNeu(1)}{}}}
  } &
 \infer[\ell]{\lambda w.wx \in \SNF(1)}{
  \infer[s]{wx \in \SNF(2)}{
  \infer[a]{wx \in \SNeu(2)}{
  \infer[v]{w \in \SNeu(1)}{} &
  \infer[s]{x \in \SNF(1)}{\infer[v]{x \in \SNeu(1)}{}}}}}}
$$
\begin{proof}[Proof of \Cref{prop:size}]
Let $\pi$ be the $R$-coloring of a normal linear lambda term $t$.
There is a one-to-one correspondence between $s$-nodes in $\pi$ and variables occurring in $t$, by walking up from the neutral body of an $s$-node to the corresponding head variable, and walking back down to the conclusion of the $s$-node.
The same argument works if we begin with a neutral linear term $t$ with $B$-coloring $\pi$, except that the head variable of $t$ itself does not lead back to an $s$-node.
\end{proof}

\subsection{Related work}
\label{sec:related}

Although lambda calculus is an old subject, its combinatorial aspects have been relatively less studied.
In the published literature, Grygiel and Lescanne \cite{grygiel-lescanne} give recurrence formulas and generating functions for counting pure lambda terms of a given size and number of free variables, as well as for counting normal forms, while David et.~al \cite{david+2013} study asymptotic properties of normalization.
Most closely related to the present paper are works on the combinatorics of linear lambda calculus (under the alternative name of ``BCI'' combinatory logic), by Grygiel, Idziak, and Zaionc \cite{grygiel-idziak-zaionc} and by Bodini, Gardy, and Jacquot \cite{bodini-et-al}.
Both note a connection between the sequence counting general linear lambda terms (not necessarily $\beta$-normal) to series {\bf A062980} of the OEIS, but they do not consider planarity.

\section{A graphical language for (neutral and normal) linear lambda terms}
\label{sec:diagrams}

\Cref{sec:tutte} establishes a bijection between normal planar lambda terms and rooted planar maps, relying on an inductive classification of rooted planar maps due to Tutte.
On the other hand, rooted maps also have a very concrete topological interpretation.
So, for the purpose of \emph{explaining} the bijection -- as well as for better understanding the motivation for studying normal planar lambda terms in the first place -- it is helpful to have an analogous graphical representation of linear lambda terms.

The representation we will use is a variation of an old representation sometimes
referred to as \emph{lambda-graphs with back-pointers}
\cite{statman74phd,ariolablom97,mairson2002dilbert}, and which itself can be
seen as a variation on linear logic proof-nets \cite{girard-linear-logic}
adapted to the special case of linear lambda calculus \cite{guerrini2001}.
Other than some superficial syntactic differences, the main refinement we
introduce is the addition of a coloring protocol that exactly reflects the
restriction of the diagrams to normal and neutral linear lambda terms.
Rather than presenting these diagrams as ``colored lambda-graphs'' or ``colored
proof-nets'', however, one aim of the next section is to explain lambda-graphs
within the well-understood framework of \emph{string diagrams}, which were
originally introduced by Joyal and Street \cite{joyal-street-i} as a categorical
formalization of many different kinds of diagrammatic reasoning (such as Penrose
diagrams in physics).
As far as we know, this rational reconstruction of lambda-graphs is new,
although the ideas we present are quite simple -- just enough to motivate the
coloring protocol.
We will try to keep the exposition relatively elementary, but some background in
category theory and lambda calculus may be helpful for reading
\Cref{sec:diagrams:scott}.
On the other hand, the intrepid reader may try skipping straight to
\Cref{fig:cdiagrams,fig:lam<=3} (in \Cref{sec:diagrams:colored}) to get a quick
feel for the colored diagrams, before heading to \Cref{sec:tutte} where we make
extensive use of this graphical language.

\subsection{From reflexive objects to lambda-graphs}
\label{sec:diagrams:scott}

Our starting point is the insight, due to Dana Scott \cite{scott1980}, that whereas terms of simply-typed lambda calculus can be interpreted as morphisms in arbitrary cartesian closed categories (see, e.g., \cite{lambekscott}), terms of pure (or ``untyped'') lambda calculus can \emph{also} be modelled internally to a cartesian closed category, given an object of that category equipped with a certain special structure (turning it into a so-called \emph{reflexive object}).

Let us recall (for background and details see \cite{maclane}) that a cartesian closed category can be described as a \emph{closed symmetric monoidal category} in which the monoidal structure is cartesian:
\begin{itemize}
\item 
A \definand{monoidal category} is a category $\cal{C}$ equipped with a tensor product and unit operation
$$
\mul : \cal{C} \times \cal{C} \to \cal{C}
\qquad
I : 1 \to \cal{C}
$$
which are associative and unital up to coherent isomorphism.
\item
It is \definand{closed} if in addition it is equipped with left and right residuation operations
$$
\setminus : \cal{C}^\op  \times \cal{C} \to \cal{C}
\qquad
/ : \cal{C}  \times \cal{C}^\op \to \cal{C}
$$
which are right adjoint to the tensor product in each component:
$$
\cal{C}(y, \resL[z]{x}) \cong 
\cal{C}(x \mul y, z) \cong
\cal{C}(x, \resR[z]{y})
$$
Note that this is equivalent to the existence of a pair of natural transformations
$$
\xymatrix{\cal{C}(y,\resL[z]{x}) & \cal{C}(x\mul y,z) \ar[l]_{\lambda^x_{y,z}}\ar[r]^{\rho^y_{x,z}} & \cal{C}(x,\resR[z]{y})}
$$
together with a pair of \emph{evaluation maps}
$$
\xymatrix{x \mul (\resL[z]{x}) \ar[r]^-{\plugL_{x,z}} & z & (\resR[z]{y}) \mul y \ar[l]_-{\plugR_{y,z}}}
$$
satisfying equations
$$
((\id_x\mul \lc[x]{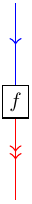});\plugL) = f = ((\rc[y]{f}\mul\id_y);\plugR)
$$

$$
g = \lc[x]{(\id_x\mul g);\plugL}
\qquad
\rc[y]{(h\mul \id_y);\plugR} = h
$$
for all morphisms $f : x \mul y \to z$ and $g : y \to \resL[z]{x}$ and $h : x
\to \resR[z]{y}$.

\item
It is \definand{symmetric} if there is a family of isomorphisms
$$
\symiso_{x,y} : x\mul y \overset\sim\to y \mul x
$$
which are involutive in the sense that $(\symiso_{x,y};\symiso_{y,x}) = \id_{x\mul y}$ for all $x,y \in \cal{C}$, and which satisfy a few additional, natural equations.
\item
It is \definand{cartesian} if the tensor product coincides with the categorical product, this meaning that we have a natural isomorphism
$$
\cal{C}(x, y\mul z) \cong \cal{C}(x,y) \times \cal{C}(x,z).
$$
Note that this is equivalent to the existence of a family of \emph{duplication} and \emph{erasure} operations
\begin{align*}
\Delta_x &: x \to x\mul x \\
e_x &: x \to I
\end{align*}
satisfying certain natural equations.
\end{itemize}
Any cartesian monoidal category is also symmetric, and the tensor product is usually written $x \mul y = x \times y$ and called a \emph{categorical product} (or simply a product), while the unit is written $I = 1$.
In a closed cartesian monoidal category (more often called a cartesian closed category, or ``ccc''), the left and right residuals, which are isomorphic, are usually written $\resL[y]{x} \cong \resR[y]{x} = y^x$ and called \emph{exponential objects}.

Now, Scott defined a \definand{reflexive object} in a ccc $\cal{C}$ as an object $u \in \cal{C}$ equipped with a pair of morphisms
$$
\xymatrix{
u \ar@<.5ex>[r]^{A} & u^u\ar@<.5ex>[l]^{L}
}
$$
such that the $L;A = \id_{u^u}$. The idea is that the two morphisms $A$ and $L$ model the operations of application and lambda abstraction, respectively, while the equation
\begin{align}
L;A &= \id_{u^u} \label{eqn:beta}
\end{align}
models $\beta$-conversion (or more precisely $\beta$-equivalence).
A trivial reflexive object in the cartesian closed category of sets and functions takes $u$ to be the one-element set $1 = \{*\}$, with $L$ and $A$ witnessing the isomorphism $1^1 \cong 1$.
In fact, for cardinality reasons, this is the \emph{only} reflexive object in the category of sets and functions.
On the other hand, Scott also gave an explicit construction of a non-trivial model of the pure lambda calculus by taking $u = \cal{P}(\N)$ to be the lattice of subsets of the natural numbers, and $u^u$ to be not the space of \emph{all} functions $\cal{P}(\N) \to \cal{P}(\N)$, but rather only those functions preserving directed joins \cite{scott1976}.
In terms of the abstract axiomatization introduced in \cite{scott1980}, Scott's (earlier) construction could be interpreted as building a reflexive object in the cartesian closed category of \emph{domains and continuous functions}.

It is also possible to consider a dual equation
\begin{align}
  \id_u &= A;L \label{eqn:eta}
\end{align}
modeling $\eta$-equivalence, which induces an isomorphism $u^u \cong u$, but Scott's definition of reflexive object (and his original model in \cite{scott1976}) only required that $u^u$ be a \emph{retract} of $u$.
A simple but important observation, however, is that the principle of $\alpha$-equivalence is valid by construction even without either equation (\ref{eqn:beta}) or (\ref{eqn:eta}), since the definition itself involves only the two operations $A$ and $L$, with no mention of formal variables.

The idea here is closely related to a technique sometimes used in programming languages and proof assistants under the heading of \emph{higher-order abstract syntax} (HOAS). 
Since a reflexive object lives inside a ccc, and since terms of \emph{simply-typed} lambda calculus may be interpreted in any ccc, we can use lambda calculus itself, in addition to the operations $A$ and $L$, in order to construct morphisms in $\cal{C}$ denoting pure lambda terms.
For example, the closed lambda term
$$\lambda x.xx$$
may be encoded in $\cal{C}$ as the morphism $\epsilon[\lambda x.xx] : 1 \to u$ defined by
\begin{align}
\epsilon[\lambda x.xx] = \lc[u]{\Delta;(A\times\id_u);\plugR};L \label{eqn:lamx.xx}
\end{align}
where we have applied the ``currying'' transformation $\lambda^u$ to the morphism
$$
\xymatrix{
u \ar[r]^-{\Delta_u} & u\times u \ar[rr]^{A\times \id_u} && u^u \times u \ar[r]^-{\plugR} & u
}
$$
to obtain a morphism $1 \to u^u$, and then composed with the operation $L : u^u \to u$.
But this can be more slickly written simply as
$$\epsilon[\lambda x.xx] = L(\bar\lambda x.A(x)@ x)$$
where the ``$\bar\lambda$'' and ``$@$'' in the definition of $\epsilon[t]$ correspond to lambda abstraction and application interpreted by appeal to the ``meta-level'', so to speak -- in other words, translated mechanically into the more explicit definition (\ref{eqn:lamx.xx}) by invoking the ccc structure of $\cal{C}$.
Similarly, the closed term
$$t = (\lambda x.xx)(\lambda y.y)$$
may be encoded in $\cal{C}$ as the morphism (again of type $1 \to u$)
\begin{align*}
\epsilon[t] &\defeq A(L(\bar\lambda x.A(x)@x))@(L(\bar\lambda y.y))
\end{align*}
and now by purely equational reasoning we can verify, for example, that the morphism encoding $t$ is equal to the morphism encoding $\lambda y.y$:
\begin{align*}
\epsilon[t] &= A(L(\bar\lambda x.A(x)@x))@(L(\bar\lambda y.y)) \tag{by definition} \\
 &= (\bar\lambda x.A(x)@x)@(L(\bar\lambda y.y)) \tag{by \ref{eqn:beta}} \\
 &= A(L(\bar\lambda y.y))@(L(\bar\lambda y.y)) \tag{by ccc axioms} \\
 &= (\bar\lambda y.y)@(L(\bar\lambda y.y)) \tag{by \ref{eqn:beta}} \\
 &= L(\bar\lambda y.y) \tag{by ccc axioms} \\
 &= \epsilon[\lambda y.y] \tag{by definition}
\end{align*}
Although it might at first appear circular, this kind of trick is often useful in practice.

Our next step is to observe that Dana Scott's idea also works perfectly well for modelling \emph{linear} lambda calculus in its pure, untyped form, if one simply drops the condition that $\cal{C}$ be a ccc and replaces it by the weaker condition that $\cal{C}$ be a closed symmetric monoidal category (smcc).
As in a ccc, in a smcc the left and right residuals are isomorphic $\resL[y]{x} \cong \resR[y]{x}$, and they are sometimes denoted collectively by $[x,y]$ (matching the notation for the internal hom in category theory) or by $x \Lolli y$ (matching the notation for the implication connective in linear logic).
For what comes next, however, it will be important for us to maintain the distinction between the two isomorphic forms of residuals, and moreover to give an explicit name
$$
\turniso_{x,y} : \resL[y]{x} \isoto \resR[y]{x}
$$
for the isomorphism from the left residual to the right residual.
\begin{definition}\label{defn:linref}\emph{
A \definand{linear reflexive object} in a smcc $\cal{C}$ is an object $u \in \cal{C}$ equipped with a pair of morphisms
$$
\xymatrix{
\resL[u]{u} \ar[r]^-{L} & u\ar[r]^-{A} & \resR[u]{u}
}
$$
such that $L;A = \turniso_{u,u}$.}
\end{definition}
There are certainly some degrees of freedom in this definition that one might consider.
For example, one could imagine defining a linear reflexive object as an object equipped with a pair of morphisms
$$
\xymatrix{
\resR[u]{u} \ar[r]^-{L'} & u\ar[r]^-{A'} & \resL[u]{u}
}
$$
such that $L';A' = \turniso^{-1}_{u,u}$, or perhaps as one equipped with a pair
$$
\xymatrix{
u \ar@<.5ex>[r]^{A''} & \resR[u]{u}\ar@<.5ex>[l]^{L''}
}
$$
such that $L'';A'' = \id_{\resR[u]{u}}$, and so on.
We will come back to the difference between these conventions later, but for now we want to take \Cref{defn:linref} as given, and explain how to go from there to a graphical representation of linear lambda terms, by applying the principles of string diagrams more or less mechanically.

We refer to Selinger's survey article \cite{selinger-survey} for background reading.
Briefly, the basic starting point for string diagrams is to dualize the usual object-and-arrow diagrams of category theory, so that objects become (possibly labelled) wires (or ``strings''), and arrows become nodes between wires:
$$\xymatrix{x \ar[r]^F & y} \quad\leadsto\quad 
\imgcenter{\includegraphics[scale=0.8]{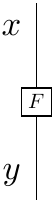}}$$
Moreover, composition of morphisms is depicted by connecting diagrams end-to-end,
$$\xymatrix{x \ar[r]^F & y \ar[r]^G & z} \quad\leadsto\quad 
\imgcenter{\includegraphics[scale=0.8]{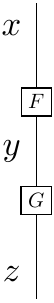}}$$
while the tensor product is depicted by laying out diagrams in parallel (with the tensor unit represented by the blank page):
$$\xymatrix{x\mul z \ar[r]^{F\mul G} & y \mul w} \quad\leadsto\quad 
\imgcenter{\includegraphics[scale=0.8]{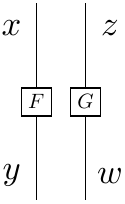}}
\qquad\qquad
\xymatrix{I \ar[r]^I & I} \quad\leadsto\quad
$$
In general, string diagrams may be treated up to deformation, meaning roughly that it is possible to freely stretch and bend wires and move around nodes, so long as the \emph{interface} of the diagram (i.e., the boundary of input and output wires) remains fixed (for a more precise definition, see \cite{joyal-street-i}).
For example, all of the diagrams
$$
\imgcenter{\includegraphics[height=2cm]{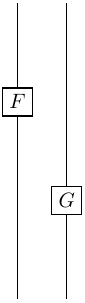}}
\quad=\quad
\imgcenter{\includegraphics[height=2cm]{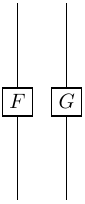}}
\quad=\quad
\imgcenter{\includegraphics[height=2cm]{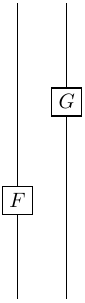}}
$$
are essentially interchangeable, where the isotopy of diagrams is justified by the equations
$$
(F; \id)\mul(\id; G) = (F\mul \id);(\id\mul G) = F\mul G = (\id\mul G);(F\mul \id) = (\id;F)\mul(G;\id)
$$
which hold in any monoidal category.

This basic setup may then be developed by supposing that the monoidal category is equipped with additional structure.
For example, the symmetry isomorphisms of a symmetric monoidal category are naturally depicted as crossing wires:
$$
\xymatrix{x\mul y \ar[r]^{\symiso_{x,y}} & y \mul x} \quad\leadsto\quad 
\imgcenter{\includegraphics[scale=0.8]{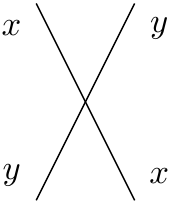}}$$
Representing the evaluation maps and currying transformations of a smcc is in general a bit more subtle (cf.~\cite{baez-stay-rosetta}), but there is a special class of closed symmetric monoidal categories known as \emph{compact closed categories}, which have a particularly simple and elegant graphical language.
A symmetric monoidal category is said to be \definand{compact closed} if every object is equipped with left and right \emph{duals}, where a \definand{right dual} of $x \in \cal{C}$ is an object $\rdual{x} \in \cal{C}$ together with a pair of maps
$$\xymatrix{I \ar[r]^-\eta & x \mul \rdual{x}}\qquad
\xymatrix{\rdual{x} \mul x \ar[r]^-\varepsilon & I}$$
such that $(\eta \mul \id_x);( \id_x\mul \varepsilon) = \id_x$ and $(\id_{\rdual{x}} \mul \eta);( \varepsilon \mul \id_{\rdual{x}}) = \id_{\rdual{x}}$, and similarly a \definand{left dual} of $x \in \cal{C}$ is an object $\ldual{x} \in \cal{C}$ together with a pair of maps
$$
\xymatrix{I \ar[r]^-{\eta'} & \ldual{x} \mul x}\qquad
\xymatrix{x \mul \ldual{x} \ar[r]^-{\varepsilon'} & I} 
$$
such that $(\eta' \mul \id_{\ldual{x}});( \id_{\ldual{x}}\mul \varepsilon') = \id_{\ldual{x}}$ and $(\id_x \mul \eta');( \varepsilon' \mul \id_x) = \id_x$.
Note that any compact closed category is also closed (i.e., has left and right residuals), by defining
\begin{align}
\resL[y]{x} \defeq \ldual{x}\mul y
\qquad
\resR[y]{x} \defeq y\mul \rdual{x} \label{defn:compactsmcc}
\end{align}
and using the maps $\eta^{(')}$  and $\varepsilon^{(')}$ to build the associated currying transformations and evaluation maps.
Also note that whenever both left and right duals exist in a symmetric monoidal category they are necessarily isomorphic, and hence the isomorphism $\ldual{x} \cong \rdual{x}$ holds in any compact closed category.

String diagrams for compact closed categories (cf.~\cite[\S4]{selinger-survey}) are constructed by first assigning \emph{orientations} to the wires to distinguish an object from its duals:
$$
x \in \cal{C} \quad\leadsto\quad
\imgcenter{\includegraphics{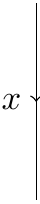}}
\qquad\qquad
\rdual{x},\ldual{x} \in \cal{C} \quad\leadsto\quad
\imgcenter{\includegraphics{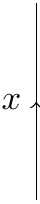}}
$$
Then, the $\eta^{(')}$  and $\varepsilon^{(')}$ maps are depicted as oriented caps and cups, 
$$\xymatrix{I \ar[r]^-\eta & x \mul \rdual{x}}
\quad\leadsto\quad
\imgcenter{\includegraphics[scale=0.8]{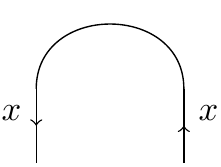}}
\qquad\qquad
\xymatrix{\rdual{x} \mul x \ar[r]^-\varepsilon & I}
\quad\leadsto\quad
\imgcenter{\includegraphics[scale=0.8]{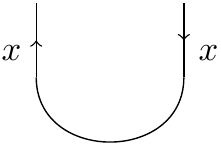}}
$$
$$\xymatrix{I \ar[r]^-{\eta'} & \ldual{x} \mul x}
\quad\leadsto\quad
\imgcenter{\includegraphics[scale=0.8]{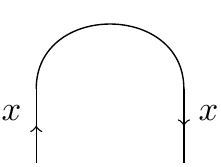}}
\qquad\qquad
\xymatrix{x \mul \ldual{x} \ar[r]^-{\varepsilon'} & I}
\quad\leadsto\quad
\imgcenter{\includegraphics[scale=0.8]{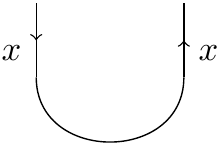}}
$$
while the equations governing them correspond to ``straightening'' the wires:
$$
(\eta \mul \id_x);( \id_x\mul \varepsilon) = \id_x
\ \leadsto\ 
\imgcenter{\includegraphics[height=1.4cm]{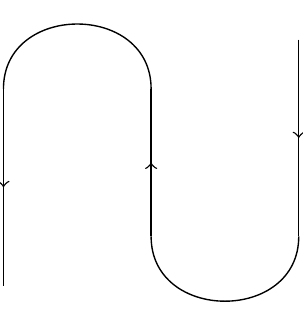}}\quad=\quad
\imgcenter{\includegraphics[height=1.4cm]{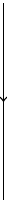}}
\qquad
(\id_{\rdual{x}} \mul \eta);( \varepsilon \mul \id_{\rdual{x}}) = \id_{\rdual{x}}
\ \leadsto\ 
\imgcenter{\includegraphics[height=1.4cm]{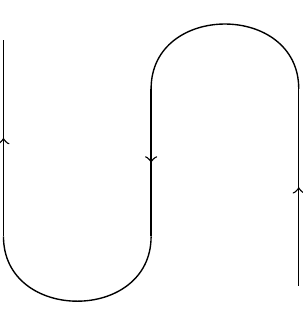}}\quad=\quad
\imgcenter{\includegraphics[height=1.4cm]{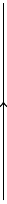}}
$$
$$(\eta' \mul \id_{\ldual{x}});( \id_{\ldual{x}}\mul \varepsilon') = \id_{\ldual{x}}
\ \leadsto\ 
\imgcenter{\includegraphics[height=1.4cm]{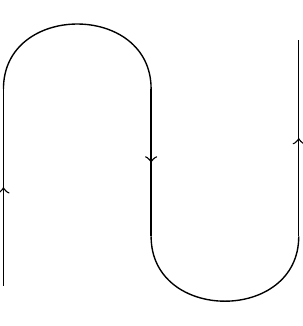}}\quad=\quad
\imgcenter{\includegraphics[height=1.4cm]{\DIAGRAMS/up-wire-plain.pdf}}
\qquad
(\id_x \mul \eta');( \varepsilon' \mul \id_x) = \id_x
\ \leadsto\ 
\imgcenter{\includegraphics[height=1.4cm]{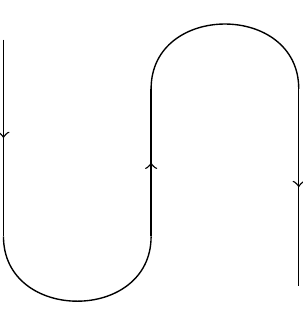}}\quad=\quad
\imgcenter{\includegraphics[height=1.4cm]{\DIAGRAMS/down-wire-plain.pdf}}
$$
Now, since every compact closed category is also a smcc, we can ask what it means for a compact closed category $\cal{C}$ to admit a linear reflexive object.
Expanding \Cref{defn:linref} in terms of the canonical description (\ref{defn:compactsmcc}) of left and right residuals in a compact closed category, a linear reflexive object in $\cal{C}$ consists of an object $u \in \cal{C}$ equipped with a pair of morphisms
$$
\xymatrix{
\ldual{u}\mul u \ar[r]^-{L} & u\ar[r]^-{A} & u\mul \rdual{u}
}
$$
such that $L;A = \turniso_{u,u}$, where the map $\turniso_{u,u}$ is constructed using the symmetry and the isomorphism $\ldual{u} \cong \rdual{u}$.
In turn, following the diagrammatic conventions for compact closed categories, such a structure corresponds to a pair of basic ``components'' $L$ and $A$,
$$\xymatrix{\ldual{u}\mul u \ar[r]^-{L} & u}
\quad\leadsto\quad
\imgcenter{\includegraphics[scale=0.8]{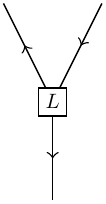}}
\qquad\qquad
\xymatrix{u\ar[r]^-{A} & u\mul \rdual{u}}
\quad\leadsto\quad
\imgcenter{\includegraphics[scale=0.8]{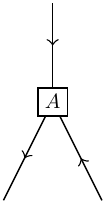}}$$
satisfying the following graphical equation:
$$
L;A = \turniso_{u,u}
\quad\leadsto\quad
\imgcenter{\includegraphics[scale=0.8]{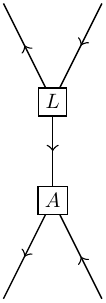}}
\quad=\quad
\imgcenter{\includegraphics[scale=0.8]{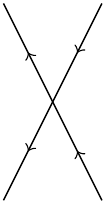}}
$$
To see how this plays out in the interpretation of linear lambda terms, let us first take the step of rendering $L$-nodes by black vertices and $A$-nodes by white vertices, so as to make the diagrams a bit more evocative:
$$
\imgcenter{\includegraphics[height=1.5cm]{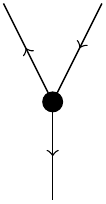}}
\qquad
\vrule
\qquad
\imgcenter{\includegraphics[height=1.5cm]{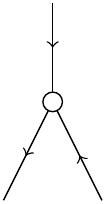}}
\qquad
\vrule
\qquad
\imgcenter{\includegraphics[height=2cm]{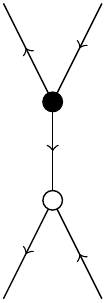}}
\quad=\quad
\imgcenter{\includegraphics[height=1.5cm]{\DIAGRAMS/turn.pdf}}
$$
As suggested by the orientations on the wires, $L$-nodes and $A$-nodes can actually be interpreted as \emph{operations on lambda terms}, with certain wires representing inputs and other wires representing outputs.
We can visualize this by annotating the wires explicitly with input and output terms:\footnote{These annotations are not to be confused with the convention of labelling wires by objects of the category. Here, every wire represents either the linear reflexive object $u$ or its duals $\ldual{u} \cong \rdual{u}$, and so the orientations suffice as object labels.}
\begin{equation}
\tag{LR}\label{rules:LR}
\imgcenter{\includegraphics[height=2.5cm]{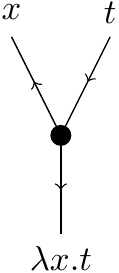}}
\qquad
\vrule
\qquad
\imgcenter{\includegraphics[height=2.5cm]{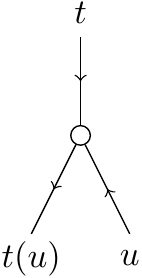}}
\qquad
\vrule
\qquad
\imgcenter{\includegraphics[height=3cm]{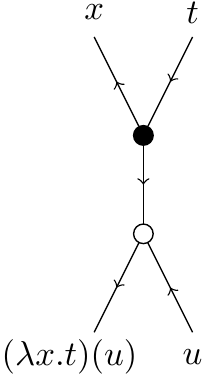}}
\quad=\quad
\imgcenter{\includegraphics[height=2.5cm]{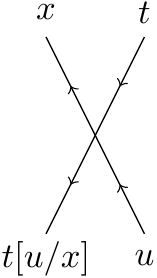}}
\end{equation}
Intuitively, an $L$-node emits a fresh variable $x$ on one wire, then \emph{binds} $x$ in the term $t$ on its incoming wire to produce a term $\lambda x.t$ on its other outgoing wire.
Similarly, an $A$-node takes two lambda terms $t$ and $u$ on its incoming wires, and outputs the application $t(u)$ on its outgoing wire.
Moreover, these interpretations are compatible with the graphical interpretation of $\beta$-conversion.

Formally, suppose we are given a lambda skeleton $p \in \SLam(i)$ and a \emph{derivation} $\pi$ of $\islinear{\G}{t}{p}$, witnessing the fact that $t$ is a linear lambda term decorating $p$.
Then whenever we have a linear reflexive object $(u,L,A)$ in a smcc $\cal{C}$, first of all we can define a morphism $\enc[u]{\pi} : \enc[u]{\Gamma} \to u$ in $\cal{C}$, where $\enc[u]{\Gamma} \in \cal{C}$ is defined inductively by $\enc[u]{x} = u$, $\enc[u]{\Gamma,\Delta} = \enc[u]{\Gamma} \mul \enc[u]{\Delta}$, $\enc[u]{\cdot} = I$.
Moreover, we can define a diagram $\pi^d$ with $i$ incoming wires and one outgoing wire
$$
\imgcenter{\includegraphics[height=2.5cm]{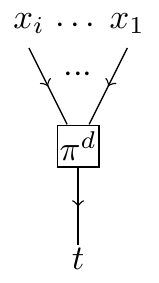}}
$$
which can be seen as a representation of the image of $\enc[u]{\pi}$ in the free compact closed category over $\cal{C}$.
The morphism $\enc[u]{\pi}$ and diagram $\pi^d$ are defined by induction on $\pi$ as follows:
\begin{description}
\item[{Case $\pi = \infer[V]{\islinear{x}{x}{\wild}}{}$}\ ]
Then $\enc[u]{\pi} = \xymatrix{u \ar[r]^{\id_u} & u}$
and
$\pi^d = \imgcenter{\includegraphics{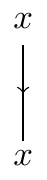}}$. (Draw an oriented wire.)

\item[{Case $\pi = \infer[A]{\islinear{\Gamma,\Delta}{t(u)}{p(q)}}{\deduce{\islinear{\Gamma}{t}{p}}{\pi_1} & \deduce{\islinear{\Delta}{u}{q}}{\pi_2}}$}\ ]
Then 
\\
$\enc[u]\pi = \xymatrix{\enc[u]\Gamma \mul \enc[u]\Delta \ar[rr]^-{\enc[u]{\pi_1} \mul \enc[u]{\pi_2}} && u \mul u \ar[r]^-{A\mul\id_u} & (\resR[u]{u}) \mul u \ar[r]^-{\plugR} & u}$
and $\pi^d = \hspace*{-1.5em}\imgcenter{\includegraphics{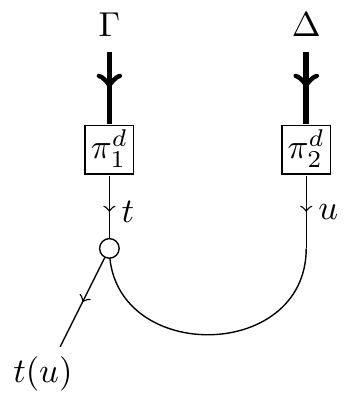}}$. \\
(Connect the outgoing wire of $\pi_2^d$ and the outgoing wire of $\pi_1^d$ to an $A$-node.)


\item[{Case $\pi = \infer[L]{\islinear{\Gamma}{\lambda x.t}{\lambda \wild.p}}{\deduce{\islinear{x,\Gamma}{t}{p}}{\pi_1}}$}\ ]
Then $\enc[u]\pi = \xymatrix{\enc[u]\G \ar[r]^-{\lc[u]{\enc[u]{\pi_1}}} & \resL[u]{u} \ar[r]^-{L} & u}$
and $\pi^d = \imgcenter{\includegraphics{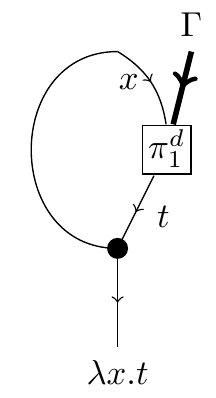}}$.\\
(Connect the outgoing wire and the leftmost incoming wire of $\pi_1^d$ to an $L$-node.)
\item[{Case $\pi = \infer[T]{\islinear{\Gamma,x,y,\Delta}{t}{p}}{\deduce{\islinear{\Gamma,y,x,\Delta}{t}{p}}{\pi_1}}$}\ ]
Then \\
$\enc[u]\pi = \xymatrixcolsep{1.25pc}\xymatrix{\enc[u]\G \mul u \mul u \mul \enc[u]\D \ar[rrr]^-{\id_{\enc[u]\G}\mul\symiso_{u,u}\mul \id_{\enc[u]\D}} &&& \enc[u]\G\mul u \mul u \mul \enc[u]\D \ar[r]^-{\enc[u]{\pi_1}} & u}$
and $\pi^d = \imgcenter{\includegraphics{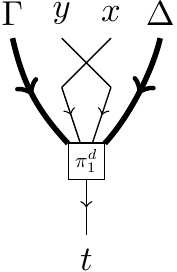}}$.\\
(Crossover the input wires representing $x$ and $y$.)
\end{description}
Here are some example lambda terms together with their annotated diagrams:
$$
[x]\lambda y.yx \quad\leadsto\quad
\imgcenter{\includegraphics[scale=0.8]{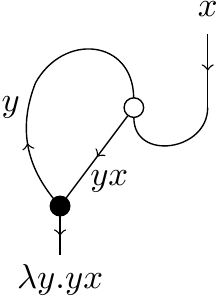}}
\qquad\qquad
[u]\lambda v.(\lambda w.wv)u\quad\leadsto\quad
\imgcenter{\includegraphics[scale=0.8]{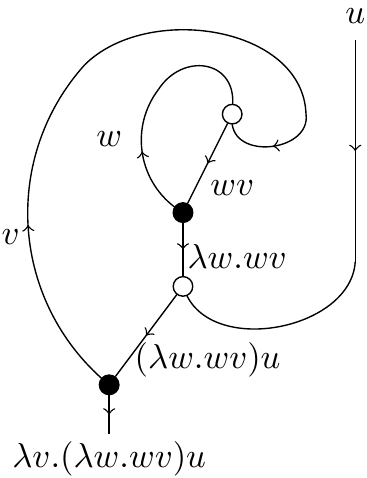}}
$$
The nice thing is that once we've defined the inductive procedure for translating linear lambda terms (i.e., decorated lambda skeletons) into string diagrams, we can work directly with the diagrams in a much more abstract way.
For instance, by the principle of $\alpha$-conversion, the following are also perfectly legal annotations of the above diagrams:
$$
\imgcenter{\includegraphics[scale=0.8]{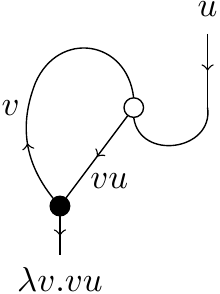}}
\qquad\qquad
\imgcenter{\includegraphics[scale=0.8]{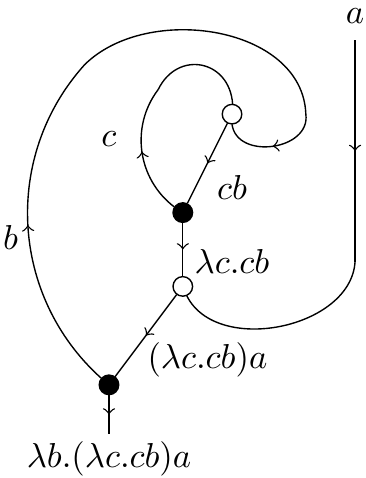}}
$$
By looking at the underlying wiring rather than the annotations (which are merely a guide for relating the diagrams to traditional syntax), we can represent lambda terms intrinsically up to renaming of variables.
Likewise, another important advantage of string diagrams is that substitution can be represented simply by plugging one diagram into another.
For example, the diagram 
$$
\imgcenter{\includegraphics[scale=0.5]{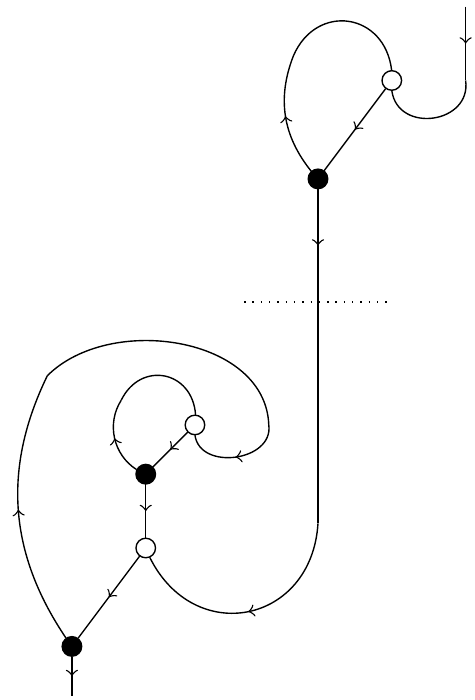}}
$$
represents the lambda term $[x]\lambda v.(\lambda w.wv)(\lambda y.yx)$ that results from substitution of the term $[x]\lambda y.yx$ for the variable $u$ in the term $[u]\lambda v.(\lambda w.wv)u$.\footnote{Note that these two advantages -- the representation of terms modulo $\alpha$-equivalence, and the ability to easily express substitution -- are well-known arguments for the use of HOAS in proof assistants.}

We conclude this section with a few historical and technical remarks:
\begin{enumerate}
\item The graphical language we have derived using the general mechanisms of string diagrams really is not much more than the linear fragment of the language of lambda-graphs (sometimes called ``lambda-graphs with back-pointers'' \cite{ariolablom97}).
In particular, the idea of representing lambda terms as directed graphs with explicit links from a variable to its binding lambda abstraction may be traced at least as far back as Richard Statman's thesis \cite{statman74phd}, if not further.\footnote{Pierre Lescanne (personal communication) notes that a similar convention appears in the opening chapter of Bourbaki's {\em Theory of Sets}, as a variable-free syntax for logical formulas involving quantifiers.}
Moreover, linear lambda-graphs have a very simple relationship with proof-nets for the implicative fragment of linear logic \cite{guerrini2001}.

\item 
However, something to emphasize is that string diagrams are \emph{not} simply special kinds of directed graphs (with two kinds of vertices, and with open edges representing inputs and outputs), but rather they are graphs drawn on the page, and so the order in which wires are positioned around $L$-nodes and $A$-nodes matters for determining planarity.
Observe that the three example diagrams we showed above are all planar diagrams in the sense that they involve no crossing wires, and indeed the three linear lambda terms
$$[x]\lambda y.yx
\qquad
[u]\lambda v.(\lambda w.wv)u
\qquad
[x]\lambda v.(\lambda w.wv)(\lambda y.yx)$$
are all planar in the sense of \Cref{defn:planar}.
On the other hand, the $\beta$-reduction of $[u]\lambda v.(\lambda w.wv)u$ results in a term $[u]\lambda v.uv$ which is not planar in the sense of \Cref{defn:planar}, and whose string diagram is not planar in the sense that it contains a crossing:
$$
\imgcenter{\includegraphics[scale=0.8]{\DIAGRAMS/non-normal-ann.pdf}}
\qquad=\qquad
\imgcenter{\includegraphics[scale=0.8]{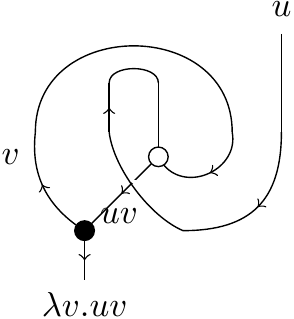}}
$$
Now, we need to be a bit careful here: if planarity is really a topological invariant of string diagrams, technically speaking it does not make sense for two diagrams to be equivalent when one is planar and the other is not.
But this just reflects the fact that $\beta$-reduction is naturally oriented, and our definition of a linear reflexive object (following the pattern of Scott's original definition of a reflexive object in a ccc) does not take that into account.
A more ``honest'' version of \Cref{defn:linref} would take place inside a higher category, so that the rule of $\beta$-reduction could be more faithfully described as an oriented cell
$$L;A \Rightarrow \turniso_{u,u}$$
rather than as an equation (cf.~\cite{seely1987}).
On the other hand, we do not need to pursue this additional level of sophistication here, because in the next section we will describe a simple way of restricting to only the string diagrams which represent $\beta$-normal lambda terms, such that this question does not even arise.
\item Since order matters, some alternative definitions of linear reflexive object would have given rise to a different notion of planarity.
For example, if we had asked for a pair of morphisms
$$
\xymatrix{
u \ar@<.5ex>[r]^{A''} & \resR[u]{u}\ar@<.5ex>[l]^{L''}
}
$$
such that $L'';A'' = \id_{\resR[u]{u}}$ (or $L'';A'' \Rightarrow \id_{\resR[u]{u}}$), then the various components would be drawn as follows:
\begin{equation}
\tag{RL}\label{rules:RL}
\imgcenter{\includegraphics[height=2.5cm]{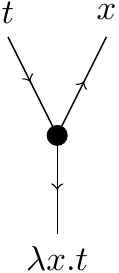}}
\qquad
\vrule
\qquad
\imgcenter{\includegraphics[height=2.5cm]{\DIAGRAMS/app-ann.pdf}}
\qquad
\vrule
\qquad
\imgcenter{\includegraphics[height=3cm]{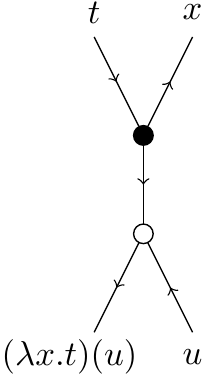}}
\quad=\quad
\imgcenter{\includegraphics[height=2.5cm]{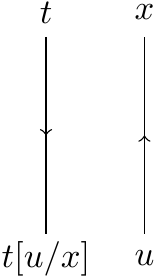}}
\end{equation}
In the literature on lambda-graphs, both of the alternative conventions (\ref{rules:LR}) and (\ref{rules:RL}) appear.
For example, Guerrini \cite{guerrini2001} displays $\beta$-reduction as a crossing, while Mairson \cite{mairson2002dilbert} displays it as a planar rewriting (and Buliga \cite{buliga2013GLC} considers both possibilities).
Notably, Abramsky \cite{abramsky-2009} has discussed a notion of planarity that coincides with the fragment of linear lambda terms whose string diagrams are planar by the (\ref{rules:RL}) convention.
For example, the two terms $[x]\lambda y.xy$ and $[x]x(\lambda y.y)$ are planar according to Abramsky's definition (``RL-planar''), whereas the (``LR-planar'') term $[x]\lambda y.yx$ is not RL-planar.

There is actually a trivial bijection between LR-planar terms and RL-planar terms, based on the fact that both are fully determined by their lambda skeletons.
Indeed, it is possible to adapt the algorithm described in \Cref{skeleton2planar:fig} to annotate a lambda skeleton with an RL-planar term, where the only modification needed is that instead of traversing applications left-to-right, they are traversed right-to-left (hence the mnemonics ``LR'' and ``RL'').
Therefore, at this level of abstraction, the choice of planarity convention might seem like just a matter of taste.
However, we believe the (\ref{rules:LR}) convention to be more natural when viewing planarity as a \emph{property} of linear lambda terms, rather than as defining an independent ``planar lambda calculus''.
We will provide some evidence for this view in \Cref{sec:tutte}, by showing that normal LR-planar lambda terms admit a computationally-natural analogue of \emph{Tutte decomposition}.
In particular, although the said bijection between LR-planar terms and RL-planar terms preserves the property of being $\beta$-normal (so that normal RL-planar terms are also in size-preserving bijection with rooted planar maps), it considerably changes the computational structure of terms.

\item
Finally, let us point out that not every possible string diagram composed out of $L$-nodes and $A$-nodes results in a valid linear lambda term.
For example, the diagram
\begin{center}\includegraphics{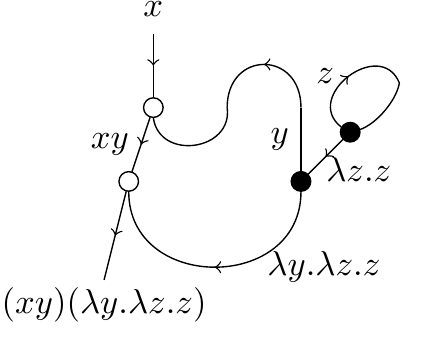}\end{center}
only represents a \emph{pseudo} lambda term,\footnote{Thanks to Ed Morehouse for this example.} which is \emph{ill-scoped} in the sense that the variable $y$ is used before it is bound by $\lambda y$.
This is a phenomenon which is well-known in the literature on lambda-graphs (see, e.g., the \emph{scoped lambda-graphs} of \cite{ariolablom97}), as well as in the literature on proof-nets (where it leads to so-called \emph{correctness criteria}).
In terms of the categorical semantics, this corresponds to the fact that arbitrary diagrams composed of $L$-nodes and $A$-nodes can be interpreted as morphisms in any compact closed category containing a linear reflexive object, but not every smcc is compact closed.
Again, though, this is perfectly fine for our purposes here, since we will always be able to verify that the diagrams we consider come from the interpretation of a linear lambda term.
\end{enumerate}

\subsection{A coloring protocol for neutral and normal terms}
\label{sec:diagrams:colored}

Now that we have explained the semantic basis of lambda-graphs as string diagrams for linear lambda terms, we will move more quickly in describing how to color these string diagrams to obtain a graphical representation of neutral and normal terms.
Our coloring protocol is again derived mechanically from a \emph{refinement} of the definition of a linear reflexive object.
\begin{definition}\label{defn:linrefpair}\emph{
A \definand{linear reflexive pair} in a smcc $\cal{D}$ is a pair of objects $B,R \in \cal{D}$ equipped with a quadruple of morphisms
$$
\xymatrix{\resL[R]{B} \ar[r]^-{\ell} & R\ar@<.5ex>[r]^c & B \ar@<.5ex>[l]^s\ar[r]^-a & \resR[B]{R}}
$$
such that $s;c = \id_B$ and $\ell;c;a = (\resL[c]{\id_B});\turniso_{b,b};(\resR[\id_B]{c})$.}
\end{definition}
This is actually a ``refinement'' in a technical sense: if one ignores the morphism $c : R \to B$ and associated equations (which we shall explain shortly), this is essentially a linear variation of the \emph{refinement type signature} originally presented by Pfenning in \cite{pfenning93types}.
There, he gave an elegant formulation of the standard inductive definition of neutral and normal lambda terms (which we recalled in \Cref{sec:genfun}), as a refinement of the higher-order abstract syntax representation of lambda terms.\footnote{In fact, he considered a slightly more sophisticated, dependently-typed HOAS representation of natural deduction proofs (which are isomorphic to simply-typed lambda terms), and its refinement to represent normal and neutral proofs. But the adaptation of the refinement type signature in \cite{pfenning93types} to the case of pure lambda calculus is straightforward.}
Our categorical reformulation is based on a functorial view of type refinement \cite{mz15popl}, the idea being that one should view a linear reflexive pair in some smcc $\cal{D}$ as living over a linear reflexive object in another smcc $\cal{C}$, equipped with a (smcc) functor $|{-}| : \cal{D} \to \cal{C}$ such that 
$$
|B| = |R| = u \qquad
|\ell| = L \qquad
|s| = \id_u = |c| \qquad
|a| = A.
$$
This definition of linear reflexive pair may also be compared to Melliès' definition of \emph{Frobenius pair} \cite{mellies2013frobenius}, which is a refinement of the notion of a Frobenius monoid.

\begin{figure}
$$
\imgcenter{\includegraphics[height=3cm]{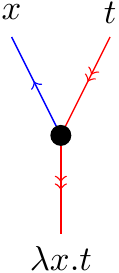}}
\quad
\vrule
\qquad
\imgcenter{\includegraphics[height=3cm]{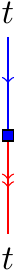}}
\qquad
\vrule
\qquad
\imgcenter{\includegraphics[height=3cm]{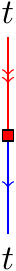}}
\qquad
\vrule
\quad
\imgcenter{\includegraphics[height=3cm]{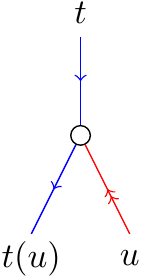}}
\quad\vrule\quad
\imgcenter{\includegraphics[scale=0.7]{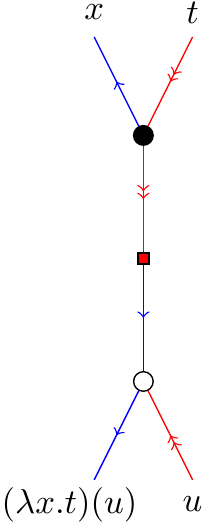}}
\quad=\quad
\imgcenter{\includegraphics[scale=0.7]{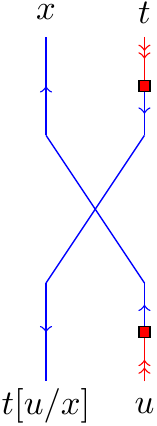}}
\qquad
\vrule
\qquad
\imgcenter{\includegraphics[scale=0.8]{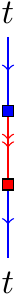}}
\quad=\quad
\imgcenter{\includegraphics[scale=0.8]{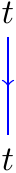}}
$$
\caption{Basic components and reduction rules representing a linear reflexive pair, annotated with their corresponding actions on lambda terms.}
\label{fig:cdiagrams}
\end{figure}

In \Cref{fig:cdiagrams}, we describe the graphical language that results from interpreting a linear reflexive pair in a compact closed category (reading the components left to right as $\ell$, $s$, $c$, and $a$, followed by the two equations).
The recipe is precisely analogous to the one we detailed in \Cref{sec:diagrams:scott}, but we comment on a few points:
\begin{enumerate}
\item The objects $B$ and $R$ are interpreted respectively as blue and red oriented wires.
To increase visual contrast and make the diagrams readable without color, we place an extra stroke on red wires.
\item 
$\ell$-nodes are a colored version of $L$-nodes, where reading counterclockwise the wires run as follows: outgoing-red, incoming-red, outgoing-blue.
\item 
$a$-nodes are a colored version of $A$-nodes, where reading counterclockwise the wires run as follows: incoming-blue, outgoing-blue, incoming-red.
\item The annotations are derived from the forgetful functor $|{-}| : \cal{D} \to \cal{C}$ described above. In particular, observe that $s$-nodes and $c$-nodes act as identity operations on lambda terms.
\end{enumerate}
Moreover, by a simple extension of the inductive procedure described in \Cref{sec:diagrams:scott}, any neutral or normal linear lambda term $[x_i,\dots,x_1]t$ may be assigned a morphism of the form $\enc{\pi} : B\mul\dots\mul B \to B$ or $\enc{\pi} : B\mul\dots\mul B \to R$ in any smcc with a linear reflexive pair, as well as a corresponding colored diagram of the form
$$
\imgcenter{\includegraphics[scale=0.8]{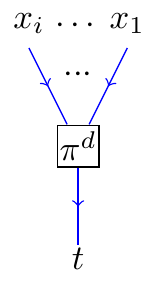}}
\qquad\text{or}\qquad
\imgcenter{\includegraphics[scale=0.8]{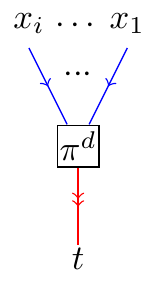}}.
$$
Diagrams of such neutral or normal terms have the additional property of containing no $c$-nodes, and in fact, \emph{all} of the diagrams that we consider below have no $c$-nodes -- so it is worth commenting on the presence of the operation $c : R \to B$ and its associated equations in \Cref{defn:linrefpair}.

\begin{figure}
\begin{center}
\includegraphics{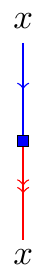}
\qquad\qquad\vrule\qquad\qquad
\includegraphics{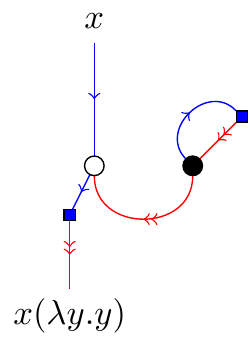}
\qquad
\includegraphics{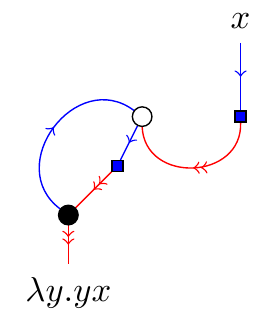}
\end{center}
\hrule
\begin{center}
\includegraphics{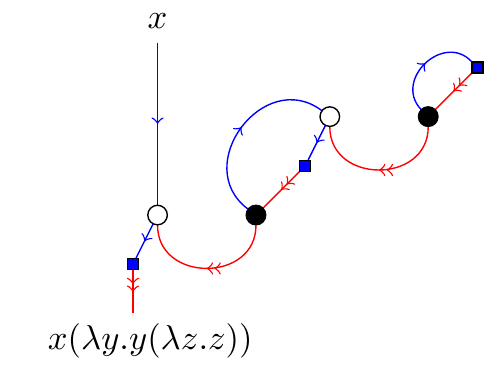}\qquad
\includegraphics{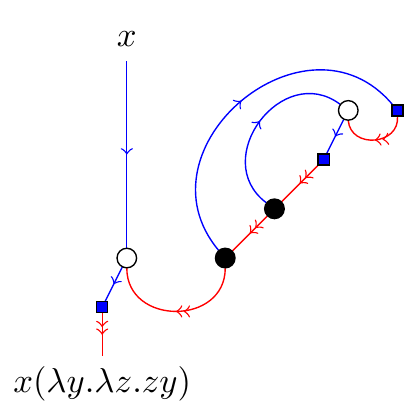}\qquad
\includegraphics{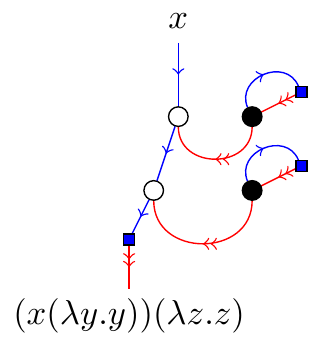}\qquad
\includegraphics{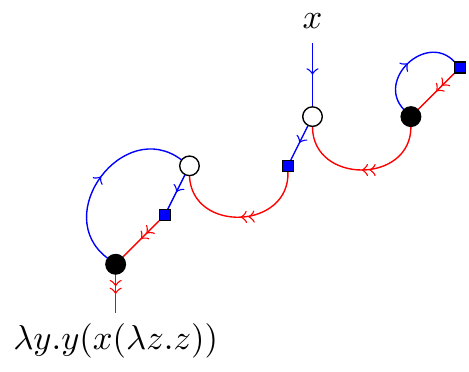}\qquad
\includegraphics{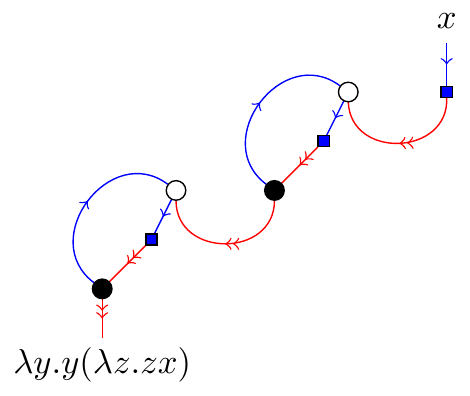}\qquad
\includegraphics{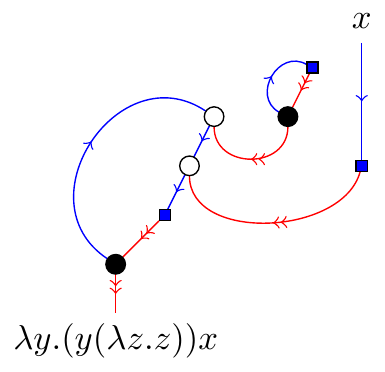}\qquad
\includegraphics{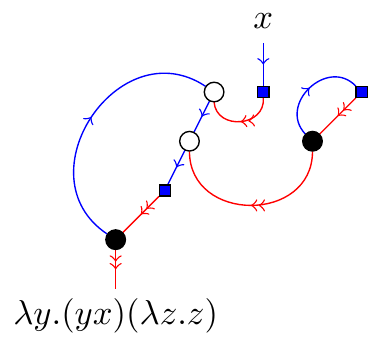}\qquad
\includegraphics{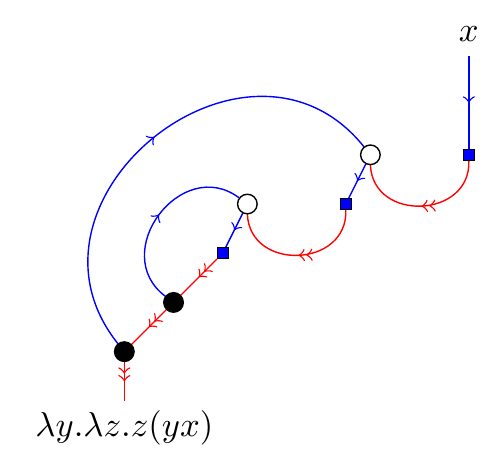}\qquad
\includegraphics{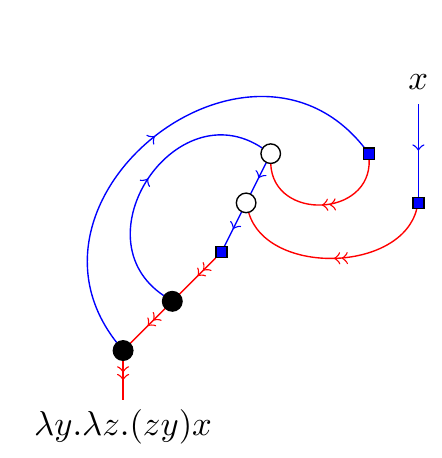}
\end{center}
\caption{String diagrams of the 12 normal planar lambda terms of size $\le 3$ (with one free variable). We annotate the incoming wire with the name of the free variable, and the outgoing wire with the represented lambda term.}
\label{fig:lam<=3}
\end{figure}

In proof theory and type systems, this technique is actually well-established
(cf.~\cite{polakow-pfenning99,daviespfenning00}): starting from a language
restricted to only neutral and normal terms, one can represent arbitrary terms
by adding ``virtual'' coercions from normal to neutral, which can then be
eliminated by a process analogous to cut-elimination for sequent calculus.
In the presence of $c$-nodes, any linear lambda term can be represented by many
different diagrams, but the normalization theorem for linear lambda calculus
implies that these can all be reduced to a unique $c$-node-free one.
However, we are not going to study normalization in this paper, and the only
reason we give the general definition of a linear reflexive pair and its
associated graphical language is because these are natural refinements of the
definition and associated graphical language of a linear reflexive object.
In the next section we will consider string diagrams representing normal linear
lambda terms, hence which are $c$-node-free, and moreover we will be mainly
interested in the planar case.
\Cref{fig:lam<=3} shows all such ($c$-node-free, planar) diagrams for the 12
normal planar lambda terms of size $\le 3$, while \Cref{sec:size4} shows all
diagrams for the 54 normal planar terms of size $= 4$.

\section{Relating normal planar lambda terms to rooted planar maps \\ via Tutte decomposition}
\label{sec:tutte}

In this section we give our main result, a size-preserving bijection between
normal planar lambda terms and rooted planar maps.
By ``normal planar lambda term'', we mean a linear lambda term which
\begin{enumerate}
\item is planar in the sense of \Cref{defn:planar},
\item is equipped with an $R$-coloring in the sense of \Cref{defn:coloring}, and
\item has one free variable.
\end{enumerate}
Our proof relies on an inductive characterization of rooted planar maps originally described by Tutte \cite{tutte1968}, so we begin by recalling his analysis in \Cref{sec:tuttemaps} (for another presentation, see Flajolet and Sedgewick \cite[VII.8.2]{flajolet-sedgewick}).
The idea will then be to reconstruct Tutte's analysis in the setting of linear lambda calculus, to obtain a ``parallel'' decomposition of normal planar lambda terms.
We show how to do this in \Cref{sec:decomp}, using both traditional lambda calculus notation (following the conventions of \Cref{sec:genfun}) as well the string diagrams of \Cref{sec:diagrams}.
Finally, in \Cref{sec:decomp:bij} we explain how to combine these parallel analyses to obtain a size-preserving bijection between rooted planar maps and normal planar lambda terms.

\subsection{Tutte decomposition of rooted planar maps}
\label{sec:tuttemaps}

A \emph{(topological) map} $M$ on a closed, oriented surface $S$~\cite{jones-singerman,landozvonkin} is a partition of $S$ into three finite sets of cells $V$, $E$, and $F$, such that:
 a \emph{vertex} $v\in V$ is a point of $S$,
 an \emph{edge} $e \in E$ is a simple open Jordan arc in $S$ whose extremities are vertices, and
 a \emph{face} $f \in F$ is a connected component of the complement of $V \cup E$ in $S$, homeomorphic to an open disk.

An edge equipped with one of two possible orientations is called a \emph{dart}.
Each dart $d$ has an \emph{opposite dart} $-d$, corresponding to the same edge with the opposite orientation.
The initial vertex (source) of a dart and the face to the left of a
dart are both said to be \emph{incident} to that dart. An \emph{isthmus} (resp.
\emph{loop}) is an edge whose two orientations are incident to the same face
(resp. \emph{vertex}). The \emph{degree} of a vertex or face counts the total
number of darts incident to that vertex or face (so that an edge is counted
twice in the degree of a face if it is an isthmus, and twice in the degree of a
vertex if it is a loop).

A \emph{planar map} is a map on the sphere. Every planar map $M$ has an
underlying graph which is a connected planar graph, possibly with loops and
multiple edges. A degenerate example of a planar map is the \emph{vertex map} --
containing a single vertex, no edges, and a single face -- while any other
planar map must contain at least one edge.
By definition, a \emph{rooting} of a planar map $M$ consists of a choice of a
dart, unless $M$ is the vertex map, in which case it is also considered rooted
by default.
Then, a \emph{rooted planar map} is a planar map equipped with a rooting,
treated up to root-preserving homeomorphism.
Tutte's analysis begins by noting that any rooted planar map $M$ can be
categorized into one of three possible classes:
\begin{enumerate}
\item $M$ is the vertex map.
\item $M$ has an \emph{isthmic} root: deleting the root edge separates the underlying graph into two connected components.
\item $M$ has a \emph{non-isthmic} root: the underlying graph remains connected when the root edge is deleted.
\end{enumerate}
The isthmic and non-isthmic cases are illustrated in \Cref{fig:tuttedecomp},
with root dart $A$ indicated by an arrow.
In these diagrams and more generally, we refer to the face incident to (i.e., to
the left of) $A$ as the \emph{outer face} of $M$ (following the convention that
we always draw rooted planar maps on the page with the ``infinite'' face to the
left of the root).
\begin{figure}
\begin{subfigure}[b]{0.33\textwidth}
\begin{center}
\includegraphics[width=0.9\textwidth]{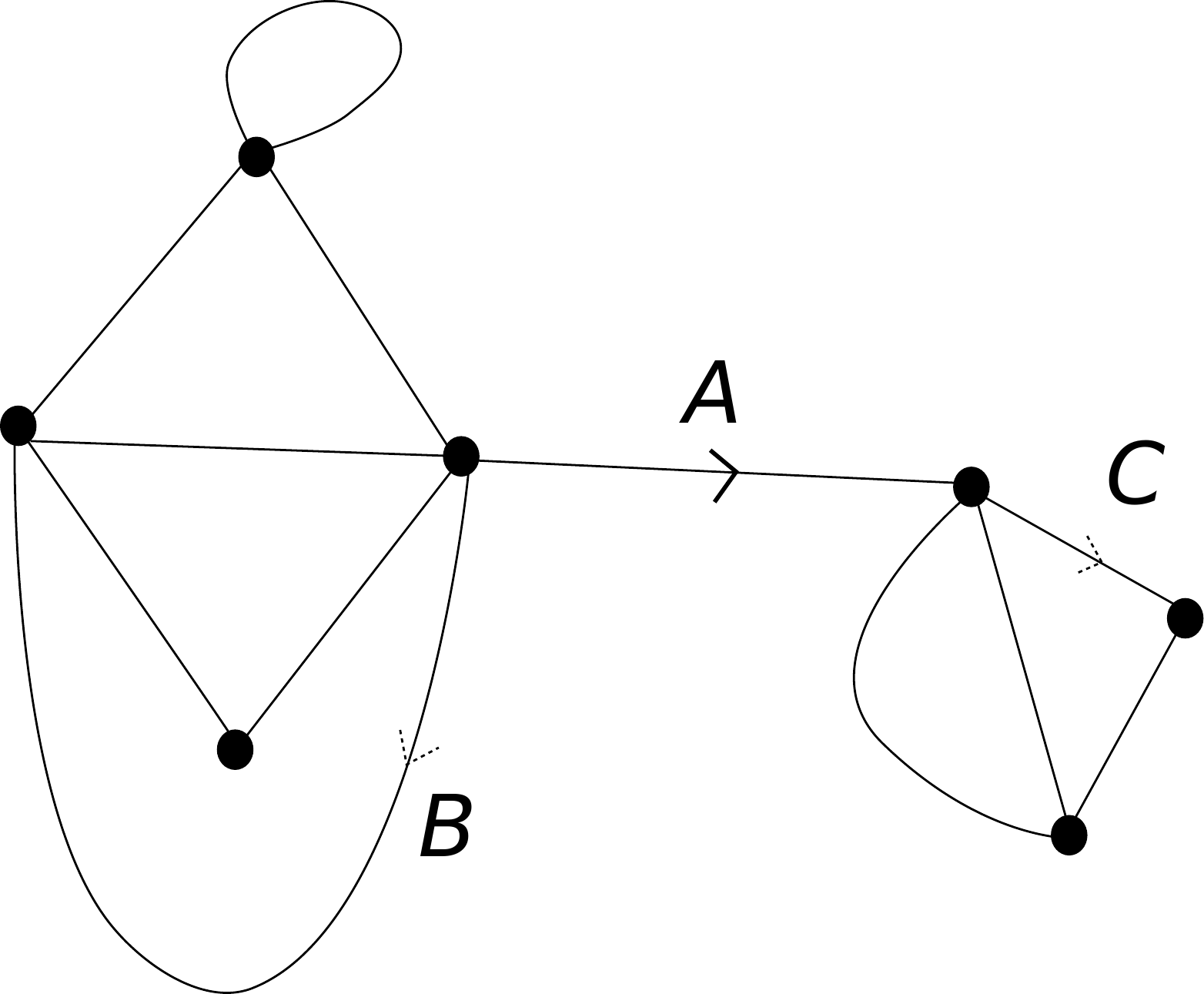}
\end{center}
\caption{Map with isthmic root.}
\label{fig:isthmic}
\end{subfigure}
\begin{subfigure}[b]{0.33\textwidth}
\begin{center}\includegraphics[width=0.9\textwidth]{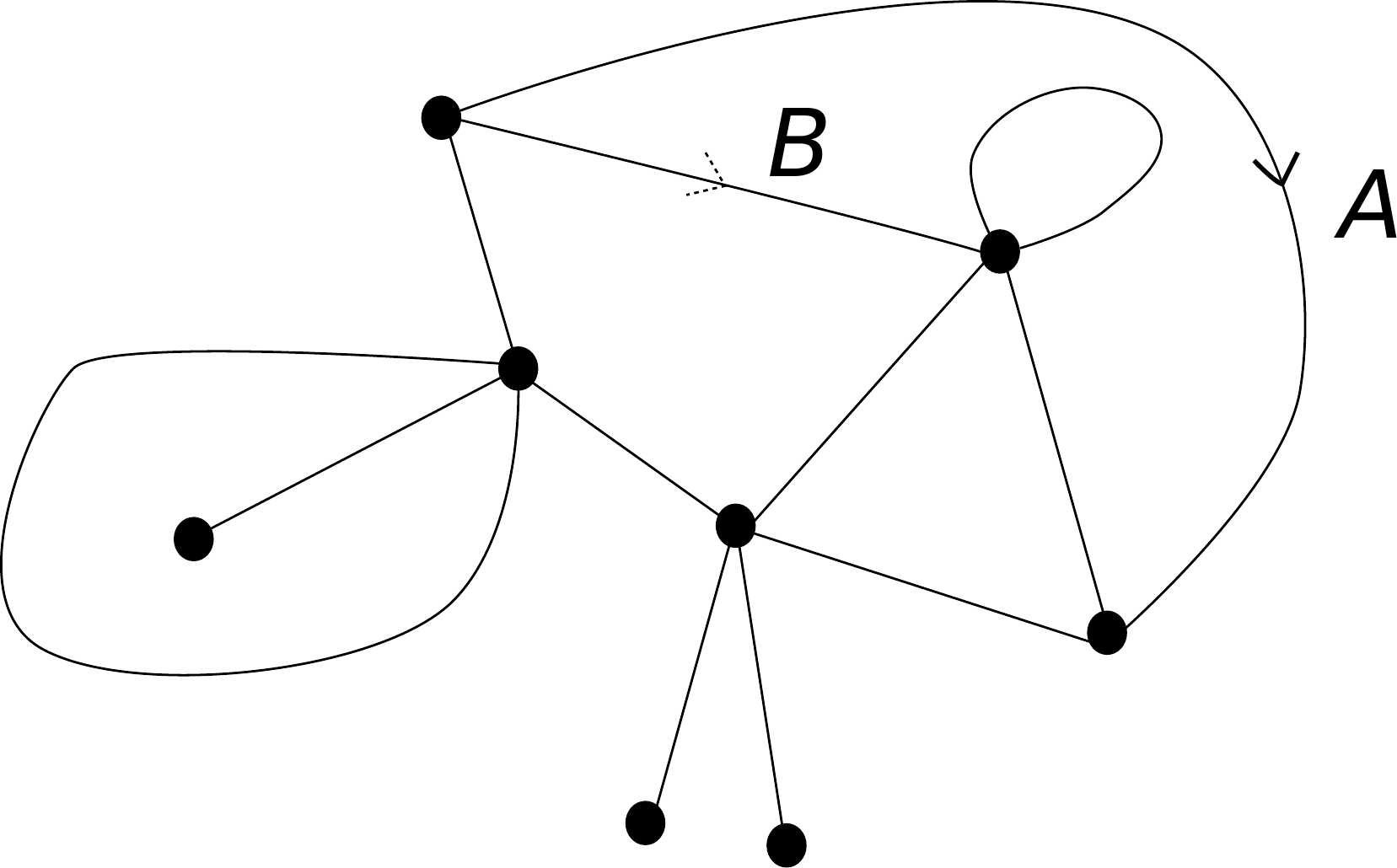}\end{center}
\caption{Map with non-isthmic root.}
\label{fig:nonisthmic1}
\end{subfigure}
\begin{subfigure}[b]{0.32\textwidth}
\begin{center}\includegraphics[width=0.86\textwidth]{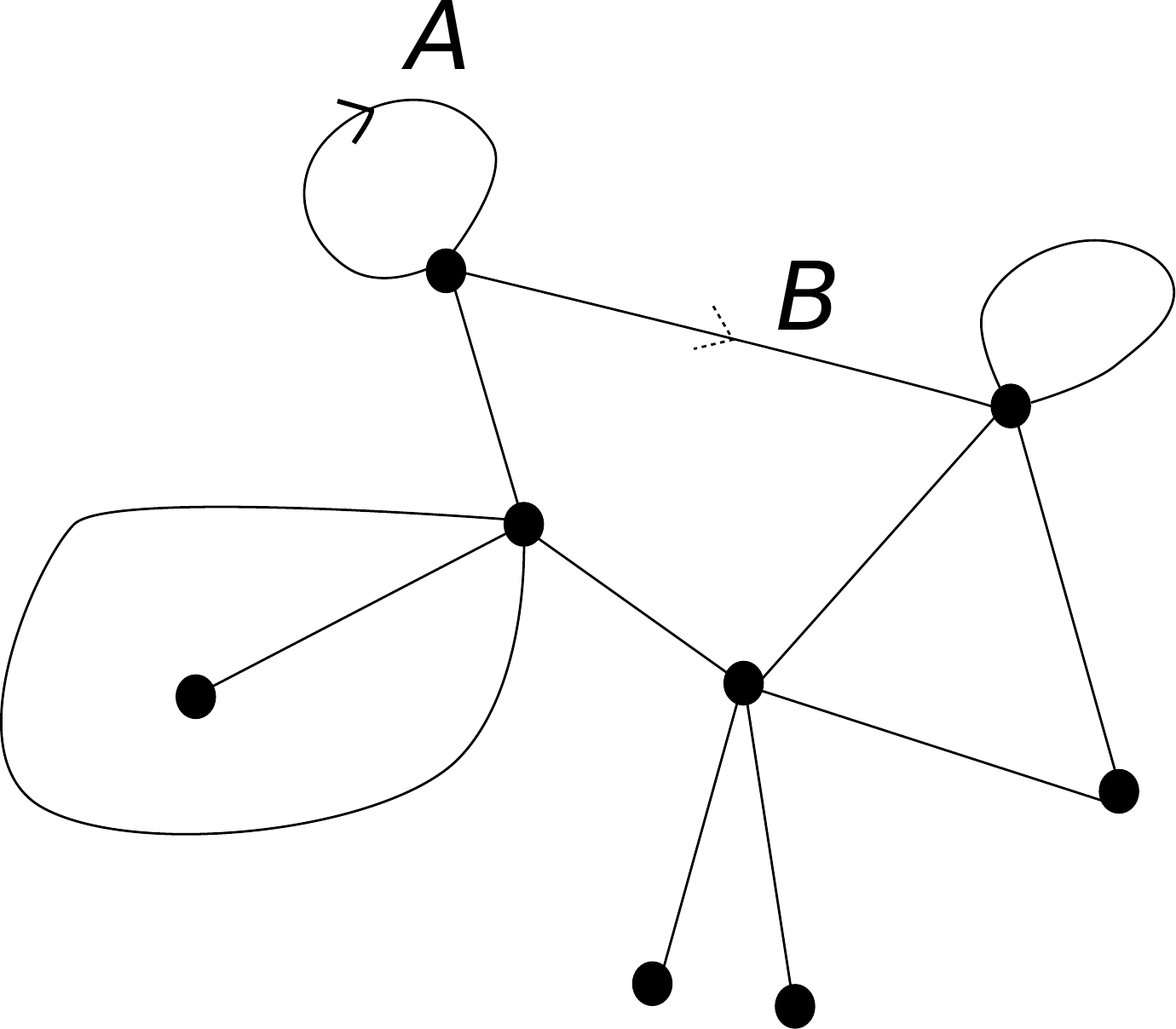}\end{center}
\caption{Map with non-isthmic root.}
\label{fig:nonisthmic2}
\end{subfigure}
\caption{Examples of rooted planar maps (root marked $A$).}
\label{fig:tuttedecomp}
\end{figure}
Tutte's analysis goes on to describe how any map which is not a vertex map may be decomposed in terms of smaller rooted planar maps, by a deterministic procedure:
\begin{enumerate}[label=\({\alph*}]
\item
{\bf Isthmic root.}
Let $M_1$ and $M_2$ be the two planar maps resulting from deleting the isthmic
root $A$. Each of $M_1$ and $M_2$ is either the vertex map, or else may be
rooted by walking along the outer face of $M$ and choosing respectively the dart
immediately following $-A$ and $A$ (marked $B$ and $C$ in \Cref{fig:isthmic}).
\item {\bf Non-isthmic root.}
Let $M_1$ be the planar map resulting from deleting the non-isthmic root $A$.
Again, if it is the vertex map then we are done, and otherwise $M_1$ 
can be rooted by taking the dart (marked $B$ in
\Cref{fig:nonisthmic1}) immediately following $-A$ when walking along its incident face, unless that dart is $-A$ itself, in which case we take the dart  (marked $B$ in \Cref{fig:nonisthmic2}) immediately following $A$ along the outer face.
\end{enumerate}
Note that Tutte \cite{tutte1968} used slightly different conventions for rooting submaps than what we describe here, but for the purpose of counting rooted planar maps (as was Tutte's original application), the precise convention used does not matter so long as it is deterministic.

Now, let us view (a) as an operation $\decomposIsthm$ taking a planar map $M$ with an
isthmic root as input, and decomposing it into a pair of rooted planar maps
$M_1$ and $M_2$.
There is clearly a reverse operation $\composIsthm$, which given any pair of rooted
planar maps $M_1$ and $M_2$ joins them together to create a rooted planar map with an
isthmic root (note that this binary operation is ``anti-commutative'', in the sense that swapping the arguments reverses the orientation of the root dart).
We have
$$
\decomposIsthm(\composIsthm(M_1,M_2)) = (M_1,M_2)
$$
for any pair of rooted planar maps $M_1$ and $M_2$, and conversely
$$
M = \composIsthm(\decomposIsthm(M))
$$
for any planar map $M$ with an isthmic root.
Moreover, we have that the number of edges in $\composIsthm(M_1,M_2)$ is equal
to one plus the sum of the numbers of edges in $M_1$ and $M_2$, and that the
degree of the outer face of $\composIsthm(M_1,M_2)$ is equal to two plus the sum
of the degrees of the outer faces of $M_1$ and $M_2$ (since the outer face is incident to both the root dart and the opposite dart: see \Cref{fig:isthmic}, which
shows a map with outer face degree nine, constructed from two maps with outer
face degrees four and three, respectively).

Similarly, let us view (b) as an operation $\decomposNonIsthm$ taking a rooted
planar map $M$ with a non-isthmic root as input, and producing a rooted planar
map $M_1$ as output. Going in the other direction, there is a \emph{family} of
operations $\composNonIsthm{k}$, which given any rooted planar map $M_1$ with
outer face of degree $\ge k$, constructs a rooted planar map with a non-isthmic
root as follows:
starting at the source vertex $x$ of the root dart, walk backwards (i.e. in the
opposite direction of the root dart) $k$ darts along the outer face of $M_1$
until reaching a vertex $y$, and then add a new edge between $x$ and $y$ with
the new root dart oriented from $x$ to $y$, in such a way that the border of the
new root face is composed in this order of the new root dart and the sequence of
$k$ darts composing the reverse walk.
(For example, the maps in \Cref{fig:nonisthmic1,fig:nonisthmic2} may be
constructed as $\composNonIsthm{8}(M_1)$ and $\composNonIsthm{11}(M_1)$,
respectively, for the same underlying rooted planar map $M_1$ of outer face degree 11.)

 We have that $$\decomposNonIsthm(\composNonIsthm{k}(M_1)) =
M_1$$ for all rooted planar maps $M_1$ and $k$ bounded by the degree of the
outer face of $M_1$, and conversely, that there \emph{exists} a $k$ such that
$$M = \composNonIsthm{k}(\decomposNonIsthm(M))$$ for any rooted planar map $M$ with a non-isthmic root.
Moreover, we have that the number of edges in $\composNonIsthm{k}(M_1)$ is equal to one
plus the number of edges in $M_1$, and that the degree of the outer face of
$\composNonIsthm{k}(M_1)$ is $k+1$.

This combination of observations yields a complete characterization of rooted planar maps, by induction on the number of edges:
\begin{thm}[Tutte \cite{tutte1968}] \label{thm:tutte}
Let $M$ be a rooted planar map with $e(M)$ edges and outer face degree $o(M)$.
Then exactly one of the following cases must hold:
\begin{enumerate}[label=(\roman*)]
\item $M$ is the vertex map and $e(M) = o(M) = 0$.
\item $M = \composIsthm(M_1,M_2)$ for some $M_1$ and $M_2$ such that $e(M) = 1 + e(M_1) + e(M_2)$ and $o(M) = 2 + o(M_1) + o(M_2)$.
\item $M = \composNonIsthm{k}(M_1)$ for some $M_1$ and $0 \le k \le o(M_1)$ such that $e(M) = 1 + e(M_1)$ and $o(M) = k+1$.
\end{enumerate}
\end{thm}

\subsection{Decomposition of normal planar lambda terms}
\label{sec:decomp}

Here and below, we write ``NLT'' and ``NPT'' as abbreviations for ``normal
linear term'' and ``normal planar term'', respectively, it being implicit that
we always consider lambda terms with exactly one free variable, unless otherwise stated.
(Note that a normal lambda term with an arbitrary number of free variables can
always be seen as one with exactly one free variable, by adding or removing
leading $\lambda$s.)

\subsubsection{A trichotomy on normal linear terms}
\label{sec:decomp:tri}

NLTs (which may or may not be planar) are naturally partitioned into three classes, depending on how the free variable is used.
\begin{definition}\emph{
\label{defn:trichotomy}
We say that $[x]t$ is \definand{the identity term} if $t = x$, that it is \definand{function-open} if $x(u)$ is a subterm of $t$ (for some $u$), and that it is \definand{value-open} if $u(x)$ is a subterm of $t$ (for some $u$).}
\end{definition}
\begin{proposition}\label{prop:trichotomy}
Every NLT is either the identity term, function-open, or value-open (mutually exclusively).
\end{proposition}
\begin{proof}
Immediate by induction, after generalizing the induction hypothesis to consider linear lambda terms with an arbitrary number of free variables.
Observe that in pure lambda calculus we have terms such as $[x]\lambda y.x$ which fail to fall into any of these classes, or terms such as $[x]\lambda y.y(xx)$ which classify as both function-open and value-open, but all such counterexamples are ruled out by the requirement that every variable is used exactly once.
\end{proof}
It is informative to restate this classification in terms of linear reflexive pairs and string diagrams.
\begin{proposition}\label{prop:classmorph}
Let $[x]t$ be a NLT, let $\enc{\pi} : B \to R$ be its interpretation in an smcc with a linear reflexive pair, and let $\pi^d$ be its corresponding string diagram.
\begin{enumerate}
\item 
If $[x]t$ is the identity term then $\enc{\pi} = s$ and $\pi^d = \imgcenter{\includegraphics[scale=0.75]{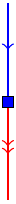}}$ .
\item
If $[x]t$ is function-open then $\enc{\pi} = (a;f)$ for some $f : \resR[B]{R} \to R$ and $\pi^d = \imgcenter{\includegraphics[scale=0.75]{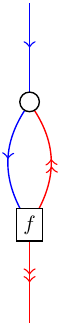}}$ .
\item
If $[x]t$ is value-open then $\enc{\pi} = (s;\lc{a;\plugR};f)$ for some $f : \resL[B]{B} \to R$ and $\pi^d = \imgcenter{\includegraphics[scale=0.75]{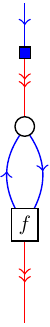}}$ .
\end{enumerate}
\end{proposition}
\Cref{prop:trichotomy,prop:classmorph} apply to arbitrary NLTs, and hence in particular to NPTs.
For example, consider once again the diagrams of the first three NPTs:
\begin{center}
\includegraphics{\DIAGRAMS/lam_1_0.pdf}
\qquad\qquad
\includegraphics{\DIAGRAMS/lam_2_0.pdf}
\qquad\qquad
\includegraphics{\DIAGRAMS/lam_2_1.pdf}
\end{center}
Here we can see that the leftmost diagram corresponds to the identity term, the middle diagram to a function-open term, and the rightmost diagram to a value-open term.
Reading across the rows of the bottom half of \Cref{fig:lam<=3}, we can quickly check that among the nine NPTs of size = 3, the first four are function-open and the next five are value-open.

We now consider how to further decompose the function-open and value-open classes, in the case where the NLT is planar.

\subsubsection{The planar function-open case}
\label{sec:decomp:fn}

\begin{proposition}\label{prop:funopen}
  Let $[x]t$ be a NPT, let $\sem{\pi} : B \to R$ be its interpretation in an smcc with a linear reflexive pair, and let $\pi^d$ be its corresponding string diagram.
If $[x]t$ is function-open, then $\enc{\pi} = (a;(\id_{{\resR[B]{R}}}\mul (\lc{h};\ell));\plugR;g)$ for some $g : B \to R$ and $h : B \to R$, and $\pi^d = \imgcenter{\includegraphics[scale=1]{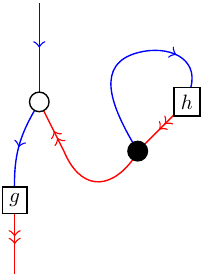}}$ .
\end{proposition}
\begin{proof}
We get half of the factorization by \Cref{prop:classmorph}(2).
For the second half, we reason that since $x$ is the \emph{only} free variable in $t$, the constraint of planarity forces the argument of $x$ to be a \emph{closed} NPT, hence of the form $\lambda y.u$ for some $u$ with one free variable $y$.
\end{proof}
\Cref{prop:funopen} suggests a natural way of performing surgery on the diagram of a function-open NPT, discarding a small piece of $\pi^d$ to obtain a pair of diagrams with the same interface:
\begin{equation}\tag{FO}\label{eqn:FO}
\imgcenter{\includegraphics{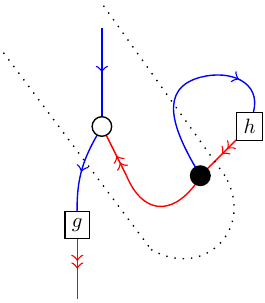}}
\qquad\fundecomp\qquad
\imgcenter{\includegraphics{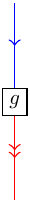}}
\quad+\quad
\imgcenter{\includegraphics{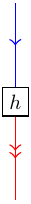}}
\end{equation}
Moreover, this operation is clearly reversible: given any pair of diagrams with one incoming blue wire and one outgoing red wire, we can join them together to obtain a diagram of the original shape:
\newcommand\mkfun{\mathbin{\overset{+}{\imgcenter{\includegraphics[width=4em]{\DIAGRAMS/brest.pdf}}}}}
$$
\imgcenter{\includegraphics{\DIAGRAMS/g.pdf}}
\qquad\mkfun\qquad
\imgcenter{\includegraphics{\DIAGRAMS/h.pdf}}
\qquad=\qquad
\imgcenter{\includegraphics{\DIAGRAMS/isthmic.pdf}}
$$
Finally, if we add annotations to the diagrams,
$$
\imgcenter{\includegraphics{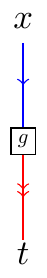}}
\qquad\mkfun\qquad
\imgcenter{\includegraphics{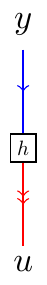}}
\qquad=
\imgcenter{\includegraphics{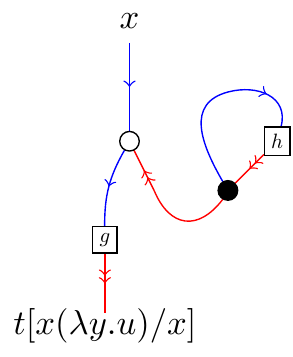}}
$$
then we can verify that this operation is easily implemented on NPTs: given a pair of NPTs $[x]t$ and $[y]u$, we produce a new function-open NPT by replacing $x$ with $x(\lambda y.u)$ in $t$.
For example, combining $[x]x(\lambda y.y)$ and $[x]\lambda y.yx$ in either order yields the following two function-open NPTs (after some renaming of variables):
$$
\imgcenter{\includegraphics[scale=0.9]{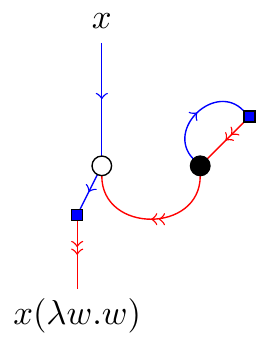}}
\qquad\mkfun\qquad
\imgcenter{\includegraphics[scale=0.9]{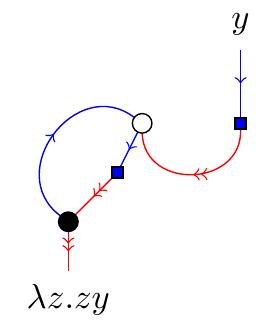}}
\quad=\hspace*{-3em}
\imgcenter{\includegraphics[scale=0.9]{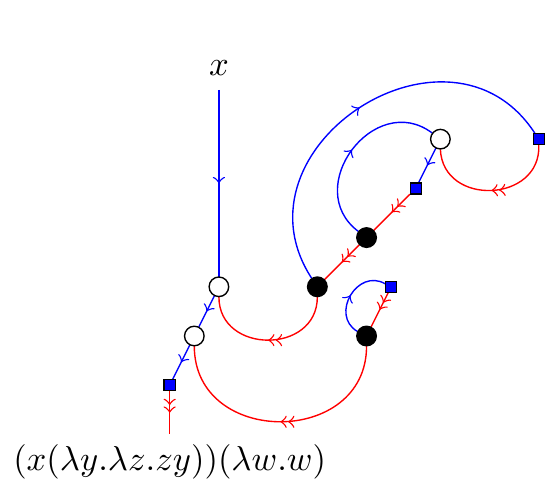}}
$$

$$
\imgcenter{\includegraphics[scale=0.9]{\DIAGRAMS/lam_2_1.pdf}}
\qquad\mkfun\quad
\imgcenter{\includegraphics[scale=0.9]{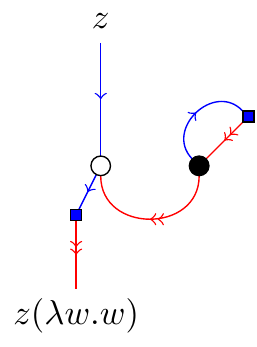}}
\quad=\hspace*{-2em}
\imgcenter{\includegraphics[scale=0.9]{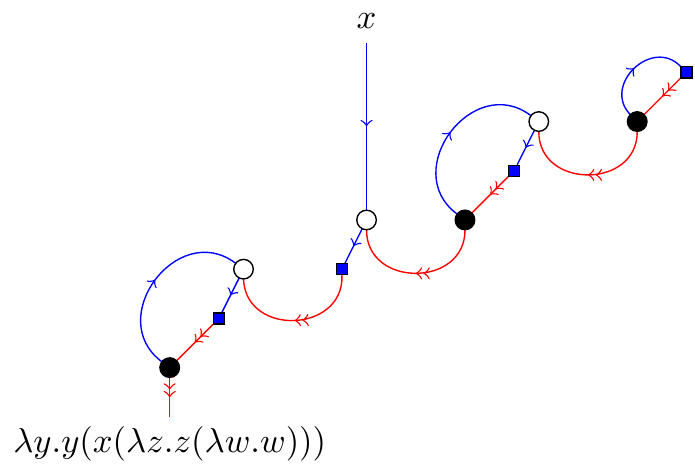}}
$$


\subsubsection{The planar value-open case}
\label{sec:decomp:arg}

We begin with an easy observation that holds in the general value-open case.
\begin{proposition}\label{prop:valopen:lambda}
If $[x]t$ is a value-open NLT then it must begin with a lambda abstraction, i.e., there exists a normal linear term (with two free variables) $[y_1,x]t'$ such that $t = \lambda y_1.t'$.
\end{proposition}
\begin{proof}
Let $[\Gamma]u$ be the neutral body of $[x]t$, in the sense of \Cref{defn:headvar}.
By construction, $\Gamma$ must be of the form $\Gamma = y_i,\dots,y_1,x$ for some $i \ge 0$.
But since $[x]t$ is value-open, $x$ cannot be applied in $u$, and hence the head variable of $u$ (again in the sense of \Cref{defn:headvar}) is necessarily distinct from $x$.
This implies that $i>0$, and the proposition follows.
\end{proof}
Combining \Cref{prop:valopen:lambda} with \Cref{prop:classmorph}, we obtain the following characterization of value-open NLTs.
\begin{proposition}\label{prop:argopen}
Let $[x]t$ be a NLT, let $\sem{\pi} : B \to R$ be its interpretation in an smcc with a linear reflexive pair, and let $\pi^d$ be its corresponding string diagram.
If $[x]t$ is value-open, then $\enc{\pi} = (s;\lc{a;\plugR};g;\ell)$ for some $g : \resL[B]{B} \to \resL[R]{B}$, and $\pi^d = \imgcenter{\includegraphics[scale=0.8]{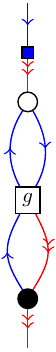}}$ .
\end{proposition}
Based on this knowledge, here is how we can perform surgery on the diagram of a any value-open NLT to obtain a new diagram with one incoming blue wire and one outgoing red wire:
\begin{equation}\tag{VO}\label{eqn:VO}
\imgcenter{\includegraphics{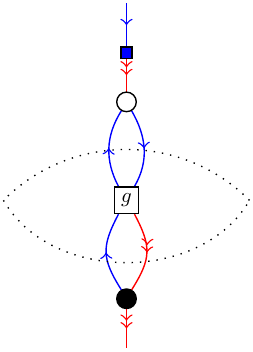}}
\qquad\valdecomp\qquad
\imgcenter{\includegraphics{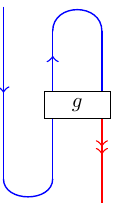}}
\end{equation}
In words, the surgery consists of first removing an $s$-node, an $a$-node, and an $\ell$-node to leave a diagram with four dangling wires, then splicing the two blue wires at the top together, and finally wrapping the blue wire at the bottom back up to the top.
In terms of the categorical semantics, surgery (\ref{eqn:VO}) corresponds to extracting the morphism $g : \resL[B]{B} \to \resL[R]{B}$ given by \Cref{prop:argopen}, pre-composing it with the currying of the identity morphism on $B$,
$$
\xymatrix{I \ar[r]^-{\lc{\id_B}} & \resL[B]{B} \ar[r]^g & \resL[R]{B}}
$$
and then uncurrying to obtain a morphism $B \to R$.
To borrow terminology from the theory of programming languages, this combined operation can be described as ``plugging $g$ with the identity continuation'', and we will notate it below by $\plug{g}{\id_B}$.
Finally, in terms of \Cref{prop:valopen:lambda}, the surgery has the effect of simply removing the outermost ``$\lambda y_1$'' and the application to $x$.

Although the (\ref{eqn:VO}) surgery works for any NLT, it is clear that it preserves the planarity of the original diagram.
For example, here is a demonstration of surgery on the diagram of a value-open NPT of size 6 (yielding a function-open NPT of size 5):
$$
\imgcenter{\includegraphics[scale=0.75]{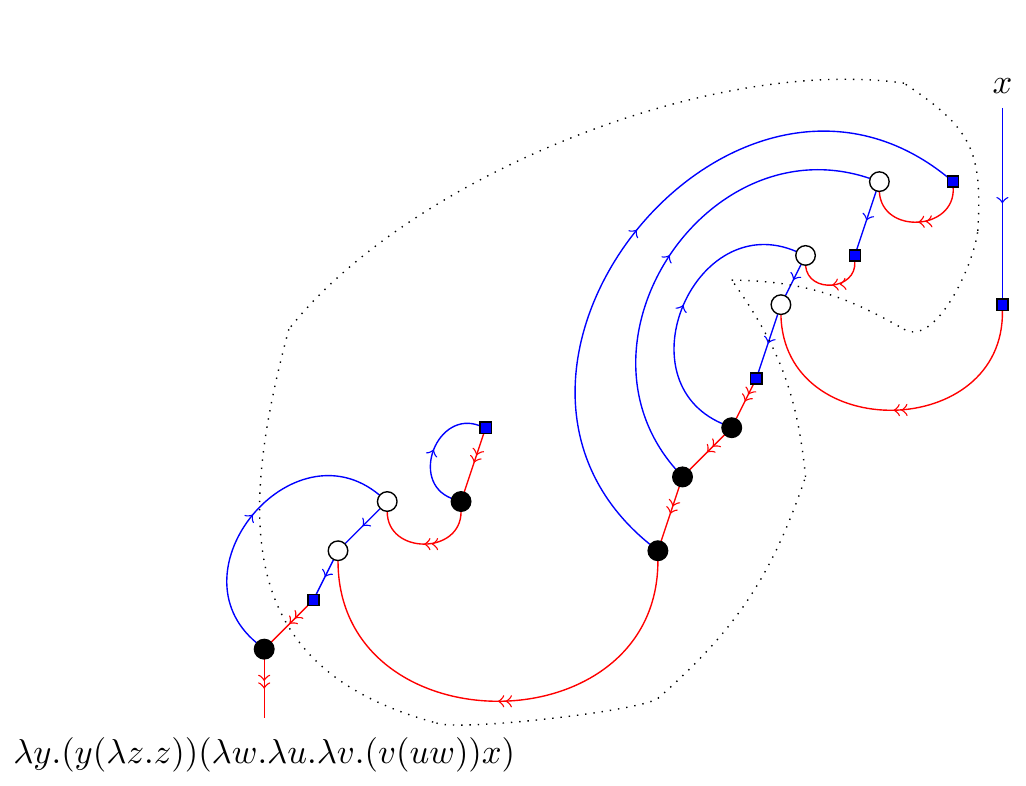}}
\qquad\overset{\leadsto}{\imgcenter{\includegraphics[height=3em]{\DIAGRAMS/urest2.pdf}}}\hspace*{-2em}
\imgcenter{\includegraphics[scale=0.75]{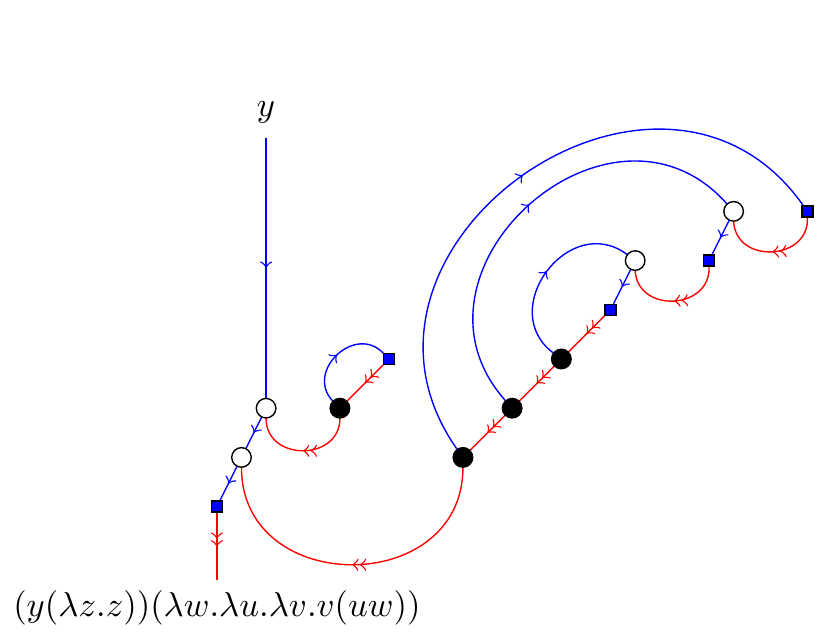}}
$$
In general, the transformation (\ref{eqn:VO}) is \emph{not} reversible: given a diagram with one incoming blue wire and one outgoing red wire representing a morphism $h : B \to R$,
in order to invert the surgery we have to \emph{choose} a particular factorization $h = \plug{g}{\id_B}$ of $h$ as a morphism $g : \resL[B]{B} \to \resL[R]{B}$ plugged with the identity continuation:
$$
\imgcenter{\includegraphics{\DIAGRAMS/h.pdf}} \quad
= \quad \imgcenter{\includegraphics{\DIAGRAMS/gp.pdf}}
$$
Topologically, such a factorization can be seen as grabbing a blue wire somewhere inside the diagram of $h$ as a ``handle'', and pulling it to the outside.
However, in order to produce a diagram representing a normal planar term, we do not want to consider \emph{arbitrary} factorizations $h = \plug{g}{\id_B}$, but only those factorizations in which the diagram of $g$ is also planar, so that the result of inverting (\ref{eqn:VO}) will be another NPT.
A bit of geometric reasoning convinces us that such factorizations should correspond precisely to \emph{blue wires incident to the outer region of the diagram}, where (by analogy to maps) we say that an oriented wire is incident to a region if it has that region to the left, and by ``outer region'' we mean the open half-plane incident to the incoming and outgoing wires.
We shall refer to such blue wires as the \definand{outer neutral handles} of the NPT, deferring a more formal description to \Cref{defn:valence} below.

Consider again the function-open NPT which resulted from value-open surgery above:
$$\imgcenter{\includegraphics[scale=0.8]{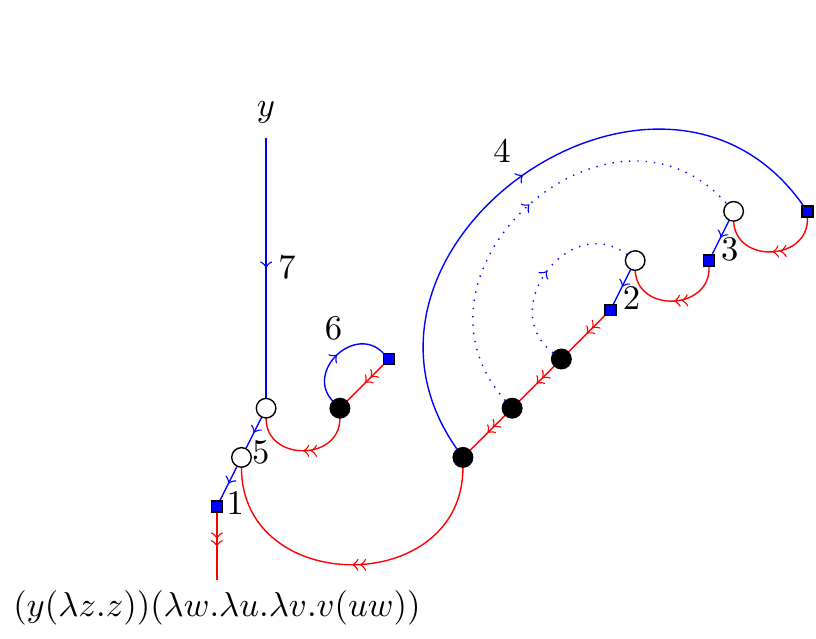}}$$
Here we have numbered all of the outer neutral handles (while dotting out the remaining blue wires), starting from the bottom of the diagram and walking backwards (i.e., counterclockwise) along the outer region.
Each outer neutral handle corresponds to a way of factoring the term by ``focusing'' on a neutral subterm:
\begin{enumerate}
\item $[y]\focus{(y(\lambda z.z))(\lambda w.\lambda u.\lambda v.v(uw))}$
\item $[y](y(\lambda z.z))(\lambda w.\lambda u.\lambda v.\focus{v(uw)})$
\item $[y](y(\lambda z.z))(\lambda w.\lambda u.\lambda v.v(\focus{uw}))$
\item $[y](y(\lambda z.z))(\lambda w.\lambda u.\lambda v.v(u\focus{w}))$
\item $[y](\focus{y(\lambda z.z)})(\lambda w.\lambda u.\lambda v.v(uw))$
\item $[y](y(\lambda z.\focus{z}))(\lambda w.\lambda u.\lambda v.v(uw))$
\item $[y](\focus{y}(\lambda z.z))(\lambda w.\lambda u.\lambda v.v(uw))$
\end{enumerate}
Performing an inverse (\ref{eqn:VO}) operation while focused on any of these subterms yields a different value-open NPT:
\begin{flalign*}
\quad\ \ \ 
t_1 &= [x]\lambda y.((y(\lambda z.z))(\lambda w.\lambda u.\lambda v.v(uw)))x & \\
t_2 &= [x]\lambda y.(y(\lambda z.z))(\lambda w.\lambda u.\lambda v.(v(uw))x) & \\
t_3 &= [x]\lambda y.(y(\lambda z.z))(\lambda w.\lambda u.\lambda v.v((uw)x)) & \\
t_4 &= [x]\lambda y.(y(\lambda z.z))(\lambda w.\lambda u.\lambda v.v(u(w x))) & \\
t_5 &= [x]\lambda y.((y(\lambda z.z))x)(\lambda w.\lambda u.\lambda v.v(uw)) & \\
t_6 &= [x]\lambda y.(y(\lambda z.zx))(\lambda w.\lambda u.\lambda v.v(uw)) & \\
t_7 &= [x]\lambda y.((y x)(\lambda z.z))(\lambda w.\lambda u.\lambda v.v(uw)) &
\end{flalign*}
In turn, performing (\ref{eqn:VO}) on any of the $t_k$ (i.e., erasing $x$ and removing the leading $\lambda y$) yields back the original function-open NPT.
Observe that certain ways of focusing on a neutral subterm are excluded because they do not correspond to \emph{outer} neutral handles.
For example, attempting to perform an inverse (\ref{eqn:VO}) operation starting from the factorization
$$[y](y(\lambda z.z))(\lambda w.\lambda u.\lambda v.v(\focus{u}w))$$
would indeed result in a value-open NLT,
$$[y] (y(\lambda z.z))(\lambda w.\lambda u.\lambda v.v(\focus{u}w))
\quad\leadsto\quad
[x]\lambda y.(y(\lambda z.z))(\lambda w.\lambda u.\lambda v.v((ux)w))
$$
but one which is not planar.
In \Cref{fig:valences} we give another example of a NPT with outer neutral handles indicated, as well as the associated value-open NPTs that arise by inverting (\ref{eqn:VO}).

\begin{figure}
\begin{minipage}{0.45\textwidth}
\begin{center}\includegraphics[scale=0.7]{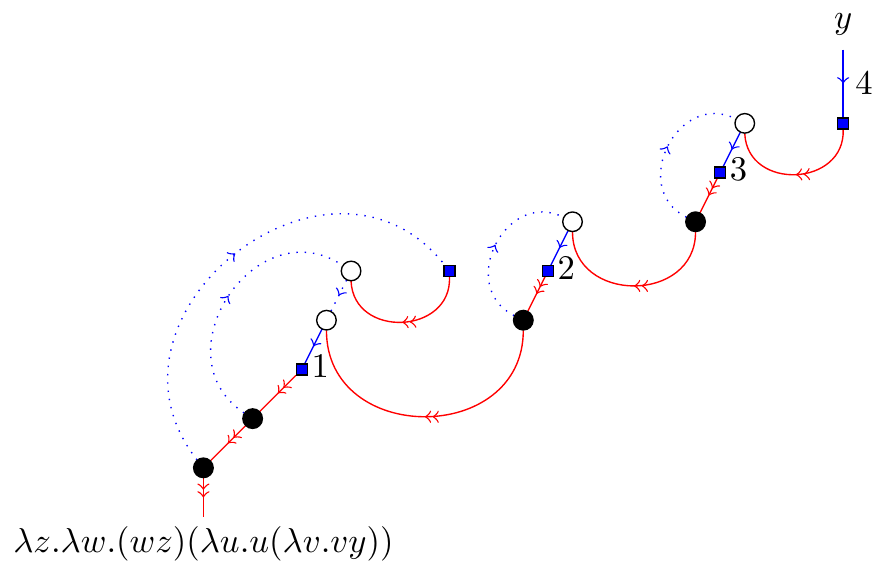}\end{center}
\end{minipage}
\vrule
\quad\begin{minipage}{0.45\textwidth}
\begin{enumerate}\small
\item $[y]\lambda z.\lambda w.\focus{(wz)(\lambda u.u(\lambda v.vy))}$
\item $[y]\lambda z.\lambda w.(wz)(\lambda u.\focus{u(\lambda v.vy)})$
\item $[y]\lambda z.\lambda w.(wz)(\lambda u.u(\lambda v.\focus{vy}))$
\item $[y]\lambda z.\lambda w.(wz)(\lambda u.u(\lambda v.v\focus{y}))$
\end{enumerate}
\begin{flalign*}
\ t_1 &=  [x]\lambda y.\lambda z.\lambda w.((wz)(\lambda u.u(\lambda v.vy)))x & \\
t_2 &=  [x]\lambda y.\lambda z.\lambda w.(wz)(\lambda u.(u(\lambda v.vy))x)& \\
t_3 &=  [x]\lambda y.\lambda z.\lambda w.(wz)(\lambda u.u(\lambda v.(vy)x))& \\
t_4 &=  [x]\lambda y.\lambda z.\lambda w.(wz)(\lambda u.u(\lambda v.v(yx)))&
\end{flalign*}
\end{minipage}
\caption{Diagram of a NPT with outer neutral handles indicated and numbered. On the right, we show the associated factorizations of the NPT into a neutral subterm and its surrounding context, as well as the corresponding value-open NPTs which result from inverting the transformation (\ref{eqn:VO}).}
\label{fig:valences}
\end{figure}%

With that by way of geometric intuition, we can now give a formal specification of the outer neutral handles of a NPT, defining these by induction for any neutral or normal planar term with any number of free variables.
\begin{definition}\label{defn:valence}\emph{
Let $[\Gamma]t$ be a neutral or normal planar lambda term.
A \definand{neutral handle} of $[\Gamma]t$ is a factorization of $t$ into a neutral subterm and its surrounding context.
The set $\ONH{\pi}$ of \definand{outer neutral handles} is defined by induction on the coloring $\pi$ of $[\Gamma]t$, as follows: 
\begin{description}
\item[{Case $\pi^d = \imgcenter{\includegraphics[scale=0.7]{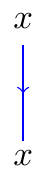}}$}] Then $\ONH{\pi} = \set{\focus{x}}$ .
\item[{Case $\pi^d = \imgcenter{\includegraphics[scale=0.6]{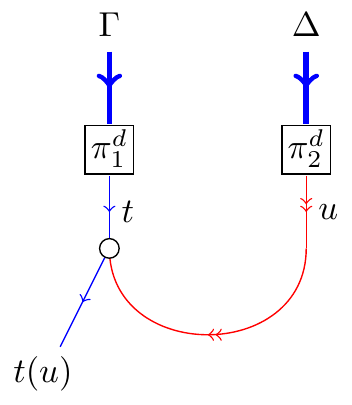}}$}]\ \\
Then $\ONH{\pi} = \begin{cases}\set{\focus{t(u)}} \cup \set{t(H) \mid H \in \ONH{\pi_2}} \cup \set{H(u) \mid H \in \ONH{\pi_1}}  & \text{if }|\D| = 0 \\ \set{\focus{t(u)}} \cup \set{t(H) \mid H \in \ONH{\pi_2}} & \text{if }|\D| > 0\end{cases}$ .
\item[{Case $\pi^d = \imgcenter{\includegraphics[scale=0.6]{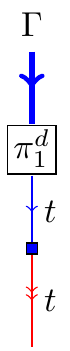}}$}] Then $\ONH{\pi} = \set{H \mid H \in \ONH{\pi_1}}$ .
\item[{Case $\pi^d = \imgcenter{\includegraphics[scale=0.6]{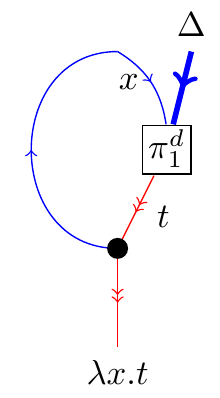}}$}] Then $\ONH{\pi} = \set{\lambda x.H \mid H \in \ONH{\pi_1}}$ .
\end{description}}
\end{definition}

\subsection{The size-preserving bijection}
\label{sec:decomp:bij}

Let us write $\composFunOpen(t_1,t_2)$ for the binary operation taking a pair of NPTs $t_1$ and $t_2$ and joining them together to form a function-open NPT by the procedure described in \Cref{sec:decomp:fn}.
Similarly, we write $\composValOpen{k}(t_1)$ for the operation taking a NPT $t_1$ with $\ge k$ outer neutral handles and factoring it along the $k$th to produce a value-open NPT by the procedure described in \Cref{sec:decomp:arg}.
We now establish a lambda calculus analogue of Theorem \ref{thm:tutte}:

\begin{thm}
\label{thm:tuttelam}
Let $[x]t$ be a NPT with $R$-coloring $\pi$.
Then exactly one of the following cases must hold:
\begin{enumerate}[label=(\roman*)]
\item $[x]t$ is the identity term and $|\pi| = |\ONH{\pi}| = 1$.
\item $[x]t = \composFunOpen([x_1]t_1,[x_2]t_2)$ for some $[x_1]t_1$ and $[x_2]t_2$ (with $R$-colorings $\pi_1$ and $\pi_2$) such that $|\pi| = |\pi_1| + |\pi_2|$ and $|\ONH{\pi}| = 1 + |\ONH{\pi_1}| + |\ONH{\pi_2}|$.
\item $[x]t = \composValOpen{k}([x_1]t_1)$ for some $t_1$ (with $R$-coloring $\pi_1$) and $1 \le k \le |\ONH{\pi_1}|$ such that $|\pi| = 1 + |\pi_1|$ and $|\ONH{\pi}| = k+1$.
\end{enumerate}
\end{thm}
\begin{proof}
Following \Cref{prop:trichotomy} and the discussions in \Cref{sec:decomp:fn,sec:decomp:arg}, what is left to verify is that the operations $\composFunOpen$ and $\composValOpen{k}$ have the right effect on the numbers of $s$-nodes and outer neutral handles.
Consider the $\composFunOpen$ operation yielding a function-open NPT, as expressed in string diagrams:
$$
\imgcenter{\includegraphics{\DIAGRAMS/g.pdf}}
\qquad\mkfun\qquad
\imgcenter{\includegraphics{\DIAGRAMS/h.pdf}}
\qquad=\qquad
\imgcenter{\includegraphics{\DIAGRAMS/isthmic.pdf}}
$$
By inspection, the resulting diagram includes all and only the $s$-nodes coming from the two input diagrams, and all of the outer neutral handles plus one additional one leading into the $a$-node.
Likewise, consider the $\composValOpen{k}$ operation yielding a value-open NPT:
$$
\imgcenter{\includegraphics{\DIAGRAMS/h.pdf}}
\quad=\quad
\imgcenter{\includegraphics{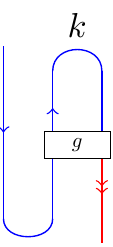}}
\qquad\Rightarrow\qquad
\imgcenter{\includegraphics[height=5em]{\DIAGRAMS/urest2.pdf}}(k)\quad
\imgcenter{\includegraphics{\DIAGRAMS/h.pdf}}
\quad=\quad
\imgcenter{\includegraphics{\DIAGRAMS/non-isthmic.pdf}}
$$
By inspection, the resulting diagram has one additional $s$-node, and exactly $k+1$ outer neutral handles, corresponding to the first $k$ outer neutral handles of the input diagram plus one additional one leading into the $s$-node.
\end{proof}

\begin{thm}\label{thm:main-bijection}
There is a one-to-one correspondence between 
rooted planar maps with $n$ edges and outer face degree $d$, and NPTs with $n+1$ $s$-nodes and $d+1$ outer neutral handles.
\end{thm}
\begin{proof}
By playing Theorems \ref{thm:tutte} nd \ref{thm:tuttelam} in parallel, to decompose a rooted planar map/NPT and then recompose it as the corresponding NPT/rooted planar map.
(This uses an induction on the number of edges of a rooted planar map, and on the size of an NPT.)
Note that the off-by-one offset between number of edges/outer face degree and size/number of outer neutral handles means that we need to apply the arithmetic identities $(n+1)+(n'+1) = 1+(n+n'+1)$ and $1+(d+1)+(d'+1) = (2+d+d')+1$ in the isthmic/function-open case.
\end{proof}
\begin{figure}
\begin{center}
\begin{subfigure}[b]{0.33\textwidth}
\qquad\qquad
\imgcenter{\includegraphics[scale=0.6]{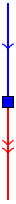}}
\quad
\imgcenter{\includegraphics[scale=1.5]{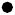}}
\end{subfigure}
\begin{subfigure}[b]{0.33\textwidth}
\qquad
\imgcenter{\includegraphics[scale=0.6]{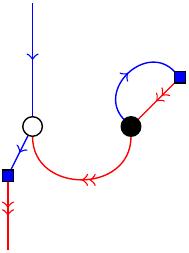}}
\quad
\imgcenter{\includegraphics[scale=1.5]{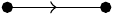}}
\end{subfigure}
\begin{subfigure}[b]{0.32\textwidth}
\qquad
\imgcenter{\includegraphics[scale=0.6]{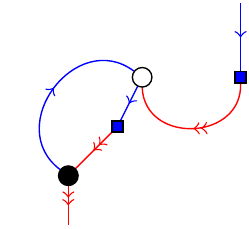}}
\imgcenter{\includegraphics[scale=1.5]{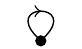}}
\end{subfigure}

\begin{subfigure}[b]{0.33\textwidth}
\imgcenter{\includegraphics[scale=0.6]{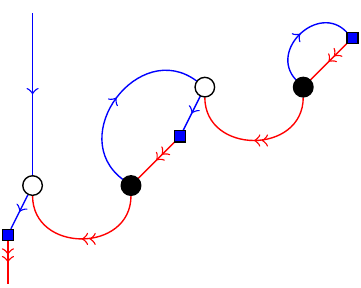}}
\quad\imgcenter{\includegraphics[scale=1.5]{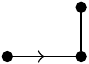}}
\end{subfigure}
\begin{subfigure}[b]{0.33\textwidth}
\qquad
\imgcenter{\includegraphics[scale=0.6]{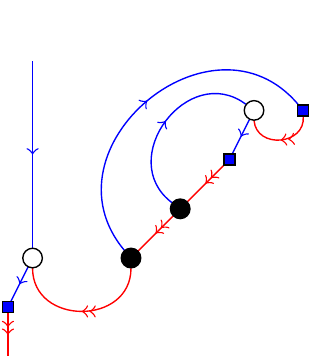}}
\imgcenter{\includegraphics[scale=1.5]{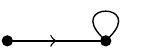}}
\end{subfigure}
\begin{subfigure}[b]{0.32\textwidth}
\quad\imgcenter{\includegraphics[scale=0.6]{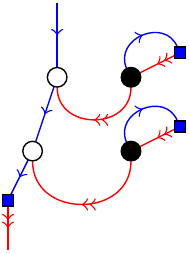}}
\quad
\imgcenter{\includegraphics[scale=1.5]{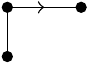}}
\end{subfigure}

\begin{subfigure}[b]{0.33\textwidth}
\imgcenter{\includegraphics[scale=0.6]{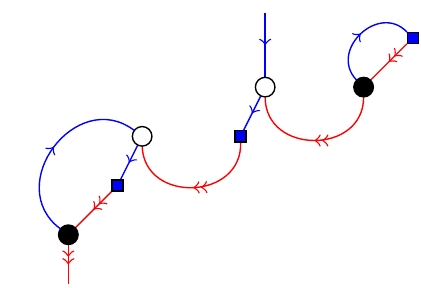}}
\imgcenter{\includegraphics[scale=1.5]{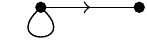}}
\end{subfigure}
\begin{subfigure}[b]{0.33\textwidth}
\quad\imgcenter{\includegraphics[scale=0.6]{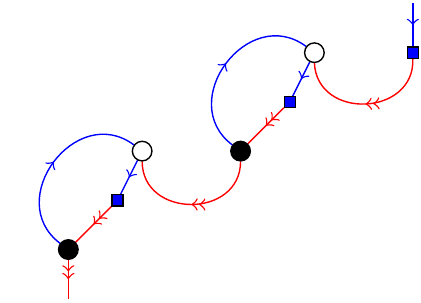}}
\ \ 
\imgcenter{\includegraphics[scale=1.5]{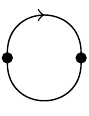}}
\end{subfigure}
\begin{subfigure}[b]{0.32\textwidth}
\imgcenter{\includegraphics[scale=0.6]{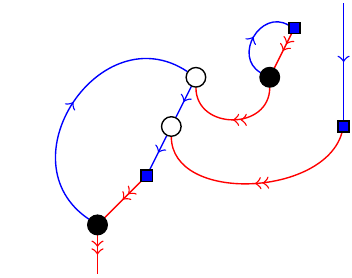}}
\quad
\imgcenter{\includegraphics[scale=1.5]{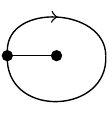}}
\end{subfigure}

\begin{subfigure}[b]{0.32\textwidth}
\imgcenter{\includegraphics[scale=0.6]{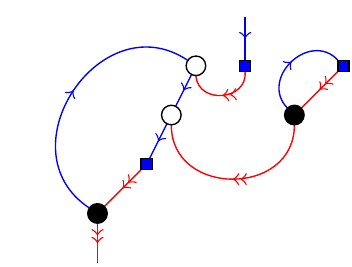}}
\imgcenter{\includegraphics[scale=1.5]{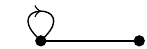}}
\end{subfigure}
\begin{subfigure}[b]{0.32\textwidth}
\imgcenter{\includegraphics[scale=0.6]{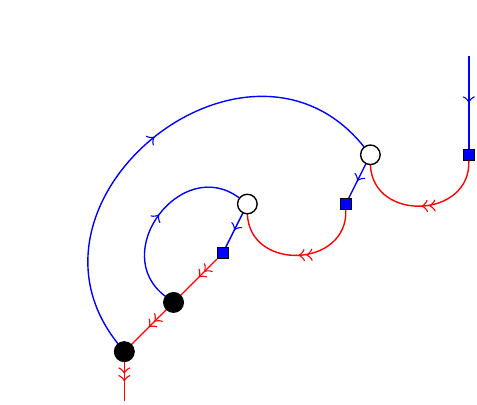}}
\imgcenter{\includegraphics[scale=1.5]{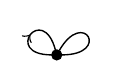}}
\end{subfigure}
\begin{subfigure}[b]{0.32\textwidth}
\imgcenter{\includegraphics[scale=0.6]{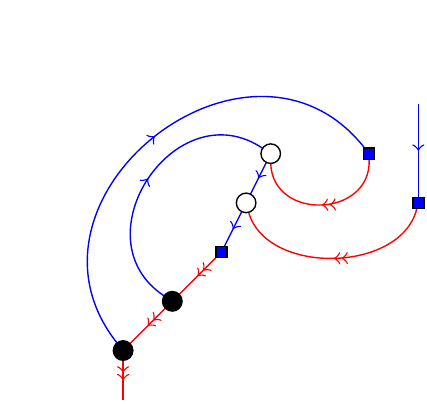}}
\imgcenter{\includegraphics[scale=1.5]{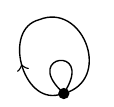}}
\end{subfigure}
\end{center}
\caption{Correspondence between the NPTs of size $\le 3$ and rooted planar maps with $\le 2$ edges.}
\label{fig:lam3map2}
\end{figure}
\begin{figure}
$$
\infer[\composValOpen{2}/\composNonIsthm{1}]{\includegraphics[scale=0.8]{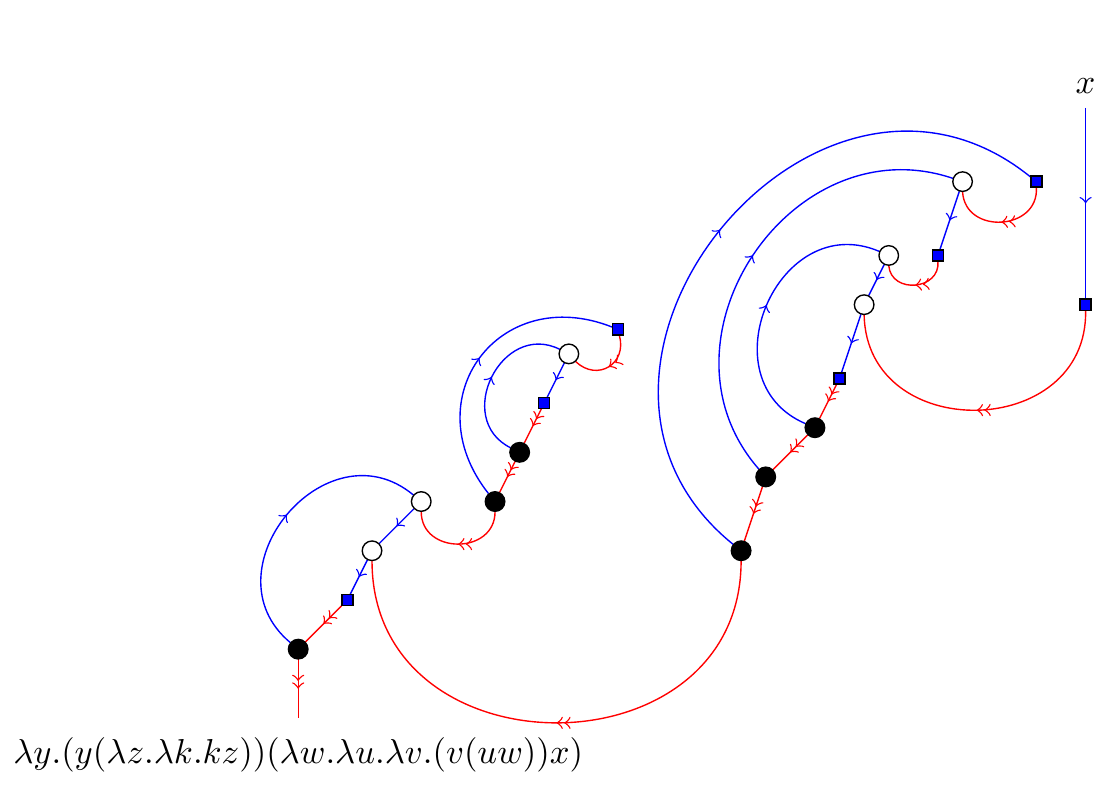} \quad\raisebox{2em}{\includegraphics[scale=0.4]{\DIAGRAMS/animation/g6.pdf}} %
}{
\infer[\composFunOpen/\composIsthm]{\includegraphics[scale=0.7]{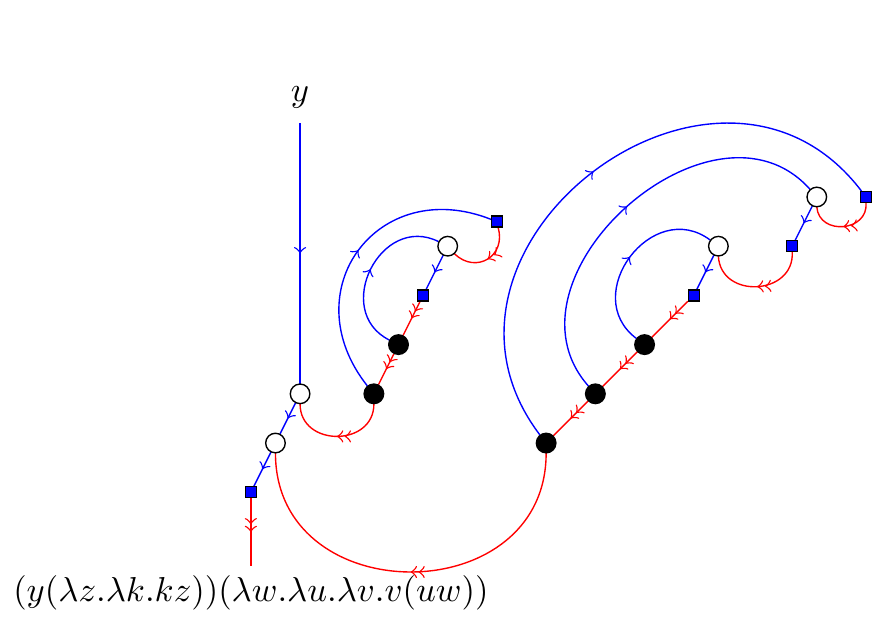} \quad\raisebox{2em}{\includegraphics[scale=0.35]{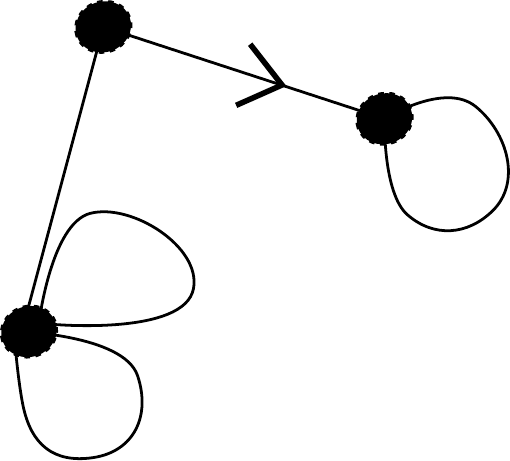}} %
}{
 \infer[\composFunOpen/\composIsthm]{\includegraphics[scale=0.6]{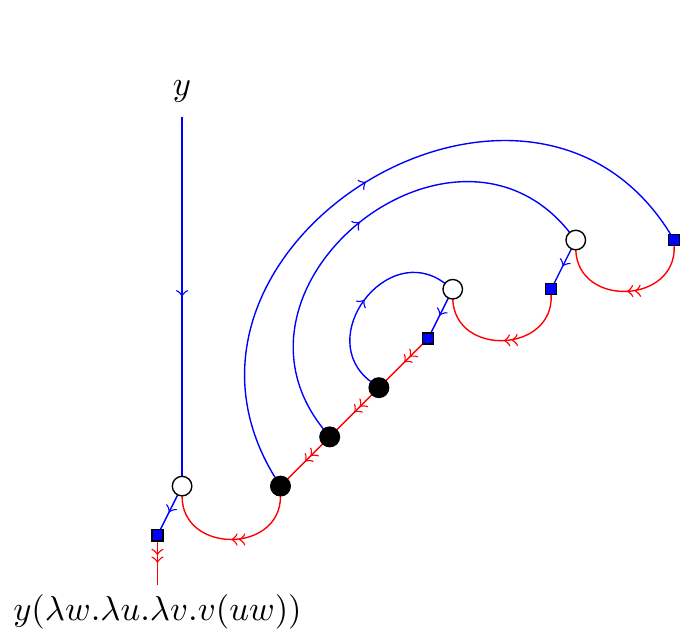} \quad\raisebox{1.5em}{\includegraphics[scale=0.3]{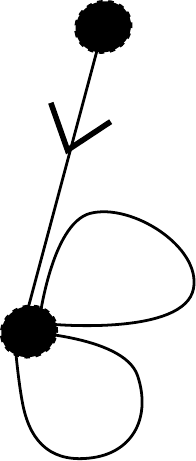}}%
}{
  \includegraphics[scale=0.5]{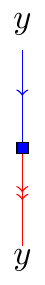}\quad \raisebox{1.5em}{\includegraphics[scale=0.25]{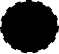}}%
 &\qquad&
  \infer[\composValOpen{2}/\composNonIsthm{1}]{\includegraphics[scale=0.5]{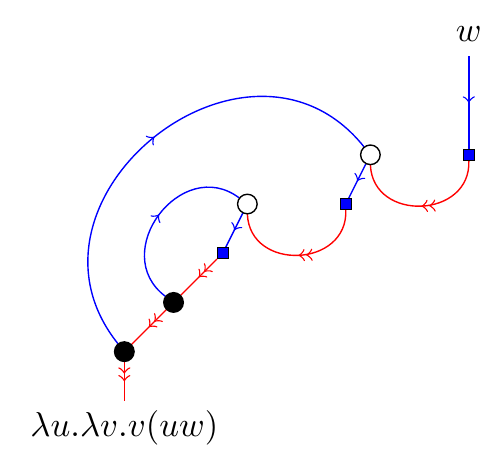} \quad\raisebox{1em}{\includegraphics[scale=0.25]{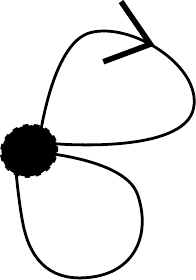}}%
}{
  \infer[\composValOpen{1}/\composNonIsthm{0}]{\includegraphics[scale=0.4]{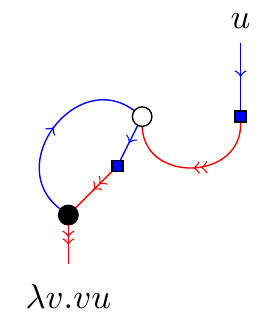} \quad\raisebox{1em}{\includegraphics[scale=0.2]{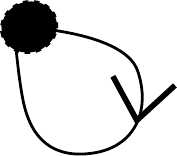}}}{
  \includegraphics[scale=0.3]{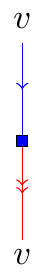}\quad \raisebox{1em}{\includegraphics[scale=0.15]{\DIAGRAMS/animation/g0.pdf}}
 }}}
 &\qquad&
 \infer[\composValOpen{1}/\composNonIsthm{0}]{
  \includegraphics[scale=0.6]{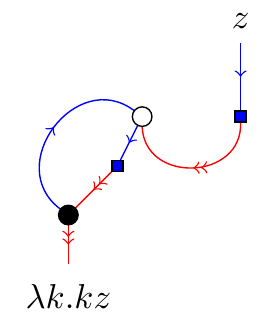} \quad\raisebox{1em}{\includegraphics[scale=0.3]{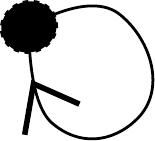}}
}{\includegraphics[scale=0.5]{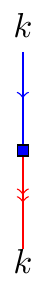} \quad\raisebox{2em}{\includegraphics[scale=0.25]{\DIAGRAMS/animation/g0.pdf}}}}}
$$

\caption{Full decomposition of a normal planar lambda term with seven $s$-nodes and three outer neutral handles, in parallel with the corresponding rooted planar map with six edges and outer face degree two.}
\label{fig:fulltutte}
\end{figure}
In \Cref{fig:lam3map2}, we show the result of applying this bijection to all NPTs of size at most three, while in \Cref{fig:fulltutte}, we give an illustration of the proof of Theorem \ref{thm:main-bijection} in action, animating the full decomposition of a particular NPT (with 7 $s$-nodes and 3 outer neutral handles) in parallel with the decomposition of the corresponding rooted planar map (\# edges = 6, outer face degree = 2).

\begin{cor}\label{cor:1to1}
The following families of objects are all in size-preserving bijection:
\begin{itemize}
\item rooted planar maps
\item normal planar lambda terms
\item $R$-colorings of lambda skeletons
\end{itemize}
\end{cor}

\section*{Acknowledgments}

We thank Alexis Saurin, Paul-André Melliès, Maciej Do{\l}\c{e}ga, and Beniamino Accattoli for discussions and pointers to related work.
We also thank Pierre Lescanne for an invitation to speak about this work at the 8th Workshop on Computational Logic and Applications at Lyon in March 2015.
The string diagrams in the paper were manipulated with the help of the TikZiT diagram editor.

\bibliographystyle{abbrvnat}

\pagebreak

\appendix

\section{All normal planar lambda terms of size four}
\label{sec:size4}

\noindent Unannotated string diagrams:

\begin{center}
1.~\includegraphics[scale=0.5]{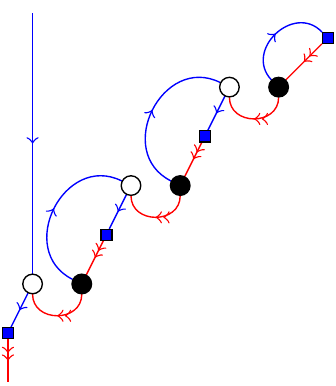}\quad
2.~\includegraphics[scale=0.5]{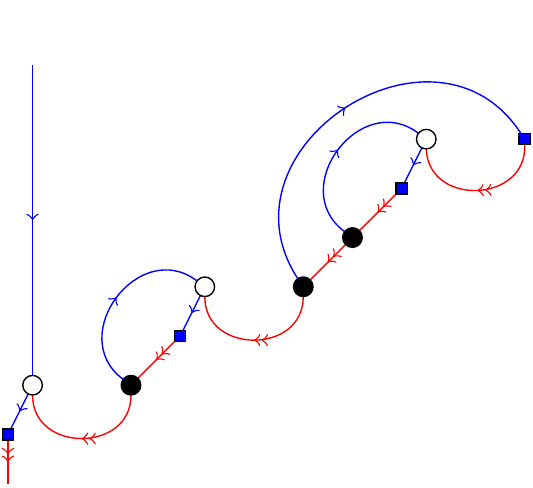}\quad
3.~\includegraphics[scale=0.5]{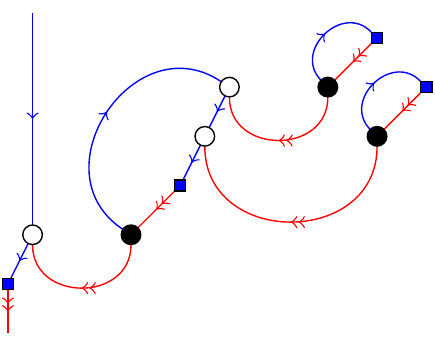}\quad
4.~\includegraphics[scale=0.5]{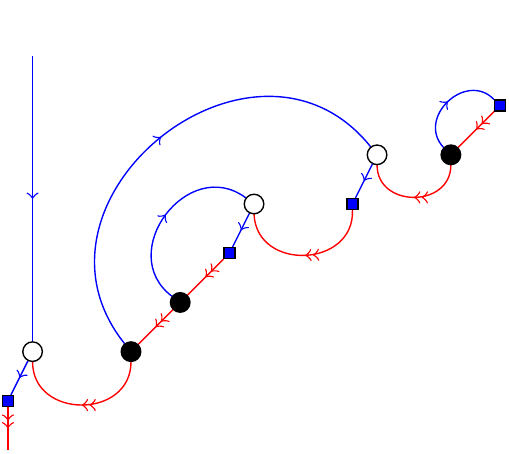}\quad
5.~\includegraphics[scale=0.5]{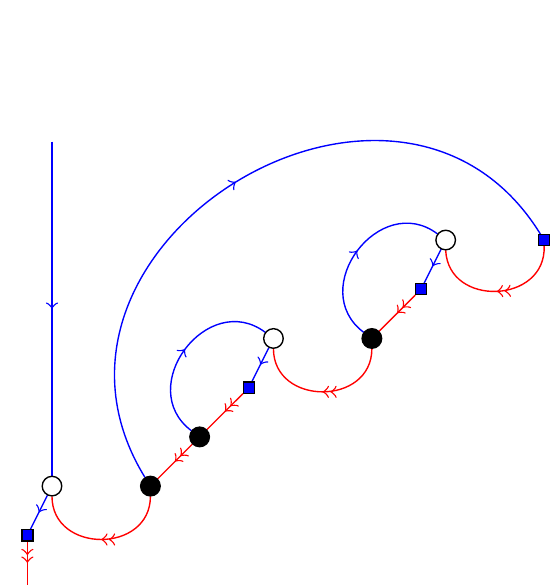}\quad
\hspace*{-1em}6.~\includegraphics[scale=0.5]{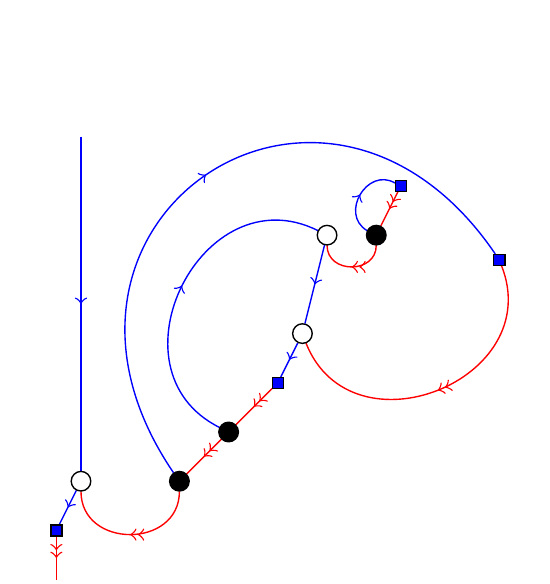}\quad
7.~\includegraphics[scale=0.5]{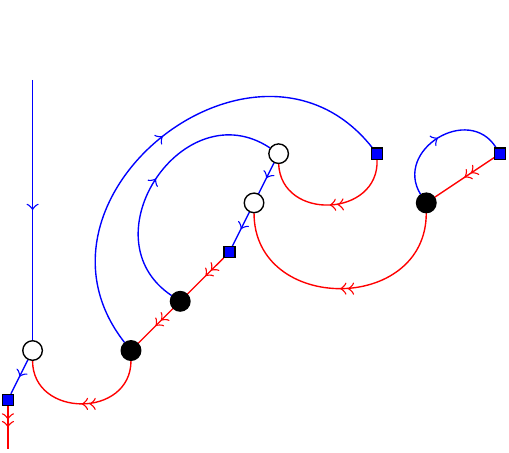}\quad
8.~\includegraphics[scale=0.5]{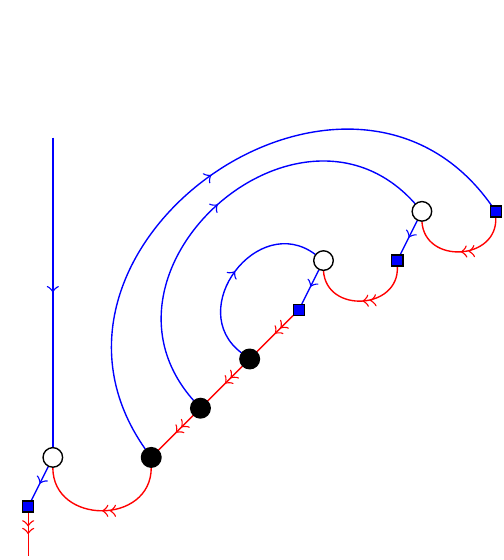}\quad 
9.~\includegraphics[scale=0.5]{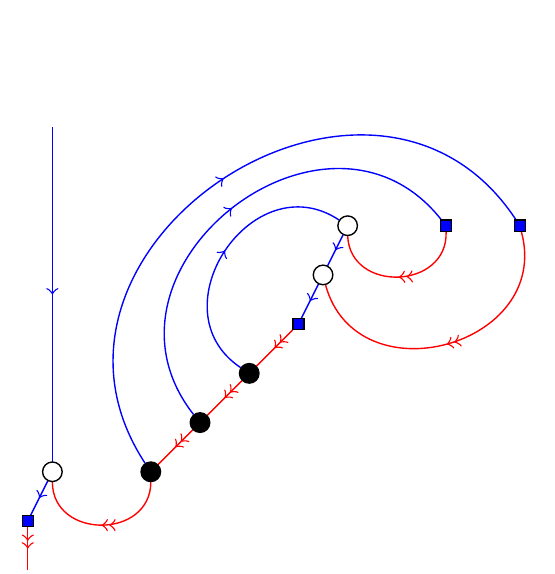}\quad
10.~\includegraphics[scale=0.5]{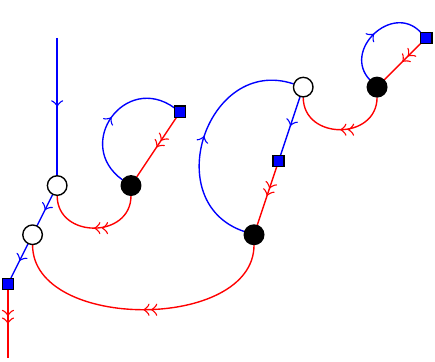}\quad
11.~\includegraphics[scale=0.5]{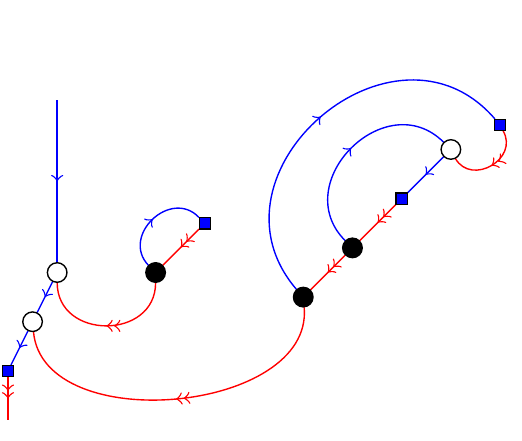}\quad
12.~\includegraphics[scale=0.5]{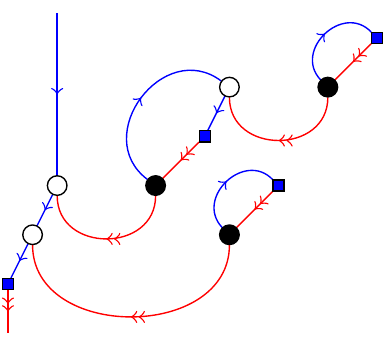}\quad
13.~\includegraphics[scale=0.5]{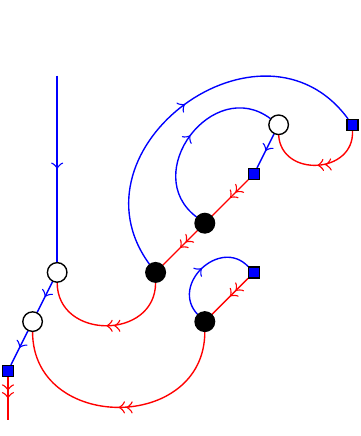}\quad
14.~\includegraphics[scale=0.5]{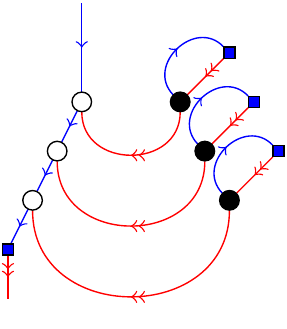}\quad
15.~\includegraphics[scale=0.5]{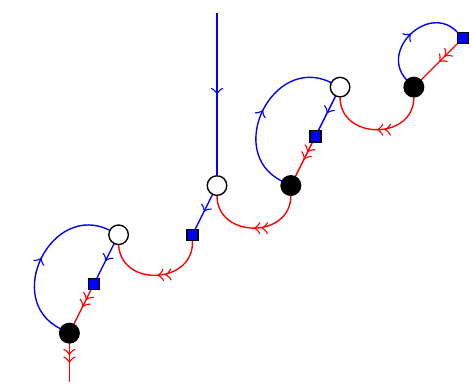}\quad
\hspace*{-1em}16.~\includegraphics[scale=0.5]{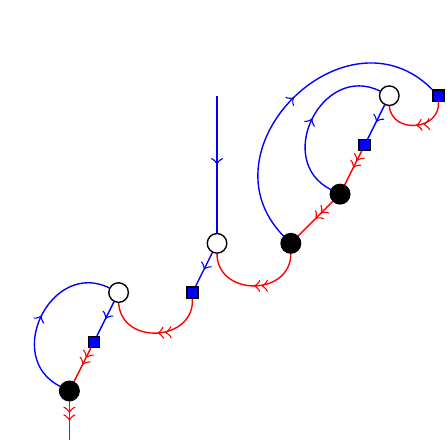}\quad
\hspace*{-1em}17.~\includegraphics[scale=0.5]{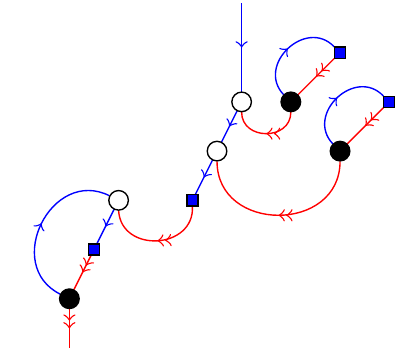}\\
18.~\includegraphics[scale=0.5]{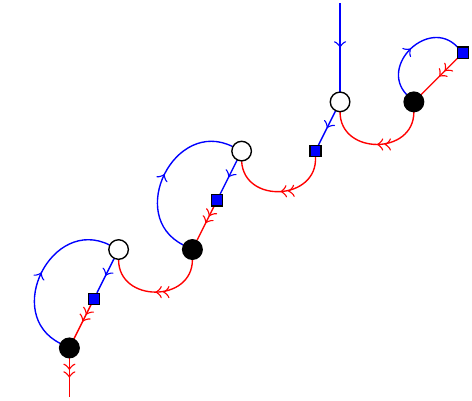}\quad
\hspace*{-1em}19.~\includegraphics[scale=0.5]{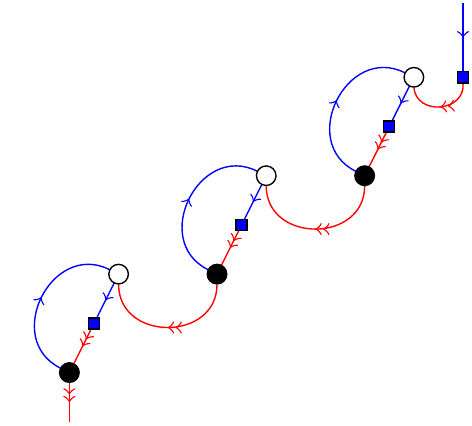}\quad
\hspace*{-1em}20.~\includegraphics[scale=0.5]{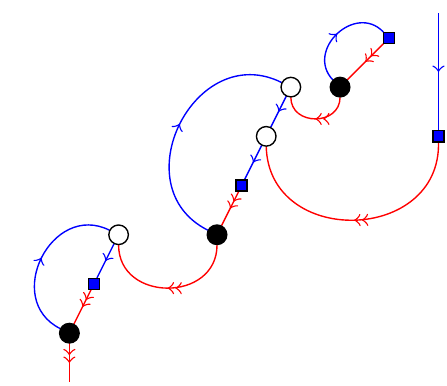}\quad
\hspace*{-1em}21.~\includegraphics[scale=0.5]{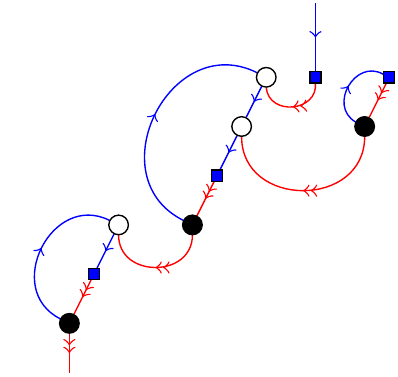}\quad
\hspace*{-1em}22.~\includegraphics[scale=0.5]{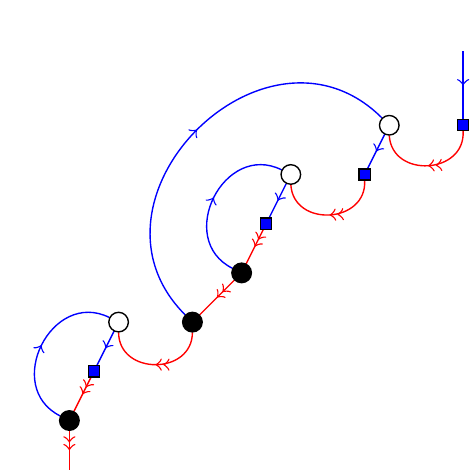}\\
\hspace*{-1em}23.~\includegraphics[scale=0.5]{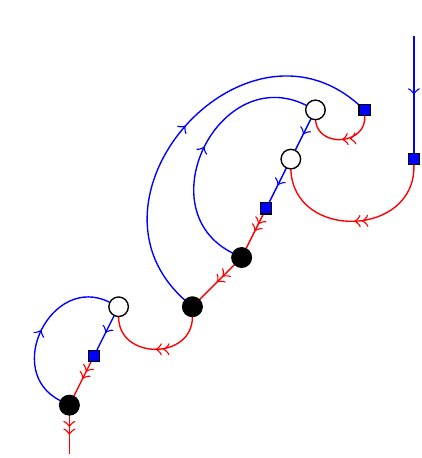}\quad
24.~\hspace*{-0.5em}\includegraphics[scale=0.5]{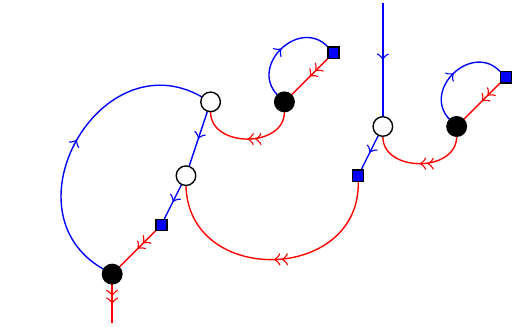}\quad
25.~\hspace*{-0.5em}\includegraphics[scale=0.5]{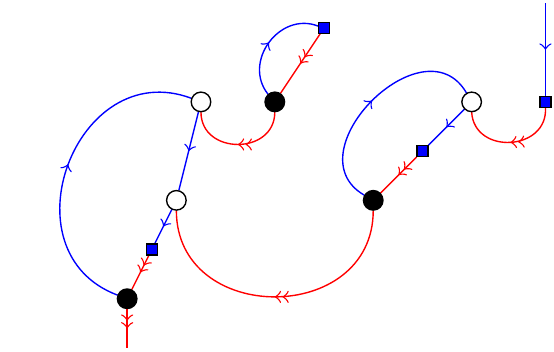}\quad
26.~\hspace*{-0.5em}\includegraphics[scale=0.5]{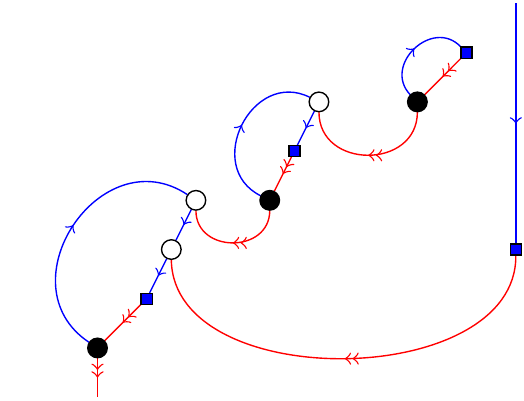}\quad
27.~\includegraphics[scale=0.5]{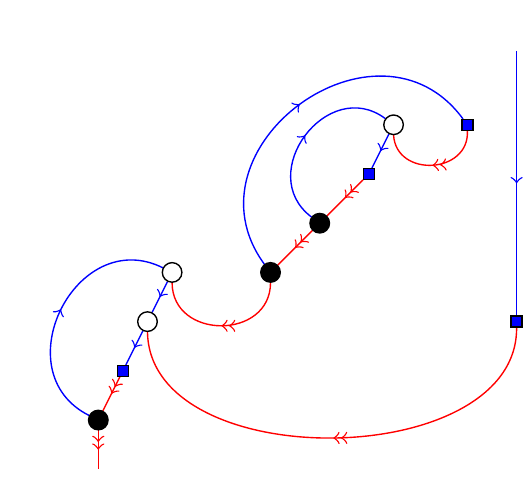}\quad
28.~\hspace*{-0.5em}\includegraphics[scale=0.5]{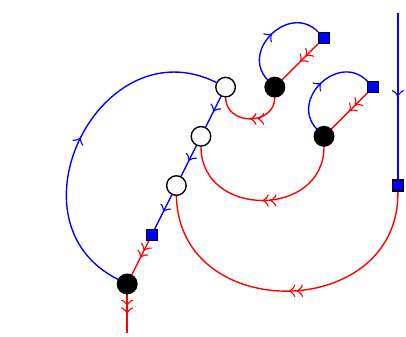}\quad
29.~\hspace*{-0.5em}\includegraphics[scale=0.5]{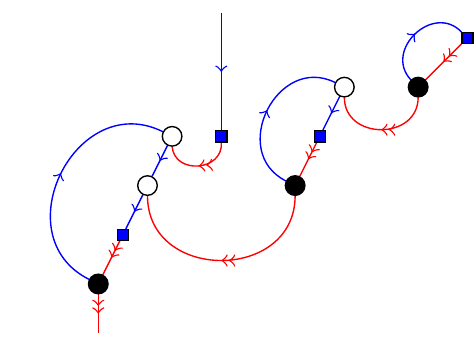}\quad
30.~\hspace*{-0.5em}\includegraphics[scale=0.5]{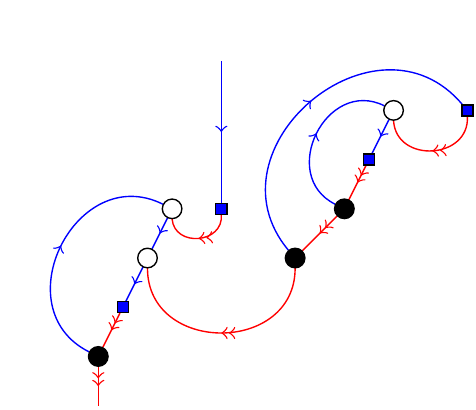}\quad
31.~\hspace*{-0.5em}\includegraphics[scale=0.5]{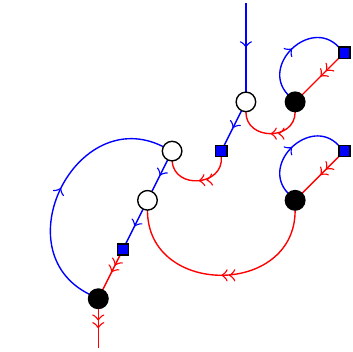}\quad
32.~\hspace*{-0.5em}\includegraphics[scale=0.5]{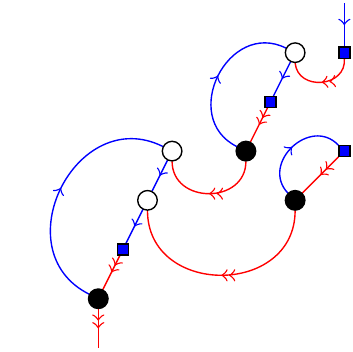}\quad
33.\hspace*{-0.5em}~\includegraphics[scale=0.5]{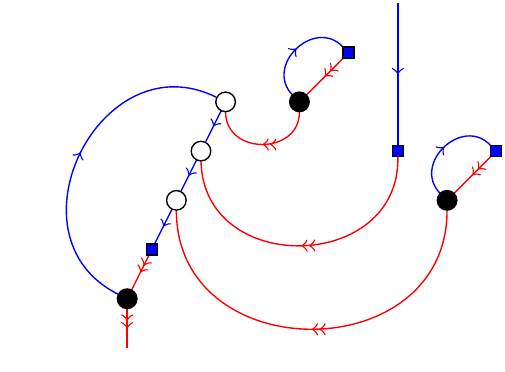}\quad
34.\hspace*{-0.5em}~\includegraphics[scale=0.5]{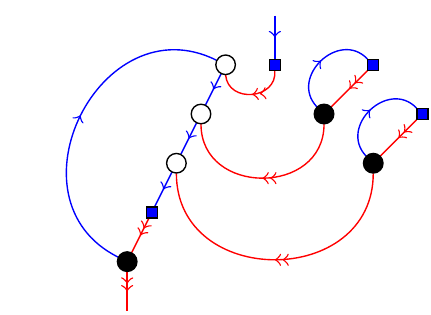}\quad
35.\hspace*{-0.5em}~\includegraphics[scale=0.5]{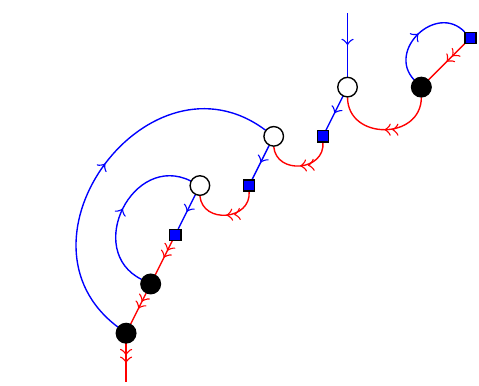}\quad
36.\hspace*{-0.5em}~\includegraphics[scale=0.5]{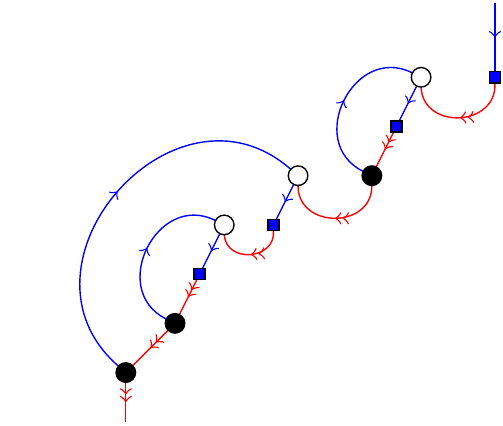}\quad
37.\hspace*{-1em}~\includegraphics[scale=0.5]{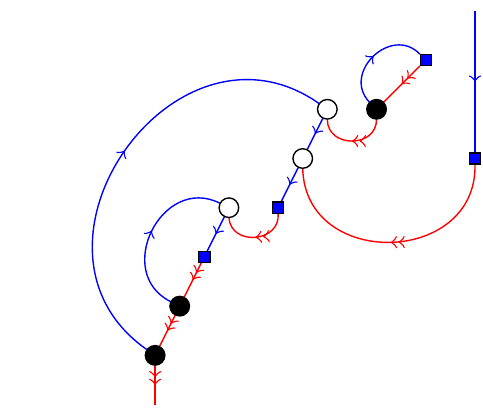}\quad
38.\hspace*{-1em}~\includegraphics[scale=0.5]{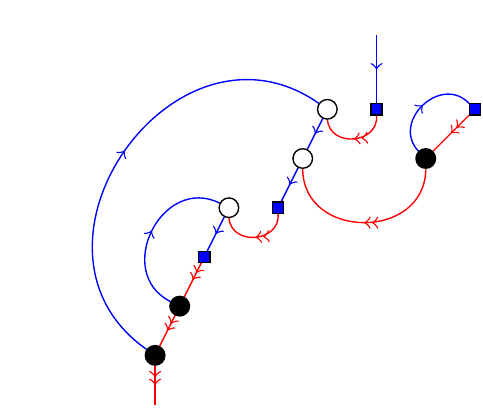}\quad
39.\hspace*{-1.5em}~\includegraphics[scale=0.5]{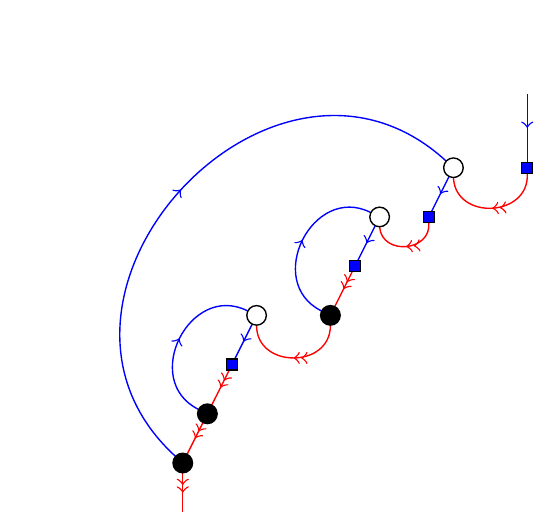}\quad
40.\hspace*{-1.5em}~\includegraphics[scale=0.5]{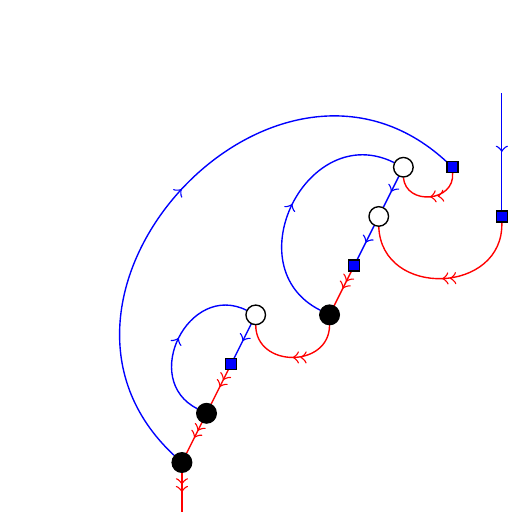}\quad
41.\hspace*{-1.5em}~\includegraphics[scale=0.5]{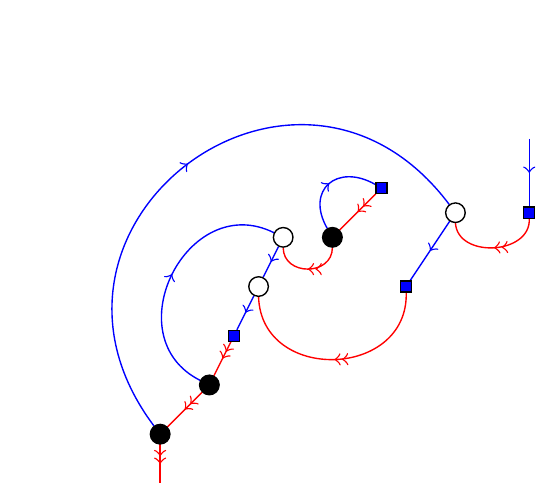}\quad
42.\hspace*{-1em}~\includegraphics[scale=0.5]{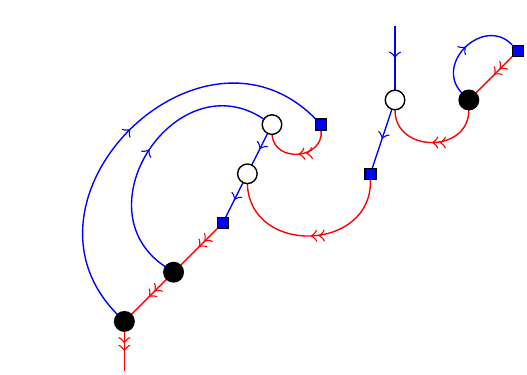}\quad
43.\hspace*{-1em}~\includegraphics[scale=0.5]{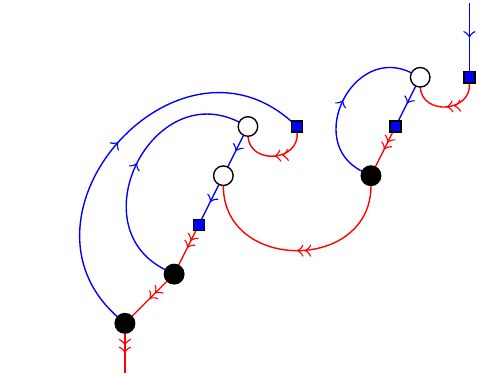}\quad
44.\hspace*{-1em}~\includegraphics[scale=0.5]{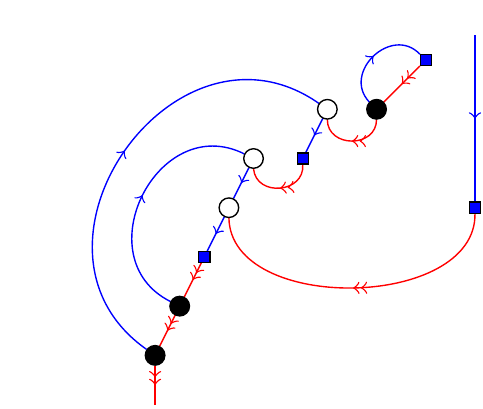}\quad
45.\hspace*{-1.5em}~\includegraphics[scale=0.5]{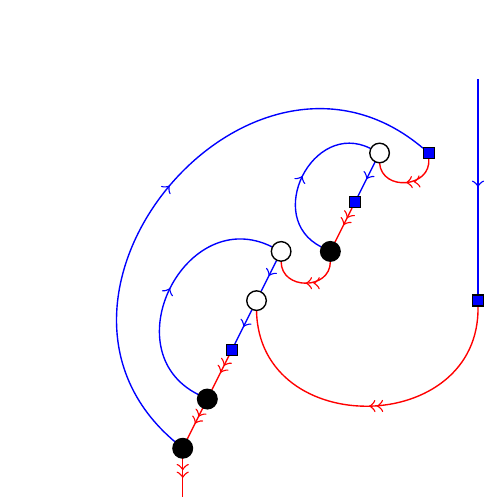}\quad
46.\hspace*{-1.5em}~\includegraphics[scale=0.5]{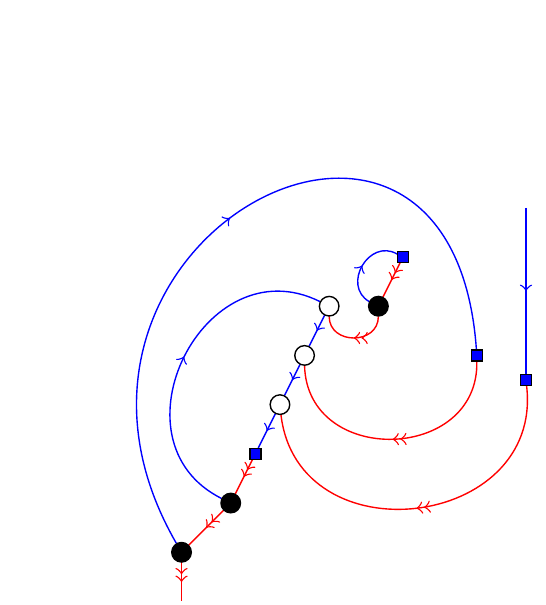}\quad
47.\hspace*{-1.5em}~\includegraphics[scale=0.5]{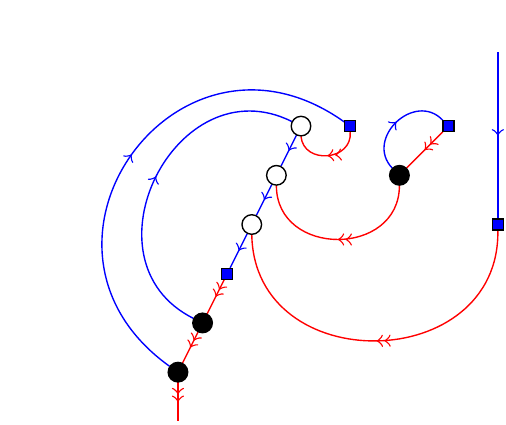}\quad
48.\hspace*{-1.5em}~\includegraphics[scale=0.5]{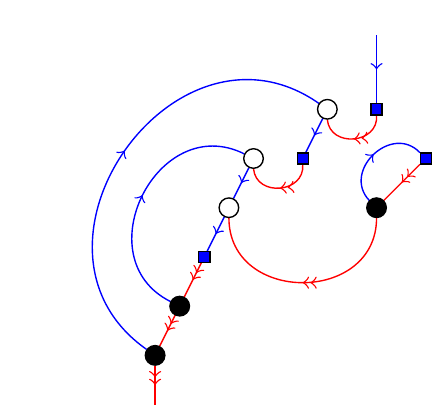}\quad
49.\hspace*{-1.5em}~\includegraphics[scale=0.5]{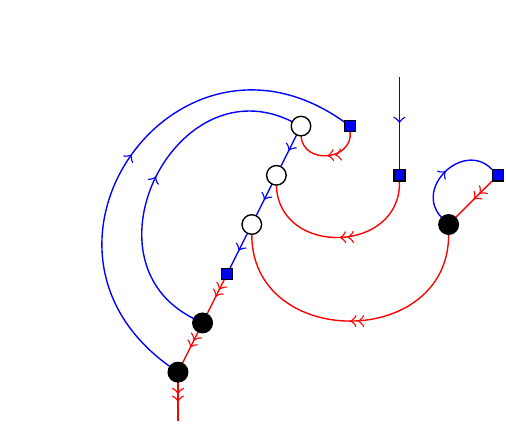}\quad
50.\hspace*{-1.5em}~\includegraphics[scale=0.5]{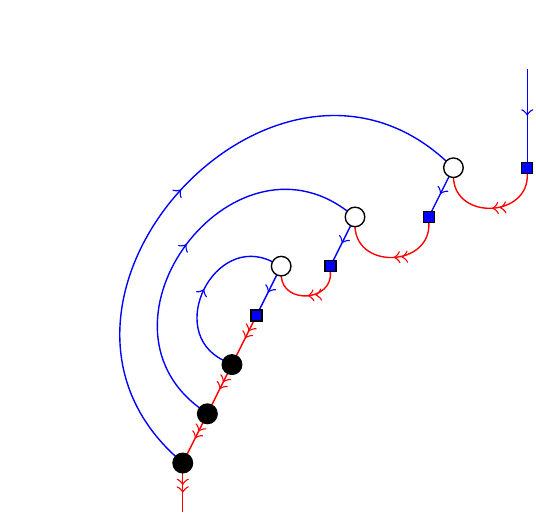}\quad
51.\hspace*{-1.5em}~\includegraphics[scale=0.5]{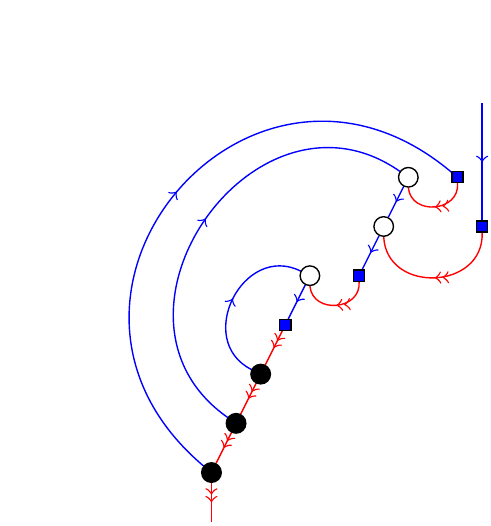}\quad
52.\hspace*{-1.5em}~\includegraphics[scale=0.5]{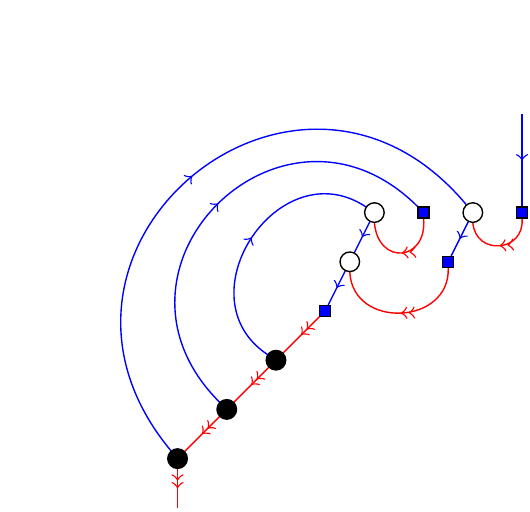}\quad
53.\hspace*{-1.5em}~\includegraphics[scale=0.5]{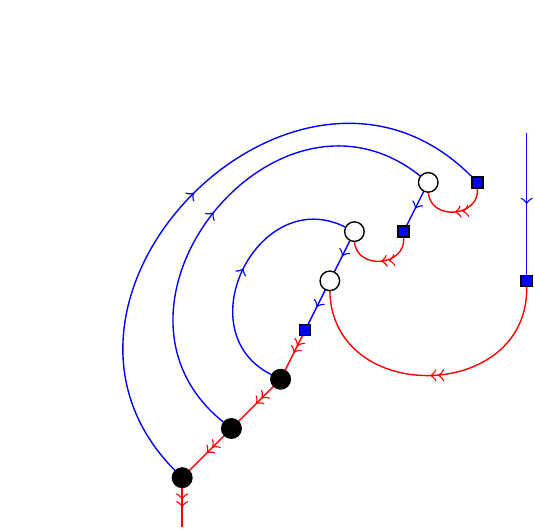}\quad
54.\hspace*{-3em}~\includegraphics[scale=0.5]{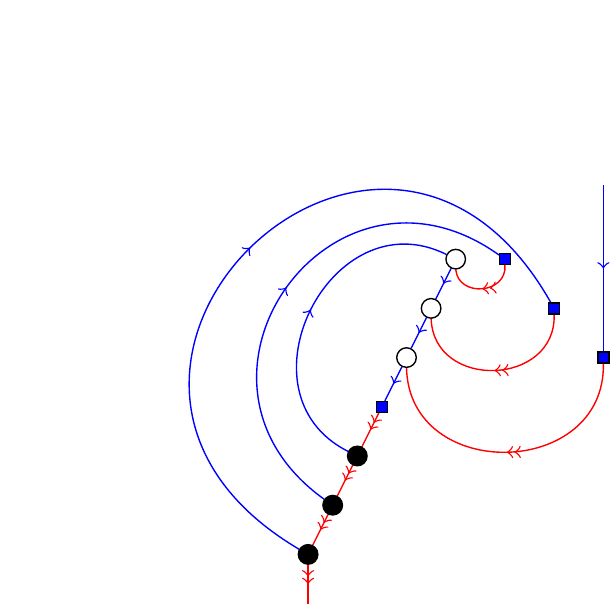}
\end{center}\bigskip


\noindent Corresponding normal planar lambda terms (with one free variable $x$):

\begin{multicols}{3}
\begin{enumerate}
\item $x(\lambda y.y(\lambda z.z(\lambda w.w)))$
\item $x(\lambda y.y(\lambda z.\lambda w.wz))$
\item $x(\lambda y.(y(\lambda z.z))(\lambda w.w))$
\item $x(\lambda y.\lambda z.z(y(\lambda w.w)))$
\item $x(\lambda y.\lambda z.z(\lambda w.wy))$
\item $x(\lambda y.\lambda z.(z(\lambda w.w))y)$
\item $x(\lambda y.\lambda z.(zy)(\lambda w.w))$
\item $x(\lambda y.\lambda z.\lambda w.w(zy))$
\item $x(\lambda y.\lambda z.\lambda w.(wz)y)$
\item $(x(\lambda y.y))(\lambda z.z(\lambda w.w))$
\item $(x(\lambda y.y))(\lambda z.\lambda w.wz)$
\item $(x(\lambda y.y(\lambda z.z)))(\lambda w.w)$
\item $(x(\lambda y.\lambda z.zy))(\lambda w.w)$
\item $((x(\lambda y.y))(\lambda z.z))(\lambda w.w)$
\item $\lambda y.y(x(\lambda z.z(\lambda w.w)))$
\item $\lambda y.y(x(\lambda z.\lambda w.wz))$
\item $\lambda y.y((x(\lambda z.z))(\lambda w.w))$
\item $\lambda y.y(\lambda z.z(x(\lambda w.w)))$
\item $\lambda y.y(\lambda z.z(\lambda w.wx))$
\item $\lambda y.y(\lambda z.(z(\lambda w.w))x)$
\item $\lambda y.y(\lambda z.(zx)(\lambda w.w))$
\item $\lambda y.y(\lambda z.\lambda w.w(zx))$
\item $\lambda y.y(\lambda z.\lambda w.(wz)x)$
\item $\lambda y.(y(\lambda z.z))(x(\lambda w.w))$
\item $\lambda y.(y(\lambda z.z))(\lambda w.wx)$
\item $\lambda y.(y(\lambda z.z(\lambda w.w)))x$
\item $\lambda y.(y(\lambda z.\lambda w.wz))x$
\item $\lambda y.((y(\lambda z.z))(\lambda w.w))x$
\item $\lambda y.(yx)(\lambda z.z(\lambda w.w))$
\item $\lambda y.(yx)(\lambda z.\lambda w.wz)$
\item $\lambda y.(y(x(\lambda z.z)))(\lambda w.w)$
\item $\lambda y.(y(\lambda z.zx))(\lambda w.w)$
\item $\lambda y.((y(\lambda z.z))x)(\lambda w.w)$
\item $\lambda y.((yx)(\lambda z.z))(\lambda w.w)$
\item $\lambda y.\lambda z.z(y(x(\lambda w.w)))$
\item $\lambda y.\lambda z.z(y(\lambda w.wx))$
\item $\lambda y.\lambda z.z((y(\lambda w.w))x)$
\item $\lambda y.\lambda z.z((yx)(\lambda w.w))$
\item $\lambda y.\lambda z.z(\lambda w.w(yx))$
\item $\lambda y.\lambda z.z(\lambda w.(wy)x)$
\item $\lambda y.\lambda z.(z(\lambda w.w))(yx)$
\item $\lambda y.\lambda z.(zy)(x(\lambda w.w))$
\item $\lambda y.\lambda z.(zy)(\lambda w.wx)$
\item $\lambda y.\lambda z.(z(y(\lambda w.w)))x$
\item $\lambda y.\lambda z.(z(\lambda w.wy))x$
\item $\lambda y.\lambda z.(z(\lambda w.w)y)x$
\item $\lambda y.\lambda z.((zy)(\lambda w.w))x$
\item $\lambda y.\lambda z.(z(yx))(\lambda w.w)$
\item $\lambda y.\lambda z.((zy)x)(\lambda w.w)$
\item $\lambda y.\lambda z.\lambda w.w(z(yx))$
\item $\lambda y.\lambda z.\lambda w.w((zy)x)$
\item $\lambda y.\lambda z.\lambda w.(wz)(yx)$
\item $\lambda y.\lambda z.\lambda w.(w(zy))x$
\item $\lambda y.\lambda z.\lambda w.((wz)y)x$
\end{enumerate}
\end{multicols}

\end{document}